\documentclass[11pt,reqno]{amsart}
\usepackage[latin9]{inputenc}
\usepackage{geometry}
\usepackage{units}
\usepackage{mathrsfs}
\usepackage{enumitem}
\usepackage{amstext}
\usepackage{amsthm}
\usepackage{amssymb}
\usepackage{graphicx}
\PassOptionsToPackage{normalem}{ulem}
\usepackage{ulem}

\makeatletter
\numberwithin{equation}{section}
\numberwithin{figure}{section}
\theoremstyle{plain}
\newtheorem{thm}{\protect\theoremname}[section]
  \theoremstyle{definition}
  \newtheorem{defn}[thm]{\protect\definitionname}
  \theoremstyle{definition}
  \newtheorem{example}[thm]{\protect\examplename}
  \theoremstyle{plain}
  \newtheorem{cor}[thm]{\protect\corollaryname}
  \theoremstyle{remark}
  \newtheorem*{rem*}{\protect\remarkname}
  \theoremstyle{plain}
  \newtheorem{lem}[thm]{\protect\lemmaname}
  \theoremstyle{plain}
  \newtheorem{prop}[thm]{\protect\propositionname}

\usepackage{amsthm}\usepackage{bm}
\usepackage{mathrsfs}

\usepackage{enumitem}
\setlist{leftmargin=*}

\renewcommand\[{\begin{equation}}
\renewcommand\]{\end{equation}} 

\makeatother

  \providecommand{\corollaryname}{Corollary}
  \providecommand{\definitionname}{Definition}
  \providecommand{\examplename}{Example}
  \providecommand{\lemmaname}{Lemma}
  \providecommand{\propositionname}{Proposition}
  \providecommand{\remarkname}{Remark}
\providecommand{\theoremname}{Theorem}

\begin{document}
\global\long\def\disgraph{\mathcal{G}}

\global\long\def\metgraph{\Gamma}

\global\long\def\hamil{\mathcal{H}}

\global\long\def\len{\mathscr{L}}

\global\long\def\lencl{\mathscr{\overline{L}}}

\global\long\def\lenbd{\mathscr{\partial L}}

\global\long\def\lenvec{\underline{l}}

\global\long\def\sgp{\theta^{SG}}

\global\long\def\ray{\mathcal{R}}

\global\long\def\energy{\mathscr{E}}

\global\long\def\E{\mathcal{E}}

\global\long\def\V{\mathcal{V}}

\global\long\def\R{\mathbb{R}}

\global\long\def\Q{\mathbb{Q}}

\global\long\def\Z{\mathbb{Z}}

\global\long\def\C{\mathbb{C}}

\global\long\def\N{\mathbb{N}}

\global\long\def\at{\mathbf{\big|}}

\global\long\def\id{\mathbf{1}}

\global\long\def\ui{\mathbf{\textrm{i}}}

\global\long\def\ue{\mathbf{\textrm{e}}}

\global\long\def\ud{\mathbf{\textrm{d}}}

\title[Quantum graphs which optimize the spectral gap]{Quantum graphs which optimize the spectral gap}

\author{Ram Band$^{1}$ and Guillaume L\'{e}vy$^{2}$}

\address{$^{1}${\small{}Department of Mathematics, Technion--Israel Institute
of Technology, Haifa 32000, Israel}}

\address{$^{2}$Laboratoire Jacques-Louis Lions, UMR 7598, Université Pierre
et Marie Curie, 75252 Paris Cedex 05, France.}

\email{$^{1}${ramband@technion.ac.il},}

\email{$^{2}${levy@ljll.math.upmc.fr}}

\begin{abstract}
A finite discrete graph is turned into a quantum (metric) graph once
a finite length is assigned to each edge and the one-dimensional Laplacian
is taken to be the operator. We study the dependence of the spectral
gap (the first positive Laplacian eigenvalue) on the choice of edge
lengths. In particular, starting from a certain discrete graph, we
seek the quantum graph for which an optimal (either maximal or minimal)
spectral gap is obtained. We fully solve the minimization problem
for all graphs. We develop tools for investigating the maximization
problem and solve it for some families of graphs.
\end{abstract}

\maketitle

\section{Introduction\label{sec:Introduction}}

The spectral gap is a vastly explored quantity due to its importance
both for applicative purposes and theoretic ones. The applicative
aspects range from estimates of convergence to equilibrium to behavior
of quantum many body systems. The theoretic study concerns with connecting
the shape of an object to a fundamental spectral property. Such relations
stand in the heart of spectral geometry and motivate the current work.

A compact quantum graph can be thought of as a three-fold object,
consisting of a topology, a metric and an operator. The topology is
described by an underlying discrete graph and the metric is simply
the assignment of a positive length to each of the edges. The operator
together with its domain complete this description. In the current
work we adopt the most common choice and fix the operator to be the
one-dimensional Laplacian acting on functions which satisfy the so
called Neumann conditions at the graph vertices (see \cite{BerKuc_quantum_graphs,GnuSmi_ap06}).
It is then most natural to fix a certain graph topology and explore
how the graph spectral properties depend on the choice of edge lengths
\cite{Friedlander_ijm05,ExnJax_pla12,BerLiu_jmaa16}. In particular,
we examine the spectral gap which, in our case, is the first positive
eigenvalue of the Laplacian. Picking a particular graph topology,
we ask which edge lengths minimize or maximize the spectral gap. We
notice that as our space of edge lengths is not compact, it is possible
that there is no maximum or no minimum. The space of edge lengths
is thus extended by allowing zero length edges so that the minima
(maxima) of this new length space are the infimums (supremums) of
the previous. This leads to a most interesting exploration direction:
sending edge lengths to zero changes the topology of the original
graph and makes us wonder what are the topologies which are obtained
as optimizers (either maximizers or minimizers) of other graphs. This
is the central question of the current paper.

Already in 1987, Nicaise showed that among all graphs with a fixed
length, the minimal spectral gap is obtained for the single edge graph
\cite{Nicaise_bsm87}. In 2005, Friedlander proved a more general
result, showing that the minimum of the $k^{\textrm{th}}$ eigenvalue
is uniquely obtained for a star graph with $k$ edges \cite{Friedlander_aif05}.
More recently, Exner and Jex showed how the change of graph edge lengths
may increase or decrease the spectral gap, depending on the graph's
topology \cite{ExnJax_pla12}. In the last couple of years, a series
of works on the subject came to light. Kurasov and Naboko \cite{KurNab_jps14}
treated the spectral gap minimization and together with Malenov\'a
they explored how the spectral gap changes with various modifications
of the graph connectivity \cite{KurMalNab_jpa13}. Kennedy, Kurasov,
Malenov\'a and Mugnolo provided a broad survey on bounding the spectral
gap in terms of various geometric quantities of the graph \cite{KenKurMalMug_ahp16}.
Karreskog, Kurasov and Trygg Kupersmidt generalized the minimization
results mentioned above to Schr\"odinger operators with potentials
and $\delta$-type vertex conditions \cite{KarKurKup_ams16}. Del
Pezzo and Rossi proved upper and lower bounds for the spectral gap
of the p-Laplacian and evaluated its derivatives with respect to change
of edge lengths \cite{DelRoss_amp16}. Rohleder solved the spectral
gap maximization problem for all eigenvalues of tree graphs \cite{Rohleder_pams16}.
When this manuscript was accpeted for publication, two additional
preprints became available online. Ariturk provides some improved
upper bounds for all graph eigenvalues \cite{Ariturk_arXiv16}. Berkolaiko,
Kennedy, Kurasov and Mugnolo further generalize lower and upper bounds
of the spectral gap in terms of the edge connectivity \cite{BerKenKurMug_arXiv17}.

We complement this literature review by mentioning some interesting
and recent works on the spectral gap of metric graphs, whose scope
is different than ours. Post \cite{Post_coll09}, Kurasov \cite{Kurasov_acta13},
Kennedy and Mugnolo \cite{KenMun_arxv16} all treated various estimates
of the spectral gap in terms of the Cheeger constant (a line of research
which already originated in \cite{Nicaise_bsm87} for quantum graphs).
Buttazzo, Ruffini and Velichkov optimize over spectral gap of graphs
given some prescribed set of Dirichlet vertices embedded in $\R^{d}$
\cite{ButRufVel_esaim12}.

The spectral gap optimization we consider in this paper is close in
nature to the first line of works mentioned above. Nevertheless, our
point of view is different as we wish to solve the optimization problem
for each and every topology. This broad phrasing of the question provides
a unified framework for several of the works mentioned above. In particular,
it allows to take a step forward and complement those.

\subsection{Discrete graphs and graph topologies\label{subsec:discrete_graphs_introduction}}

Let $\disgraph=\left(\V,\E\right)$ be a connected graph with finite
sets of vertices $\V$ and edges $\E$ and we denote $V:=\left|\V\right|,~E:=\left|\E\right|$.
We allow edges to connect either two distinct vertices or a vertex
to itself. In the latter case, this edge is called a loop, or sometimes
a petal.

For a vertex $v\in\V$, its \emph{degree, $d_{v}$,} equals the number
of edges connected to it. Vertices of degree one are called leaves.
Furthermore, we abuse this naming and frequently also use the name
leaf for an edge which is connected to a vertex of degree one.

An important topological quantity of the graph is 
\begin{equation}
\beta:=E-V+1,\label{eq:Betti_number}
\end{equation}
which counts the number of ``independent'' cycles on the graph (assuming
the graph is connected). This is also known as the first Betti number,
which is the dimension of the graph's first homology. In particular,
tree graphs are characterized by $\beta=0$.

We consider the following two ways for treating the graph connectivity.
The graph's edge connectivity is the minimal number of edges one needs
to remove in order to disconnect the graph. If the graph's edge connectivity
equals one, then an edge whose removal disconnects the graph is called
a \emph{bridge}. In particular, leaf edges are bridges. Similarly,
the graph's vertex connectivity is the number of vertices needed to
be removed in order to disconnect the graph. In particular, we show
the special role played by graphs of edge connectivity one (Theorem
\ref{thm:infimizers}) and of vertex connectivity one (Theorem \ref{thm:supremum_of_gluing}).

\subsection{Spectral theory of quantum graphs\label{subsec:Spectral-theory-of-Quantum-Graphs}}

A \emph{metric graph} is a discrete graph for which each edge, $e\in\E$,
is identified with a one-dimensional interval $[0,l_{e}]$ of positive
finite length $l_{e}$. We assign to each edge $e\in\E$ a coordinate,
$x_{e}$, which measures the distance along the edge from the starting
vertex of $e$. We denote a coordinate by $x$, when its precise nature
is unimportant.

A function on the graph is described by its restrictions to the edges,
$\left\{ \left.f\right|_{e}\right\} _{e\in\E}$, where $\left.f\right|_{e}:\left[0,l_{e}\right]\rightarrow\C$.
We equip the metric graphs with a self-adjoint differential operator,
\begin{equation}
\mathcal{H}\ :\ \left.f\right|_{e}(x_{e})\mapsto-\frac{\ud^{2}}{\ud x_{e}^{2}}\left.f\right|_{e}\left(x_{e}\right),\label{eq:metric_Schroedinger}
\end{equation}
which in our case is just the one-dimensional negative Laplacian on
every edge\footnote{Note that more general operators appear in the literature. See for
example the book \cite{BerKuc_quantum_graphs} and the survey \cite{GnuSmi_ap06}.}. It is most common to call this setting of a metric graph and an
operator by the name quantum graph.

To complete the definition of the operator we need to specify its
domain. We denote by $H^{2}(\metgraph)$ the following direct sum
of Sobolev spaces 
\begin{equation}
H^{2}(\metgraph):=\bigoplus_{e\in\E}H^{2}([0,l_{e}])\ .\label{eq:Sobolev_on_graph}
\end{equation}
In addition we require the following matching conditions on the graph
vertices. A function $f\in H^{2}(\metgraph)$ is said to satisfy the
Neumann vertex conditions at a vertex $v$ if
\begin{enumerate}
\item $f$ is continuous at $v\in\V$, i.e., 
\begin{equation}
\forall e_{1},e_{2}\in\E_{v\,\,\,\,\,}\left.f\right|_{e_{1}}(0)=\left.f\right|_{e_{2}}(0),\label{eq:Neumann_continuity}
\end{equation}
where $\E_{v}$ is the set of edges connected to $v$, and for each
$e\in\E_{v}$ we choose the coordiante such that $x_{e}=0$ at $v$.
\item the outgoing derivatives of $f$ at $v$ satisfy 
\begin{equation}
\sum_{e\in\E_{v}}\left.\frac{\ud f}{\ud x_{e}}\right|_{e}\left(0\right)=0.\label{eq:Neumann_deriv_conditions}
\end{equation}
\end{enumerate}
Another common vertex condition is called the Dirichlet condition.
Imposing Dirichlet condition at vertex $v\in\V$ means 
\begin{equation}
\forall e\in\E_{v\,\,\,\,\,}\left.f\right|_{e}(0)=0.\label{eq:Dirichlet_vertex_condition}
\end{equation}
Requiring either of these conditions at each vertex leads to the operator
\eqref{eq:metric_Schroedinger} being self-adjoint and its spectrum
being real and bounded from below \cite{BerKuc_quantum_graphs}. In
addition, since we only consider compact graphs, the spectrum is discrete.
We number the eigenvalues in the ascending order and denote them with
$\left\{ \lambda_{n}\right\} _{n=0}^{\infty}$ and their corresponding
eigenfunctions with $\left\{ f_{n}\right\} _{n=0}^{\infty}$. As the
operator is both real and self-adjoint, we may choose the eigenfunctions
to be real, which we will always do.

In this paper, we almost solely consider graphs whose vertex conditions
are Neumann at all vertices. Those are called \emph{Neumann graphs.}
For Neumann graphs, we define the Rayleigh quotient 
\[
\mathcal{R}(f):=\frac{\int_{\metgraph}|f'(x)|^{2}dx}{\int_{\metgraph}|f(x)|^{2}dx},
\]
which makes sense whenever $f\in H^{1}(\metgraph)$ (see \eqref{eq:Sobolev_on_graph}).
The eigenvalues of a Neumann graph have a nice expression using the
Rayleigh quotient. Indeed, denoting $V_{n}:=\textrm{Span}\{f_{0},\dots,f_{n}\}$
for $n\in\mathbb{N}$, we have 
\[
\lambda_{n}=\min_{f\perp V_{n-1}}\mathcal{R}(f).
\]
 In particular, the spectrum of a Neumann graph is nonnegative, which
means that we may represent the spectrum by the non-negative square
roots of the eigenvalues, $k_{n}=\sqrt{\lambda_{n}}$, and say that
$\left\{ k_{n}\right\} _{n=0}^{\infty}$ are the $k$-eigenvalues
of the graph. For convenience, we express most of our results and
proofs in terms of the $k$-eigenvalues. This choice makes all expressions
of this paper look nicer. A Neumann graph has $k_{0}=0$ with multiplicity
which equals the number of graph components (which is taken to be
one throughout this paper). It is $k_{1}$ which is in the focus of
this paper and is called the spectral gap\footnote{This terminology is justified, as a spectral gap is a common name
for the difference between some trivial eigenvalue (which is $k_{0}=0$
in our case) and the next eigenvalue. We note that in this sense it
is also common to call $\lambda_{1}$ the spectral gap.}.

\subsection{Graph Optimizers}
\begin{defn}
\label{def:discrete_and_metric_graph}Let $\disgraph$ be a discrete
graph with $E$ edges. 

\begin{enumerate}
\item Denote by 
\[
\len_{\disgraph}:=\left\{ \left(l_{1},\ldots,l_{E}\right)\in\R^{E}\left|~\sum_{e=1}^{E}l_{e}=1~\textrm{and~}\forall e,~l_{e}>0\right.\right\} 
\]
the space of all possible lengths we may assign to the edges of $\disgraph$.
We further denote by $\lencl_{\disgraph}$ the closure of $\len$
in $\R^{E}$ and by $\lenbd$ its boundary.
\item Denote by $\metgraph(\disgraph;~\lenvec$) the metric graph whose
connectivity is the same as $\disgraph$ and whose edge lengths are
given by $\lenvec\in\lencl_{\disgraph}$. We take $\metgraph(\disgraph;~\lenvec)$
to be a Neumann graph. If $\lenvec\in\partial\len$, then $\lenvec$
has some vanishing entries and in this case the connectivity of $\metgraph(\disgraph;~\lenvec$)
is not the same as $\disgraph$. For each vanishing entry, $l_{e}=0$,
the edge $e$ does not exist in $\metgraph(\disgraph;~\lenvec),$
but rather the vertices at the endpoints of this edge are identified
and form a single vertex when considered in $\metgraph(\disgraph;~\lenvec)$.
\end{enumerate}
\end{defn}
\noindent We emphasize that the definition above contains a normalization
choice; unless otherwise stated, all the graphs studied in this paper
are required to have total metric length one.

This paper studies the spectral gap, $k_{1}\left[\metgraph(\disgraph;~\lenvec)\right]$,
as a function of $\lenvec\in\lencl_{\disgraph}$. A first step is
to show that the function $k_{1}\left[\metgraph(\disgraph;~\lenvec)\right]$
is continuous on $\lencl_{\disgraph}$, which is done in Appendix
\ref{sec:appendix_eigenvalue_continuity}. Combining this continuity
statement with the compactness of ths set $\overline{\len_{\disgraph}}$,
the existence of a maximum and a minimum of the spectral gap on $\overline{\len_{\disgraph}}$
(but not necessarily on $\len_{\disgraph}$) follows. Indeed, the
focus of the current paper is on the extremal points of $k_{1}\left[\metgraph(\disgraph;~\lenvec)\right]$.
In particular we investigate whether the extremal points are obtained
on $\len_{\disgraph}$ or on $\lenbd_{\disgraph}$ and to which metric
graphs $\metgraph(\disgraph;~\lenvec)$ they correspond. This motivates
the following.
\begin{defn}
\label{def:optimizers}Let $\disgraph$ be a discrete graph. 

\begin{enumerate}
\item $\metgraph(\disgraph;~\lenvec^{*})$ is called a maximizer of $\disgraph$
if $\lenvec^{*}\in\len_{\disgraph}$ and 
\[
\forall\lenvec\in\len_{\disgraph}\quad k_{1}\left[\metgraph\left(\disgraph;~\lenvec^{*}\right)\right]\geq k_{1}\left[\metgraph\left(\disgraph;~\lenvec\right)\right].
\]
In this case we call $k_{1}\left[\metgraph(\disgraph;~\lenvec^{*})\right]$
the maximal spectral gap of $\disgraph$.
\item $\metgraph(\disgraph;~\lenvec^{*})$ is called a supremizer of $\disgraph$
if $\lenvec^{*}\in\lencl_{\disgraph}$ and 
\[
\forall\lenvec\in\lencl_{\disgraph}\quad k_{1}\left[\metgraph\left(\disgraph;~\lenvec^{*}\right)\right]\geq k_{1}\left[\metgraph\left(\disgraph;~\lenvec\right)\right].
\]
In this case we call $k_{1}\left[\metgraph(\disgraph;~\lenvec^{*})\right]$
the supremal spectral gap of $\disgraph$.
\item $\metgraph(\disgraph;~\lenvec^{*})$ is called the unique maximizer
of $\disgraph$ if for all $\lenvec\neq\lenvec^{*}$, $\metgraph(\disgraph;~\lenvec)$
is not a maximizer of $\disgraph$. The same definition holds for
the unique supremizer.
\item Analogous definitions to the above hold for minimizers and infimizers.
\item $\metgraph(\disgraph;~\lenvec^{*})$ is called an optimizer of $\disgraph$
if it is either a supremizer, a maximizer, an infimizer or a minimizer
of $\disgraph$.
\end{enumerate}
\end{defn}
Continuing the discussion preceding the definition, we note that there
might be graphs which do not have a maximizer or a minimizer. Yet,
a supremizer and an infimizer exist for any graph. Let $\disgraph$
be a discrete graph and $\metgraph(\disgraph;~\lenvec^{*})$ be its
supremizer (infimizer), with $\lenvec^{*}\in\lencl_{\disgraph}$.
Denote by $\disgraph^{*}$ the discrete graph which corresponds to
$\metgraph(\disgraph;~\lenvec^{*})$. We note that if $\lenvec^{*}\in\len_{\disgraph}$
then $\disgraph^{*}=\disgraph$ and if $\lenvec^{*}\in\lenbd_{\disgraph}$
then $\disgraph^{*}$ is obtained from $\disgraph$ by contracting
all edges which correspond to the zero entries of $\lenvec$.

The questions which motivate this work are the following: what are
the metric graphs $\metgraph(\disgraph;~\lenvec^{*})$ which serve
as supremizers (or infimizers) and what are all the possible topologies
(i.e. the discrete graphs $\disgraph^{*}$) obtained by these optimizations? 

We start by presenting a few examples of topologies which form part
of the answer to the questions above.
\begin{example}
\label{exa:star}Star graph

Let $\disgraph$ be a graph with $V\geq3$ vertices, and $E=V-1$
edges, where one of the vertices (called the central vertex) is connected
by edges to all the $V-1$ other vertices (Figure \ref{fig:star_flower_stower}(a)).
$\disgraph$ is called a star graph. The graph $\metgraph(\disgraph;~\lenvec)$
with $\lenvec=(\frac{1}{E},\ldots,\frac{1}{E})$ is called the equilateral
star. A simple calculation shows that $k_{1}\left[\metgraph(\disgraph;~\lenvec)\right]=\frac{\pi}{2}E$.
We show (Theorem \ref{thm:tree_supremizer}) that the equilateral
star is the unique maximizer of the star topology and that it is also
the unique supremizer of any tree graph with $E$ leaves. If we choose
above $V=2,~E=1$ we get an interval, which is the unique infimizer
of any graph with a bridge (Theorem \ref{thm:infimizers}).
\end{example}
\begin{figure}
\hfill{}%
\begin{minipage}[t]{0.3\columnwidth}%
(a)\includegraphics[scale=0.44]{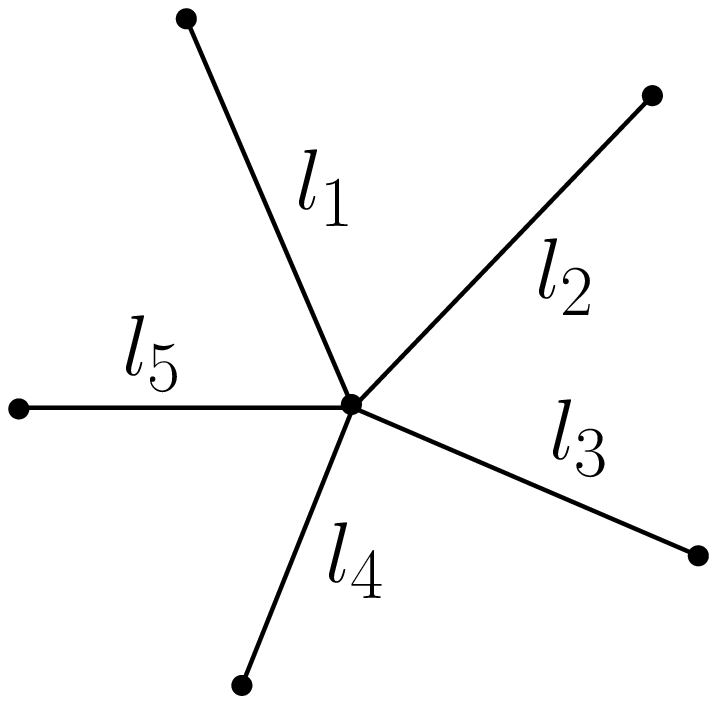}%
\end{minipage}\hfill{}%
\begin{minipage}[t]{0.3\columnwidth}%
(b)\includegraphics[scale=0.44]{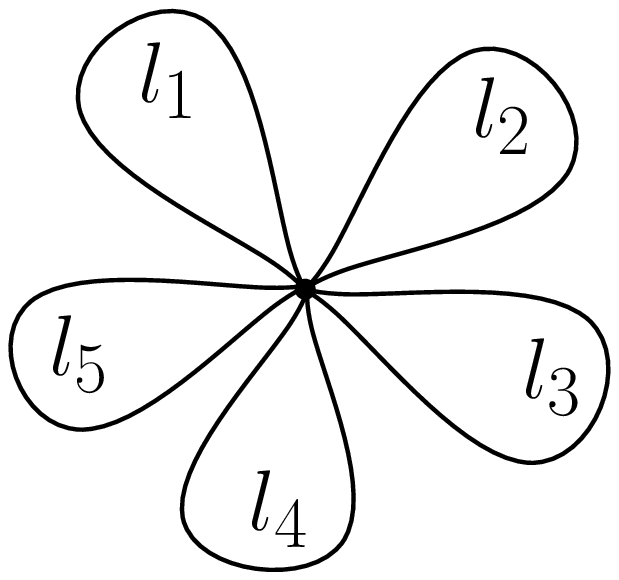}%
\end{minipage}\hfill{}%
\begin{minipage}[t]{0.3\columnwidth}%
(c)\includegraphics[scale=0.44]{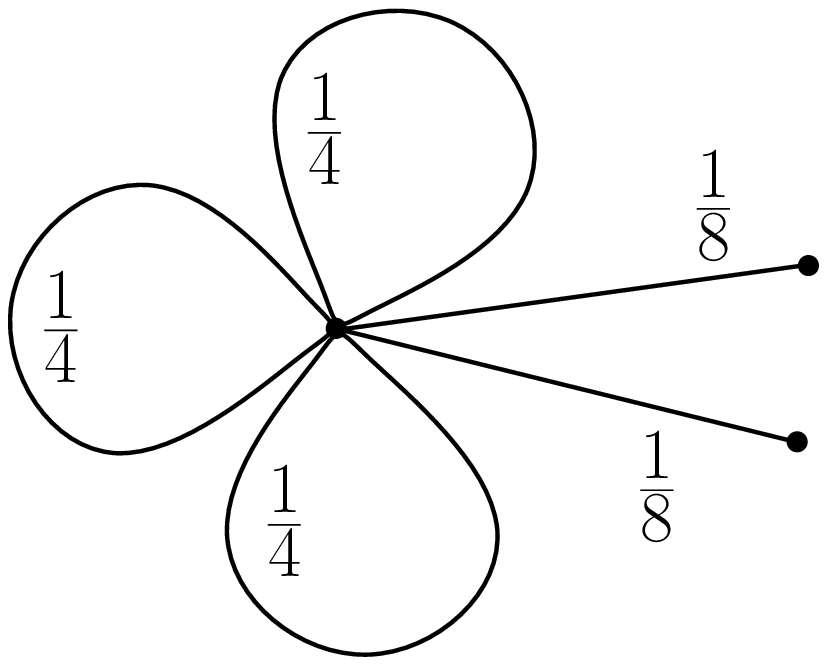}%
\end{minipage}\hfill{}

\caption{A few basic examples. (a) star graph (b) flower graph (c) equilateral
stower graph with $E_{p}=3,~E_{l}=2$}
\label{fig:star_flower_stower}
\end{figure}

\begin{example}
\label{exa:flower}Flower graph

Let $\disgraph$ be a graph with a single vertex and $E\geq2$ edges,
where each edge is a loop (petal) connecting that single vertex to
itself (Figure \ref{fig:star_flower_stower}(b)). $\disgraph$ is
called a flower graph. The graph $\metgraph(\disgraph;~\lenvec)$
with $\lenvec=(\frac{1}{E},\ldots,\frac{1}{E})$ is called the equilateral
flower. A simple calculation shows that $k_{1}\left[\metgraph(\disgraph;~\lenvec)\right]=\pi E$.
We show (Corollary \ref{cor:stower_supremum}) that the equilateral
flower is the unique maximizer of the flower topology. If we choose
above $E=1$ we get a single loop graph, which is an infimizer for
all bridgeless graphs (Theorem \ref{thm:infimizers}).
\end{example}
~
\begin{example}
\label{exa:stower}Stower graph

Let $\disgraph$ be a graph with $V$ vertices and $E=E_{p}+E_{l}\geq2$
edges. $E_{p}$ of the edges are loops which connect a single vertex
to itself (the same vertex for all those edges) and, as before, they
are called petals. Each of the rest $E_{l}=V-1$ edges connect this
single vertex to another graph vertex and they are called dangling
edges or just leaves (Figure \ref{fig:star_flower_stower}(c)). Being
a hybrid between a star graph and a flower graph, such $\disgraph$
is called a stower graph. We note that a flower graph is a stower
(with $E_{l}=0$) and a star graph is a stower as well (with $E_{p}=0$).
The graph $\metgraph(\disgraph;~\lenvec)$ with $\lenvec=\frac{1}{2E_{p}+E_{l}}(\underbrace{2,\ldots,2}_{E_{p}},\underbrace{1,\ldots,1}_{E_{l}})$
is called the equilateral stower. Note that we abuse terminology and
call the graph equilateral, even though not all edges of the description
above have the same length. A simple calculation shows that $k_{1}\left[\metgraph(\disgraph;~\lenvec)\right]=\frac{\pi}{2}(2E_{p}+E_{l})$.
We show (Corollary \ref{cor:stower_supremum}) that the equilateral
stower is the unique maximizer of the stower topology, except when
$E_{p}=E_{l}=1$, for which the supremizer is actually a single loop.
Furthermore, spectral gaps of stowers obey a sort of additive property
in the following sense: if two graphs whose supremizers are stowers
are glued at non-leaf vertices to form a single graph, then this graph's
supremizer is a stower graph obtained by adding the petals and the
leaves of the two individual stower supremizers (Corollary \ref{cor:stower_supremum}).
\end{example}
~
\begin{example}
\label{exa:mandarin}Mandarin graph

Let $\disgraph$ be a graph with $2$ vertices and $E$ edges, each
connecting those two vertices (Figure \ref{fig:mandarin_and_necklace}(a)).
Such $\disgraph$ is called a mandarin graph. In the literature it
is also called a watermelon or a pumpkin, but we adopt the name mandarin
which was used in a thorough exploration of spectral properties of
these graphs, \cite{BanBerWey_jmp15}. The graph $\metgraph(\disgraph;~\lenvec)$
with $\lenvec=(\frac{1}{E},\ldots,\frac{1}{E})$ is called the equilateral
mandarin. A simple calculation shows that $k_{1}\left[\metgraph(\disgraph;~\lenvec)\right]=\pi E$.
The equilateral mandarin is the unique maximizer of the mandarin topology,
as was shown recently in \cite{KenKurMalMug_ahp16} (theorem 4.2 there).
\end{example}
\begin{figure}[h]
\hfill{}%
\begin{minipage}[t]{0.4\columnwidth}%
(a)\includegraphics[scale=0.58]{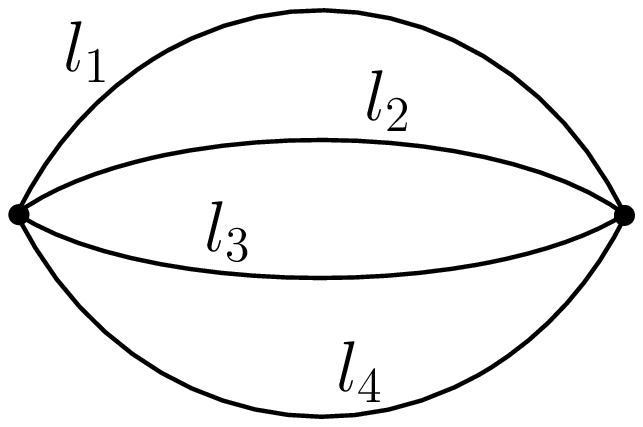}%
\end{minipage}\hfill{}%
\begin{minipage}[t]{0.5\columnwidth}%
(b)\includegraphics[scale=0.7]{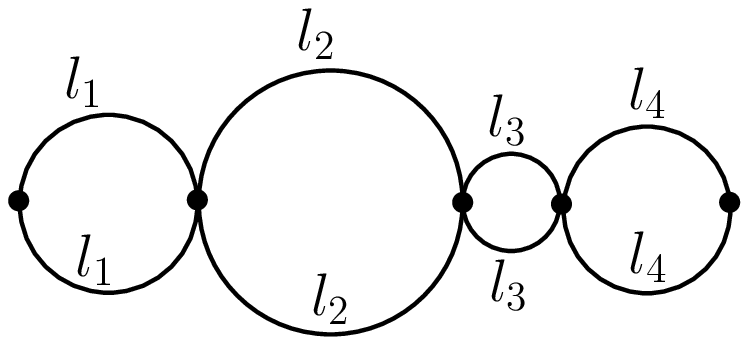}%
\end{minipage}\hfill{}

\caption{(a) mandarin graph (b) symmetric necklace graph}
\label{fig:mandarin_and_necklace}
\end{figure}

\begin{example}
\label{exa:necklace}Necklace graph

Let $\disgraph$ be a graph with $V$ vertices and $E=2\left(V-1\right)$
edges, such that every two adjacent vertices, $v_{i},v_{i+1}$ ($1\leq i\leq V-1$)
are connected by two edges (Figure \ref{fig:mandarin_and_necklace}(b)).
If $\lenvec$ is chosen such that every pair of parallel edges connecting
two vertices have the same length, $\metgraph(\disgraph;~\lenvec)$
is called a symmetric necklace. Note that the two vertices at the
endpoints of the necklace are redundant, being Neumann vertices of
degree two (they are merely used here to shorten the graph description).
Necklace graphs are the only graphs which may serve as infimizers
of bridgeless graphs (Theorem \ref{thm:infimizers}).
\end{example}

\section{Main Results\label{sec:Main-Results}}

The main results of the current paper are stated below, arranged by
subjects. In each of the following subsections, we mention which section
of the paper contains the relevant proofs and discussions.

\subsection{Infimizers (section \ref{sec:Infimizers})}
\begin{thm}
\label{thm:infimizers}~

\begin{enumerate}
\item Let $\mathcal{G}$ be a graph with a bridge. Then the infimal spectral
gap of $\mathcal{G}$ equals $\pi$. Moreover, the unique infimizer
is the unit interval.
\item Let $\mathcal{G}$ be a bridgeless graph. Then the infimal spectral
gap of $\mathcal{G}$ equals $2\pi$. Moreover, any infimizer is a
symmetric necklace graph.
\end{enumerate}
\end{thm}
We note that it was already proved in \cite{Nicaise_bsm87,Friedlander_aif05,KurNab_jps14}
that $\pi$ is a universal lower bound for the spectral gap, attained
only by the interval. In \cite{Friedlander_aif05} it is even shown
that $\pi n$ is a lower bound for $k_{n}$. The paper \cite{KurNab_jps14}
proves that the lower bound may be improved to $2\pi$ if all vertices
have even degrees. Theorem \ref{thm:infimizers} extends the set of
graph topologies whose spectral gap is bounded by $2\pi$ to all bridgeless
graphs (indeed graphs whose all vertices are of even degrees form
a particular case). Furthermore, combining Theorem \ref{thm:infimizers}
with the continuity of eigenvalues with respect to the graphs edge
lengths (Appendix \ref{sec:appendix_eigenvalue_continuity}) allows
to conclude that our result cannot be improved by imposing further
restrictions on the graph topology. For any bridgeless graph $\disgraph$,
there exists $\lenvec^{*}\in\lencl_{\disgraph}$ for which $\metgraph(\disgraph;~\lenvec^{*})$
is a single cycle graph with spectral gap $2\pi$. As $k_{1}\left[\metgraph(\disgraph;~\lenvec)\right]$
is a continuous function of $\lenvec$, the spectral gap may be as
close to $2\pi$ as we wish, by choosing $\lenvec\in\len_{\disgraph}$
close enough to $\lenvec^{*}$. Similarly, the lower bound $\pi$
cannot be improved for graphs with a bridge. Therefore, Theorem \ref{thm:infimizers}
complements the previous results and provides a complete answer to
the infimization problem.

\subsection{Supremizers of tree graphs (section \ref{sec:TreeGraphs})}
\begin{thm}
\label{thm:tree_supremizer}Let $\disgraph$ be a tree graph with
$E_{l}\geq2$ leaves. Then the unique supremizer of $\disgraph$ is
the equilateral star with $E_{l}$ edges, whose spectral gap is $\frac{\pi}{2}E_{l}$.
In particular, the uniqueness implies that this supremizer is a maximizer
if and only if $\disgraph$ is a star graph.
\end{thm}
Theorem \ref{thm:tree_supremizer} completely solves the optimization
problem for tree graphs. While writing this paper, we became aware
of the recent work, \cite{Rohleder_pams16}, which solves the maximization
problem for trees (theorem 3.2 there). In the course of doing so,
that work provides the upper bound $\frac{\pi}{2}E$ on the spectral
gap of trees\footnote{Theorem 3.2 in that paper is actually more general and provides the
upper bound $\frac{\pi n}{2}E$ for $k_{n}$. }. Our proof is close in spirit to that of theorem 3.4 in \cite{Rohleder_pams16}.
Yet, thanks to a basic geometric observation (Lemma \ref{lem:tree_diameter_lower_bound}
here), the better bound $\frac{\pi}{2}E_{l}$ is obtained\footnote{Furthermore, the same geometric observation may be used to improve
the more general theorem 3.2 of \cite{Rohleder_pams16}.}.

Theorem \ref{thm:tree_supremizer} allows to deduce the following.
\begin{cor}
\label{cor:non_tree_supremizer} Let $\disgraph$ be a non-tree graph.
Then its supremizer is not a tree graph.
\end{cor}

\subsection{Supremizers whose spectral gap is a simple eigenvalue (section \ref{sec:LocalOptimizers})}

Whenever the spectral gap is a simple eigenvalue, it is differentiable
with respect to edge lengths, which allows to search for local maximizers.
There are indeed examples for critical values (not just maximizers)
of the spectral gap, which we demonstrate in Proposition \ref{prop:Standarins}.
If such local critical point is actually a supremizer it is possible
to prove the following.
\begin{thm}
\label{thm:simple_supremum} Let $\disgraph$ be a discrete graph
and let $\lenvec\in\len_{\disgraph}$. Assume that $\metgraph\left(\disgraph;~\lenvec\right)$
is a supremizer of $\disgraph$ and that the spectral gap $k_{1}\left(\metgraph(\disgraph;~\lenvec)\right)$
is a simple eigenvalue. Then $\metgraph\left(\disgraph;~\lenvec\right)$
is not a unique supremizer. There exists a choice of lengths $\lenvec^{*}\in\lencl_{\disgraph}$
such that $\metgraph\left(\disgraph;~\lenvec^{*}\right)$ is an equilateral
mandarin and 
\[
k_{1}\left(\metgraph\left(\disgraph;~\lenvec\right)\right)=k_{1}\left(\metgraph\left(\disgraph;~\lenvec^{*}\right)\right).
\]
\end{thm}

\subsection{Supremizers of vertex connectivity one (sections \ref{sec:Gluing_Graphs},
\ref{sec:Symmetrization}, \ref{sec:Applications_of_Gluing_and_Symmetrization})}

Next, we describe a bottom to top construction which allows to find
out a supremizer of a graph by knowing the supremizers of two of its
subgraphs. This is possible for graphs of vertex connectivity one.
In order to state the result, the following criteria are introduced.
\begin{defn}
\label{def:Dirichlet_criterion}

\begin{enumerate}
\item \label{enu:def_Dirichlet_criterion_1} A Neumann graph $\Gamma$ obeys
the \emph{Dirichlet criterion} with respect to its vertex $v$ if
imposing Dirichlet vertex condition at $v$ does not change the value
of $k_{1}$ (comparing to the one with Neumann condition at $v$).
\item \label{enu:def_Dirichlet_criterion_2} A Neumann graph $\Gamma$ obeys
the \emph{strong Dirichlet criterion} with respect to its vertex $v$
if it obeys the Dirichlet criterion and if imposing the Dirichlet
vertex condition at $v$ strictly increases the eigenvalue multiplicity
of $k_{1}$. 
\end{enumerate}
\end{defn}
\begin{thm}
\label{thm:supremum_of_gluing}Let $\disgraph_{1},\disgraph_{2}$
be discrete graphs, let $v_{i}$ ($i=1,2$) be a vertex of $\disgraph_{i}$.
Let $\disgraph$ be the graph obtained by identifying $v_{1}$ and
$v_{2}$. Let $\lenvec^{\left(i\right)}\in\lencl_{\disgraph_{i}}$
and $\metgraph_{i}:=\metgraph(\disgraph;~\lenvec^{\left(i\right)})$
be the corresponding metric graphs. Define $\lenvec:=(L\lenvec^{\left(1\right)},~\left(1-L\right)\lenvec^{\left(2\right)})\in\lencl_{\disgraph}$,
for some $L\in\left[0,1\right]$. Then the graph $\metgraph:=\metgraph(\disgraph;~\lenvec)$
is a supremizer of $\disgraph$ if all the following conditions are
met:

\begin{enumerate}
\item \label{enu:thm_SGP_condition_1} $L=\frac{k_{1}\left(\Gamma_{1}\right)}{k_{1}\left(\Gamma_{1}\right)+k_{1}\left(\Gamma_{2}\right)}$.
\item \label{enu:thm_SGP_condition_2} $\Gamma_{i}$ is a supremizer of
$\disgraph_{i}$ ($i=1,2$).
\item \label{enu:thm_SGP_condition_3} $\Gamma_{i}$ obeys the Dirichlet
criterion with respect to $v_{i}$ ($i=1,2$).
\end{enumerate}
If we further assume either of the following: \renewcommand{\labelenumi}{(\alph{enumi})}
\begin{enumerate}
\item For both $i=1,2$ , $\metgraph_{i}$ is a unique supremizer of $\disgraph_{i}$
or
\item For both $i=1,2$, $\metgraph_{i}$ obeys the strong Dirichlet criterion
and any other supremizer of $\disgraph_{i}$ violates the Dirichlet
criterion.
\end{enumerate}
then $\Gamma$ is the unique supremizer of $\disgraph$.

\end{thm}
\begin{rem*}
This theorem may be strengthened by weakening condition \eqref{enu:thm_SGP_condition_3}.
Yet, the description of the weaker condition is more technical and
we leave its specification, as well as the proof of the stronger version
of this theorem, to section \ref{sec:Gluing_Graphs}.
\end{rem*}
We note that the equilateral stower obeys the Dirichlet criterion
with respect to its central vertex. Obviously, this observation also
includes the equilateral star and equilateral flower as special cases.
This observation together with theorem \ref{thm:supremum_of_gluing}
allow to prove the following corollaries.
\begin{cor}
\label{cor:gluing_stowers} Let $\disgraph_{1},\disgraph_{2}$ be
discrete graphs. Denote by $v_{1},v_{2}$ non-leaf vertices of each
of those graphs and let $\disgraph$ be the graph obtained by identifying
$v_{1}$ and $v_{2}$. If the (unique) supremizer of $\disgraph_{i}$
is the equilateral stower with $E_{p}^{\left(i\right)}$ petals and
$E_{l}^{\left(i\right)}$ leaves, such that $E_{p}^{\left(i\right)}+E_{l}^{\left(i\right)}\geq2$,
then the (unique) supremizer of $\disgraph$ is an equilateral stower
with $E_{p}^{\left(1\right)}+E_{p}^{\left(2\right)}$ petals and $E_{l}^{\left(1\right)}+E_{l}^{\left(2\right)}$
leaves.
\end{cor}
We note that as we have shown (Theorem \ref{thm:tree_supremizer})
that equilateral stars are the unique supremizers of trees, the corollary
above implies that gluing a tree (at its internal vertex) to any graph
whose (unique) supremizer is a stower gives a graph whose (unique)
supremizer is a stower as well.
\begin{cor}
\label{cor:stower_supremum} Let $\disgraph$ be a stower graph with
$E_{p}$ petals and $E_{l}$ leaves, such that $E_{p}+E_{l}\geq2$
and $(E_{p},E_{l})\neq(1,1)$ . Then it has a maximizer which is the
equilateral stower graph with $E_{p}$ petals and $E_{l}$ dangling
edges and the corresponding spectral gap is $\frac{\pi}{2}\left(2E_{p}+E_{l}\right)$.
Furthermore, this maximizer is unique for all cases except $\left(E_{p},E_{l}\right)\in\left\{ \left(2,0\right),\left(1,2\right)\right\} $.
\end{cor}
We remark that a partial result of the above was already proved within
the proof of theorem 4.2 in \cite{KenKurMalMug_ahp16}. It was shown
there that the equilateral flower is the unique maximizer among all
flowers\footnote{It is claimed there that the equilateral flower is the unique maximizer
for all flowers with $E\geq2$. Actually, the uniqueness does not
hold for the $E=2$ case, as we show in the proof of Corollary \ref{cor:stower_supremum}.}. This was used there to prove the global bound $k_{1}\left[\metgraph\right]\leq\pi E$
(theorem 4.2 in \cite{KenKurMalMug_ahp16}). Having corollary \ref{cor:stower_supremum},
it is possible to prove the following improved bound.
\begin{cor}
\label{cor:upper_bound_for_supremum}Let $\disgraph$ be a graph with
$E$ edges, out of which $E_{l}$ are leaves. Then 
\begin{equation}
\forall~\lenvec\in\len_{\disgraph},~~\;\:k_{1}\left[\metgraph\left(\disgraph;~\lenvec\right)\right]\leq\pi\left(E-\frac{E_{l}}{2}\right),\label{eq:global_upper_bound}
\end{equation}
provided that $(E,E_{l})\notin\left\{ \left(1,1\right),\,\left(1,0\right),\,\mbox{\ensuremath{\left(2,1\right)}}\right\} $.

Assume in addition that $(E,E_{l})\notin\left\{ \left(2,0\right),\,\left(3,2\right)\right\} $.
Then an equality above implies that the graph $\metgraph(\disgraph;~\lenvec)$
achieving the inequality is either an equilateral mandarin or an equilateral
stower.
\end{cor}
This latter bound is sharp as it is attained by most equilateral stower
graphs (see Example \ref{exa:stower} and Corollary \ref{cor:stower_supremum}).

\section{Infimizers \label{sec:Infimizers}}
\begin{proof}
[Proof of Theorem \ref{thm:infimizers}] Let $\metgraph$ be a metric
graph whose total edge length equals one and let $f$ be an eigenfunction
corresponding to the spectral gap $k_{1}(\metgraph)$ and normalized
such that its $L^{2}$ norm equals one. Denote
\begin{align}
m & :=\min f<0\\
M & :=\max f>0,
\end{align}
where the inequalities arise as $f$, being a Neumann eigenfunction
is orthogonal to the constant function. In what follows we bound from
below the Rayleigh quotient of $f$ by using the rearrangement technique
in a similar manner to the proof of lemma 3 in \cite{Friedlander_aif05}.
We further define
\[
\mu_{f}\left(t\right):=\left|\left\{ \left.x\in\metgraph~\right|~f\left(x\right)<t\right\} \right|~~~~\textrm{for }~t\in\left[m,M\right]
\]
where $\left|\cdot\right|$ denotes the Lebesgue measure of the corresponding
set on the graph. This allows to define a continuous, non-decreasing
function $f^{*}$ on the interval $\left[0,1\right]$, such that $\mu_{f^{*}}=\mu_{f}$.
This property gives 
\begin{equation}
1=\int_{\metgraph}\left|f\left(x\right)\right|^{2}\textrm{d}x=\int_{m}^{M}t^{2}\textrm{d}\mu_{f}=\int_{0}^{1}\left|f^{*}\left(x\right)\right|^{2}\textrm{d}x\label{eq:thm_infimizers_1}
\end{equation}

and

\begin{equation}
0=\int_{\metgraph}f\left(x\right)\textrm{d}x=\int_{m}^{M}t\textrm{d}\mu_{f}=\int_{0}^{1}f^{*}\left(x\right)\textrm{d}x,\label{eq:thm_infimizers_2}
\end{equation}
where the first equality in \eqref{eq:thm_infimizers_2} holds since
$f$ is orthogonal to the constant function. 

Another ingredient we use in the proof is the co-area formula \cite{Chavel_RiemannianGeometry}.
Let $t\in\left[m,M\right]$ such that if $f\left(x\right)=t$ then
$x$ is not a vertex and $f'\left(x\right)\neq0$ and call this $t$
a regular value. By Sard's theorem, the non-regular values are of
zero measure. According to the co-area formula if $t$ is a regular
value then 
\begin{equation}
\mu_{f}'\left(t\right)=\sum_{x~;~f\left(x\right)=t}\frac{1}{\left|f'\left(x\right)\right|},\label{eq:thm_infimizers_3}
\end{equation}
and for any $L^{1}$ function $g$ on the graph
\begin{equation}
\int_{\metgraph}g\left(x\right)\left|f'\left(x\right)\right|\textrm{d}x=\int_{m}^{M}\left(\sum_{x~;~f\left(x\right)=t}g\left(t\right)\right)\textrm{d}t.\label{eq:thm_infimizers_4}
\end{equation}

We now estimate the numerator of the Rayleigh quotient, $\int_{\metgraph}\left|f'\left(x\right)\right|^{2}\textrm{d}x$,
as follows. Denote by $x_{m},x_{M}$ two points for which $f\left(x_{m}\right)=m,~f\left(x_{_{M}}\right)=M$
(they are not necessarily unique). Let $t\in\left[m,M\right]$ be
a regular value. As $\metgraph$ is connected there is a path on the
graph connecting $x_{m}$ with $x_{M}$ and by continuity of $f$
it attains the value $t$ at least once along this path, say at some
point $x_{t}$. By the choice of $t$, $x_{t}$ is not a vertex. If
$\metgraph$ is a bridgeless graph, then cutting the graph at $x_{t}$,
the graph is still connected and we can find another path joining
$x_{m}$ and $x_{M}$. By the same reasoning $f$ attains the value
$t$ along this path as well, so that $t$ is attained by $f$ at
least twice on $\metgraph$. Denoting by $n\left(t\right)$ the number
of times that the value $t$ is attained by $f$ on the graph, we
get that 
\begin{equation}
n\left(t\right)\geq\begin{cases}
1 & \textrm{if ~}\metgraph\textrm{~has a bridge,}\\
2 & \textrm{if ~}\metgraph\textrm{~is bridgeless.}
\end{cases}\label{eq:thm_infimizers_n_of_t}
\end{equation}
We may also bound $n\left(t\right)$ from above 
\begin{align}
\left(n\left(t\right)\right)^{2} & =\left(\sum_{x~;~f\left(x\right)=t}\frac{1}{\sqrt{\left|f'\left(x\right)\right|}}\sqrt{\left|f'\left(x\right)\right|}\right)^{2}\\
 & \leq\left(\sum_{x~;~f\left(x\right)=t}\frac{1}{\left|f'\left(x\right)\right|}\right)\left(\sum_{x~;~f\left(x\right)=t}\left|f'\left(x\right)\right|\right)\\
 & =\mu_{f}'\left(t\right)\left(\sum_{x~;~f\left(x\right)=t}\left|f'\left(x\right)\right|\right),\label{eq:thm_infimizers_5}
\end{align}
by applying the Cauchy-Schwarz inequality and \eqref{eq:thm_infimizers_3}.
Writing \eqref{eq:thm_infimizers_4} with $g\left(x\right)=\left|f'\left(x\right)\right|$
gives

\begin{equation}
\int_{\metgraph}\left|f'\left(x\right)\right|^{2}\textrm{d}x=\int_{m}^{M}\left(\sum_{x~;~f\left(x\right)=t}\left|f'\left(x\right)\right|\right)\textrm{d}t\geq\int_{m}^{M}\frac{\left(n\left(t\right)\right)^{2}}{\mu_{f}'\left(t\right)}\textrm{d}t.\label{eq:thm_infimizers_6}
\end{equation}

We may repeat the arguments above for $f^{*}$, which attains each
regular value exactly once and obtain that \eqref{eq:thm_infimizers_5},\eqref{eq:thm_infimizers_6}
hold for $f^{*}$ as equalities and with $n^{*}\left(t\right)=1$.
Therefore 
\begin{equation}
\int_{\metgraph}\left|f'\left(x\right)\right|^{2}\textrm{d}x\geq\underset{m\leq t\leq M}{\textrm{ess~inf}}\left(n\left(t\right)\right)^{2}\int_{\metgraph}\left|\left(f^{*}\right)'\left(x\right)\right|^{2}\textrm{d}x,\label{eq:thm_infimizers_7}
\end{equation}
where the infimum above is taken only with respect to regular values.
As $f$ is the eigenfunction corresponding to $k_{1}(\metgraph)$
with unit $L^{2}$ norm we have $\int_{\metgraph}\left|f'\left(x\right)\right|^{2}\textrm{d}x=\left(k_{1}(\metgraph)\right)^{2}$.
Considering $f^{*}$ as a test function of unit $L^{2}$ norm (see
\eqref{eq:thm_infimizers_1}) and zero mean (see \eqref{eq:thm_infimizers_2})
on the unit interval we get that its Rayleigh quotient is no less
than the first positive eigenvalue, namely that $\int_{\metgraph}\left|\left(f^{*}\right)'\left(x\right)\right|^{2}\textrm{d}x\geq\pi^{2}$.
Combining this with \eqref{eq:thm_infimizers_7} and \eqref{eq:thm_infimizers_n_of_t}
we get the lower bounds, 
\begin{equation}
k_{1}(\metgraph)\geq\begin{cases}
\pi & \textrm{if ~}\metgraph\textrm{~has a bridge,}\\
2\pi & \textrm{if ~}\metgraph\textrm{~is bridgeless.}
\end{cases}\label{eq:thm_infimizers_8}
\end{equation}
All that remains to complete the proof is the characterization of
the infimizers.

Assume first that $\metgraph$ has a bridge. An equality in \eqref{eq:thm_infimizers_8}
is possible only if $n\left(t\right)=1$ for all regular $t\in\left[m,M\right]$.
This implies that $\metgraph$ does not have vertices of degree $3$
and above. Otherwise, due to continuity of $f$, we would have $n\neq1$
in the vicinity of such a vertex. $\metgraph$ cannot be a single
cycle graph as it has a bridge and is therefore the unit interval,
$\left[0,1\right]$. Hence it is the unique candidate for an infimizer.
Indeed, its spectral gap is $\pi$ and starting from any discrete
graph $\disgraph$ with a bridge, $\metgraph(\disgraph;~\lenvec)$
is the unit interval if $\lenvec\in\lencl_{\disgraph}$ is chosen
such that all of its entries vanish, except the entry corresponding
to the bridge.

Next, the possible minimizers of bridgeless graphs are characterized.
By Menger's theorem \cite{Menger_FM27}, a graph is bridgeless if
and only if there are at least two edge disjoint paths connecting
any pair of points. We use that to deduce that if $\disgraph$ is
bridgeless then $\metgraph(\disgraph;~\lenvec)$ is bridgeless as
well. Indeed, any path between a pair of points in $\metgraph(\disgraph;~\lenvec)$
corresponds to at least one path between those points in $\disgraph$.
Thus, to seek for a possible minimizer, we assume that $\metgraph$
is bridgeless and $k_{1}\left(\metgraph\right)=2\pi$. As a bridgeless
graph is $2$-edge-connected, we deduce from Menger's theorem that
there are at least two edge disjoint paths connecting $x_{m}$ with
$x_{M}$. Pick two such paths and denote them by $\gamma_{1},\gamma_{2}$.
A necessary condition for $k_{1}(\metgraph)=2\pi$ is that $n\left(t\right)=2$
for each regular value $t\in\left[m,M\right]$. By continuity, $f$
attains each regular value at least once on $\gamma_{1}$ and at least
once on $\gamma_{2}$. As $n\left(t\right)=2$ for a regular value
$t$, $f$ attains the value $t$ \emph{exactly} once on each of $\gamma_{1}$
and $\gamma_{2}$. Hence $f$ is strictly increasing on $\gamma_{1}$
from $x_{m}$ to $x_{M}$ and the same holds for $\gamma_{2}$. We
further conclude that $f$ may attain only non-regular values at $\metgraph\backslash\left\{ \gamma_{1}\cup\gamma_{2}\right\} $.
In particular, if there exists an edge in $\metgraph\backslash\left\{ \gamma_{1}\cup\gamma_{2}\right\} $,
$f$ should be constant on that edge and due to $-f''=\left(2\pi\right)^{2}f$
this constant equals zero. Thus, the edges of $\metgraph\backslash\left\{ \gamma_{1}\cup\gamma_{2}\right\} $
may be removed from $\metgraph$, such that $f$ still satisfies the
Neumann conditions on the remaining graph $\gamma_{1}\cup\gamma_{2}$
and it is an eigenfunction on that graph. However, by this we find
an eigenfunction of $k$-eigenvalue $2\pi$ on a bridgeless graph
whose total length smaller than one, which contradicts the lower bound,
\eqref{eq:thm_infimizers_8}. Hence $\metgraph$ consists of just
the union of the paths $\gamma_{1},\gamma_{2}$. As $\gamma_{1},\gamma_{2}$
are edge disjoint, $\gamma_{1}\cap\gamma_{2}$ contains only vertices.
We denote those vertices by $v_{0},\ldots,v_{n}$, with $v_{0}=x_{m},~v_{n}=x_{M}$
and the indices are arranged in an increasing order along the path
$\gamma_{1}$. As $f$ is strictly increasing along both $\gamma_{1},\gamma_{2}$,
the order of those vertices along $\gamma_{2}$ is the same: $v_{0},\ldots,v_{n}$.
Consider two adjacent vertices $v_{i},v_{i+1}$ ($0\leq i\leq n-1$)
and denote the corresponding path segments connecting them by $\gamma_{1}\left(v_{i},v_{i+1}\right)$,$\gamma_{2}\left(v_{i},v_{i+1}\right)$.
As $f$ takes the same values on the endpoints of $\gamma_{1}\left(v_{i},v_{i+1}\right)$,$\gamma_{2}\left(v_{i},v_{i+1}\right)$,
is increasing and satisfies $-f''=\left(2\pi\right)^{2}f$ on both,
we conclude $\left.f\right|_{\gamma_{1}\left(v_{i},v_{i+1}\right)}=\left.f\right|_{\gamma_{2}\left(v_{i},v_{i+1}\right)}$
and also that $\gamma_{1}\left(v_{i},v_{i+1}\right)$ has the same
length as $\gamma_{2}\left(v_{i},v_{i+1}\right)$. Hence $\metgraph=\gamma_{1}\cup\gamma_{2}$
is a symmetric necklace.
\end{proof}
\begin{rem*}
A further exploration of symmetric necklace graphs appears in Proposition
\ref{prop:Standarins}. It is shown there that a symmetric necklace
graph belongs to a family of graphs in which every graph has a simple
spectral gap and its spectral gap $k_{1}\left[\metgraph(\disgraph;\lenvec)\right]$
is a critical value when considered as a function of $\lenvec\in\len_{\disgraph}$.
\end{rem*}
Theorem \ref{thm:infimizers} provides a complete answer to the minimization
problem. In particular, it states that any infimizer of a bridgeless
graph is a symmetric necklace. A further task would be to classify
the entire family of necklace graphs which serve as infimizers of
a particular discrete graph. We start treating this by observing that
the spectral gap of any symmetric necklace (of total length one) is
$2\pi$. This follows from noting that $2\pi$ is an eigenvalue of
any symmetric necklace and combining this with Theorem \ref{thm:infimizers}.
Now, let $\disgraph$ be a bridgeless graph and let $\lenvec^{*}\in\len_{\disgraph}$,
such that $\metgraph(\disgraph;~\lenvec^{*})$ is a symmetric necklace
with some $\beta$ number of cycles. By the observation above and
Theorem \ref{thm:infimizers} we have that $\metgraph(\disgraph;~\lenvec^{*})$
is an infimizer of $\disgraph$. Furthermore, by choosing other values
for $\lenvec\in\len_{\disgraph}$ we may get $\metgraph(\disgraph;~\lenvec)$
to be any symmetric necklace with at most $\beta$ cycles, and from
the above this $\metgraph(\disgraph;~\lenvec)$ would also serve as
an infimizer. Therefore, the answer to the classification problem
above would be given once we find what is the maximal number of cycles
among all symmetric necklaces that can be obtained from a given discrete
graph $\disgraph$. Solving this requires some elements from the theory
of graph connectivity which we shortly present below. A graph is called
$k$-edge-connected if it remains connected whenever less than $k$
edges are removed. In particular, a bridgeless graph is $2$-edge-connected.
A cactus graph is a graph in which every edge is contained in exactly
one cycle. Let $\disgraph$ be a bridgeless graph. There exists $\lenvec\in\len_{\disgraph}$
such that $\metgraph(\disgraph;~\lenvec)$ is a cactus graph with
the following property. For every two edges $e,e'$ which form a 2-edge-cut
in $\disgraph$ (two edges whose removal disconnects the graph), we
have $l_{e},l_{e'}\neq0$. Namely, those two edges also appear in
$\metgraph(\disgraph;~\lenvec)$. The theory leading to this result
appears in \cite{DinKarLom_sdm76,FleFra_report09,NagamochiIbaraki_book08}
for general $k$-connected graphs and is very nicely explained for
the particular case of $2$-edge-connected graphs in section 10 of
the recent paper \cite{MehNeuSch_alg17}. Now, in order to determine
the maximal number of cycles of a necklace obtained from $\disgraph$
we perform the following procedure. Find all subgraphs of $\disgraph$
which are $3$-edge-connected and contract each of them to a vertex;
for example by choosing $\lenvec\in\len_{\disgraph}$ such that the
corresponding entries vanish and considering $\metgraph(\disgraph;~\lenvec)$.
This yields a cactus graph with the property mentioned above \cite{MehNeuSch_alg17}.
The cactus graph has a tree-like structure. This can be observed by
considering an auxiliary graph $\metgraph'$, where each cycle of
$\metgraph(\disgraph;~\lenvec)$ is represented by a vertex of $\metgraph'$
and two vertices of $\metgraph'$ are connected if the corresponding
cycles in $\metgraph(\disgraph;~\lenvec)$ share a vertex (a cactus
graph has the property that any two cycles of it, share at most one
vertex). The obtained graph, $\metgraph'$ turns to be a tree graph.
Any path of this tree graph then corresponds to a necklace which can
be obtained from the cactus $\metgraph(\disgraph;~\lenvec)$ by further
setting some edge lengths to zero. The longest possible necklace is
found by identifying the longest path of the tree $\metgraph'$.

\section{Supremizers of tree graphs \label{sec:TreeGraphs}}

The proof of Theorem \ref{thm:tree_supremizer} is based on bounding
the graph diameter, as follows.
\begin{defn}
\label{def:diameter}Let $\Gamma$ be a compact metric graph. The
diameter of $\Gamma$ is 
\[
d(\Gamma):=\max\left\{ \left.\textrm{dist}\left(x,y\right)~\right|~x,y\in\metgraph\right\} 
\]
\end{defn}
\begin{lem}
\label{lem:tree_diameter_lower_bound} Let $\Gamma$ be a metric tree
graph of total length $1$ and with $E_{l}\geq2$ leaves. Then 
\begin{equation}
d(\Gamma)\geq\frac{2}{E_{l}}\label{eq:lem_tree_diameter_lower_bound}
\end{equation}
with equality if and only if $\Gamma$ is an equilateral star.
\end{lem}
\begin{proof}
Choose two points, $x_{1},x_{2}$, in $\Gamma$ such that the distance
between them is exactly $d(\Gamma)$. We show that $x_{1},x_{2}$
are necessarily leaves. Assume by contradiction that (w.l.o.g) $x_{1}$
is not a leaf. Then $\Gamma\setminus\{x_{1}\}$ has at least two connected
components. Let $\Gamma_{1}$ be one of these components satisfying
$x_{2}\not\in\Gamma_{1}$. Let $z$ be a point of $\Gamma_{1}$ different
from $x_{1}$. As $\metgraph$ is a tree, any path from $z$ to $x_{2}$
contains $x_{1}$, which yields 
\[
d(x_{2},z)>d(x_{2},x_{1})=d\left(\metgraph\right),
\]
thus contradicting the definition of $d(\Gamma)$. Let now $\mathcal{P}$
be the shortest path connecting $x_{1}$ to $x_{2}$ and denote by
$x_{0}$ its middle, such that 
\[
d(x_{1},x_{0})=d(x_{2},x_{0})=\frac{d(\Gamma)}{2}.
\]
We cover $\Gamma$ with $E_{l}$ paths, each starting at $x_{0}$
and ending at a leaf of $\metgraph$. The length of each of these
paths is at most $d(x_{1},x_{0})$ (otherwise, we may replace $x_{1}$
by a different leaf and increase $d\left(\metgraph\right)$). As the
union of these paths cover $\Gamma$, whose total length is $1$,
we have 
\begin{equation}
1\leq\sum_{v\text{ is a leaf}}d(x_{0},v)\leq\sum_{v\text{ is a leaf}}d(x_{0},x_{1})=E_{l}\frac{d(\Gamma)}{2},
\end{equation}
from which the inequality of the lemma follows. The first inequality
can be an equality if and only if $\Gamma$ is a star and $x_{0}$
is its central vertex. Assuming this, the second inequality can be
an equality if an only if the star is equilateral.
\end{proof}
Aided with Lemma \ref{lem:tree_diameter_lower_bound}, we turn to
the proof of the theorem.
\begin{proof}
[Proof of Theorem \ref{thm:tree_supremizer}] We show in the following
that there exists a test function $f$ on $\Gamma$ such that its
Rayleigh quotient satisfies 
\begin{equation}
\mathcal{R}(f)\leq\left(\frac{\pi}{d(\Gamma)}\right)^{2}.\label{RayleighTree}
\end{equation}
Indeed, let $y,z$ be two leaves of $\Gamma$ such that the distance
between them is exactly $d(\Gamma)$. Let us denote by $\mathcal{P}$
a path of $\Gamma$, of length $d(\Gamma)$, connecting $y$ and $z$.
We consider $\mathcal{P}$ as the interval $[0,d(\Gamma)]$, for example
by identifying $y$ with $0$ and $z$ with $d(\Gamma)$ and define
the following function on $\mathcal{P}$, 
\[
f(x)=\cos\left(\frac{\pi x}{d(\Gamma)}\right)\text{ for \ensuremath{x\in\mathcal{P}}.}
\]
We extend $f$ to be defined on the whole graph, $\metgraph$, by
setting its value on each connected component of $\Gamma\setminus\mathcal{P}$
to the unique constant which preserves the continuity of $f$. Referring
to Appendix \ref{sec:appendix_test_functions_using_subgraphs} and
using $f-\left\langle f\right\rangle $ as our test function we have
from \eqref{eq:rayleigh_with_mean}, 
\begin{align}
\mathcal{R}\left(f-\left\langle f\right\rangle \right) & =\frac{\int_{\Gamma}|f'(x)|^{2}dx}{\int_{\Gamma}|f(x)|^{2}dx-\left(\int_{\Gamma}f(x)dx\right)^{2}}\\
 & =\frac{\left(\frac{\pi}{d(\Gamma)}\right)^{2}\frac{d(\Gamma)}{2}}{\frac{d(\Gamma)}{2}+\int_{\Gamma\setminus\mathcal{P}}|f(x)|^{2}dx-\left(\int_{\Gamma}f(x)dx\right)^{2}}\label{eq:thm_tree_supremum_eq_1}
\end{align}

As the integral of $f$ on $\mathcal{P}$ vanishes, using Cauchy-Schwarz
inequality we get 
\begin{equation}
\left(\int_{\Gamma}f(x)dx\right)^{2}=\left(\int_{\Gamma\setminus\mathcal{P}}f(x)dx\right)^{2}\leq(1-d(\Gamma))\int_{\Gamma\setminus\mathcal{P}}|f(x)|^{2}dx.\label{eq:thm_tree_supremum_eq_2}
\end{equation}

Plugging \eqref{eq:thm_tree_supremum_eq_2} in \eqref{eq:thm_tree_supremum_eq_1}
gives
\begin{equation}
\mathcal{R}\left(f-\left\langle f\right\rangle \right)\leq\frac{\left(\frac{\pi}{d(\Gamma)}\right)^{2}\frac{d(\Gamma)}{2}}{\frac{d(\Gamma)}{2}+d\left(\metgraph\right)\int_{\Gamma\setminus\mathcal{P}}|f(x)|^{2}dx}\leq\left(\frac{\pi}{d(\Gamma)}\right)^{2}.\label{eq:thm_tree_supremum_eq_3}
\end{equation}
Using this and Lemma \ref{lem:tree_diameter_lower_bound} we get 
\begin{equation}
k_{1}\left(\metgraph\right)\leq\frac{\pi}{d(\Gamma)}\leq\frac{\pi}{2}E_{l}.\label{eq:thm_tree_supremum_eq_4}
\end{equation}
Let $\disgraph$ be a tree graph with $E_{l}$ leaves. We may choose
$\lenvec\in\lencl_{\disgraph}$ such that $\metgraph(\disgraph;~\lenvec)$
is an equilateral star graph with $E_{l}$ leaves, so that $k_{1}\left[\metgraph(\disgraph;~\lenvec)\right]=\frac{\pi}{2}E_{l}$
and from the bound above we get that $\metgraph(\disgraph;~\lenvec)$
is a supremizer. This is a unique supremizer as having equality in
the right inequality of \eqref{eq:thm_tree_supremum_eq_4} implies
by Lemma \ref{lem:tree_diameter_lower_bound} that $\metgraph$ is
an equilateral star with $E_{l}$ leaves.
\end{proof}
\begin{rem*}
We note that the upper bound $k_{1}\left(\metgraph\right)\leq\frac{\pi}{d(\Gamma)}$,
which is obtained in the course of the proof above, is a particular
case of a result proven recently in \cite{Rohleder_pams16}. There
it was shown that for any $n$, $k_{n}\left(\metgraph\right)\leq\frac{\pi n}{d(\Gamma)}$.
Applying \eqref{eq:lem_tree_diameter_lower_bound} to the latter we
may get that for any $n\geq1$, $k_{n}\left(\metgraph\right)\leq\frac{\pi n}{2}E_{l}$,
which improves the bound $k_{n}\left(\metgraph\right)\leq\frac{\pi n}{2}E$
given in \cite{Rohleder_pams16}.
\end{rem*}
The theorem above yields the following.
\begin{proof}
[Proof of Corollary \ref{cor:non_tree_supremizer}] Let $\disgraph$
be a graph with $\beta>0$ cycles and $E_{l}$ leaves. We start by
observing that for $\left(\beta,E_{l}\right)\in\left\{ \left(1,0\right),\left(1,1\right)\right\} $,
the supremizer is the single cycle graph (see Lemma \ref{lem:stower_1_1}),
which is not a tree. We continue assuming $\left(\beta,E_{l}\right)\notin\left\{ \left(1,0\right),\left(1,1\right)\right\} $.
Choose a maximal spanning tree of $\disgraph\backslash\mathcal{E}_{l}$,
where $\mathcal{E}_{l}$ is the set of the graph's $E_{l}$ leaves.
Choose $\lenvec^{*}\in\lencl_{\disgraph}$ such that all of its entries
corresponding to the spanning tree edges are set to zero. This makes
$\metgraph(\disgraph;~\lenvec^{*})$ a stower with $\beta$ petals
and $E_{l}$ leaves. Furthermore, $\lenvec^{*}$ may be chosen such
that $\metgraph(\disgraph;~\lenvec^{*})$ is an equilateral stower.
The spectral gap of this graph is $\frac{\pi}{2}\left(2\beta+E_{l}\right)$
(see Example \ref{exa:stower}). Alternatively, if $\lenvec\in\lencl_{\disgraph}$
is such that $\metgraph(\disgraph;~\lenvec)$ is a tree then the number
of its leaves is at most $E_{l}$ and by Theorem \ref{thm:tree_supremizer}
its spectral gap is at most $\frac{\pi}{2}E_{l}$. Therefore, the
stower graph $\metgraph(\disgraph;~\lenvec^{*})$ obtained above has
a greater spectral gap than any tree graph $\metgraph(\disgraph;~\lenvec)$.
\end{proof}

\section{Spectral gaps as critical values\label{sec:LocalOptimizers}}

In this section we assume that the spectral gap, $k_{1}\left(\metgraph\left(\disgraph;~\lenvec\right)\right)$,
is a simple eigenvalue. This allows to take derivatives of the eigenvalue
with respect to the edge lengths, $\lenvec\in\len_{\disgraph}$, and
to find critical points which serve as candidates for maximizers.
We prove here Theorem \ref{thm:simple_supremum} which shows that
such local maximizers do not achieve a spectral gap higher than that
achieved by turning the graph into a mandarin or a flower.
\begin{lem}
\label{lem:constant_energy} Let $\metgraph$ be a metric graph and
$f$ an eigenfunction corresponding to the eigenvalue $k^{2}$ with
arbitrary vertex conditions. Then the function $f'(x)^{2}+k^{2}f(x)^{2}$
is constant along each edge.
\end{lem}
\begin{proof}
The proof is immediate by differentiating the function $f'(x)^{2}+k^{2}f(x)^{2}$
along an edge.
\end{proof}
The last lemma motivates us to define the energy\footnote{A simple harmonic oscillator whose spring constant is $k$ and whose
position is given by $f(x)$ has a total energy of $\frac{1}{2}\energy_{e}$.} of an eigenfunction on an edge $e$ as $\energy_{e}:=f'(x)^{2}+k^{2}f(x)^{2}$
for any $x\in e$. This energy shows up naturally when differentiating
an eigenvalue with respect to an edge length. In order to evaluate
such derivatives we extend Definition \ref{def:discrete_and_metric_graph}
so that $\metgraph\left(\disgraph;~\lenvec\right)$ is defined for
all $\lenvec\in\R^{E}$ with positive entries and relax the restriction
$\sum_{e=1}^{E}l_{e}=1$, imposed by $\lenvec\in\len_{\disgraph}$.
The following lemma appears also as Lemma A.1 in \cite{CdV_ahp14}
and within the proof of a lemma in \cite{Friedlander_ijm05}.
\begin{lem}
\label{lem:Derivative_on_edge_equals_energy} Let $\disgraph$ be
a discrete graph and let $\lenvec\in\R^{E}$ with positive entries.
Assume that the spectral gap, $k_{1}\left[\metgraph(\disgraph;~\lenvec)\right]$
is a simple eigenvalue and let $f$ be the corresponding eigenfunction,
normalized to have unit $L^{2}$ norm. Then $k_{1}\left[\metgraph(\disgraph;~\lenvec)\right]$
is differentiable with respect to any edge length $l_{\tilde{e}}$
and 
\begin{equation}
\frac{\partial}{\partial l_{\tilde{e}}}\left(\left(k_{1}\left[\metgraph\left(\disgraph;~\lenvec\right)\right]\right)^{2}\right)=-\energy_{\tilde{e}}.
\end{equation}
 
\end{lem}
\begin{proof}
In this proof we use the analyticity of the eigenvalues and eigenfunctions
with respect to the edge lengths. This is established for example
in sections 3.1.2, 3.1.3 of \cite{BerKuc_quantum_graphs}. Let $s\in\R$
and let $\tilde{e}$ be an edge of $\metgraph(\disgraph;~\lenvec)$.
Denote $\lenvec\left(s\right):=\lenvec+s\vec{e}$, with $\vec{e}\in\R^{E}$
a vector with one at its $\tilde{e}^{\textrm{th}}$ position and zeros
in all other entries. We use the notation $\metgraph\left(s\right):=\metgraph(\disgraph;~\lenvec\left(s\right))$
and denote by $k_{1}\left(s\right)$ the spectral gap of $\metgraph\left(s\right)$.
By assumption, $k_{1}\left(0\right)$ is a simple eigenvalue and hence
there is a neighborhood of zero for which all $k_{1}\left(s\right)$
are simple eigenvalues. The corresponding eigenfunctions are denoted
by $f\left(s;~\cdot\right)$ and we further assume that all those
eigenfunctions have unit $L^{2}$ norm, 
\begin{equation}
\int_{\Gamma(s)}\left(f\left(s;~x\right)\right){}^{2}dx=\sum_{e=1}^{E}\int_{0}^{l_{e}(s)}\left(f\left(s;~x_{e}\right)\right){}^{2}dx_{e}=1,\label{eq:derivative_equals_energy_eq_1}
\end{equation}
where $l_{e}(s)=l_{e}+\delta_{e,\tilde{e}}s$ and $\delta_{e,\tilde{e}}$
being the Kronecker delta function.

Taking a derivative of the above with respect to $s$, 
\begin{equation}
\left(f\left(s;~l_{\tilde{e}}\left(s\right)\right)\right){}^{2}+2\sum_{e=1}^{E}\int_{0}^{l_{e}(s)}f(s;~x_{e})\frac{\partial}{\partial s}f(s;~x_{e})dx_{e}=0.\label{eq:derivative_equals_energy_eq_2}
\end{equation}

In addition, evaluating the Rayleigh quotient of $f$, 
\begin{equation}
k_{1}\left(s\right){}^{2}=\ray\left[f\left(s;\cdot\right)\right]=\sum_{e=1}^{E}\int_{0}^{l_{e}(s)}\left(\frac{\partial}{\partial x_{e}}f\left(s;~x_{e}\right)\right)^{2}dx_{e},\label{eq:derivative_equals_energy_eq_3}
\end{equation}
using that $f\left(s;\cdot\right)$ has unit norm. Differentiating
this with respect to $s$ gives 
\begin{equation}
\frac{d}{ds}\left(k_{1}\left(s\right){}^{2}\right)=\left(\frac{\partial}{\partial x_{\tilde{e}}}f\left(s;~l_{\tilde{e}}\left(s\right)\right)\right)^{2}+2\sum_{e=1}^{E}\int_{0}^{l_{e}(s)}\frac{\partial}{\partial x_{e}}f\left(s;~x_{e}\right)\frac{\partial^{2}}{\partial s\partial x_{e}}f\left(s;~x_{e}\right)dx_{e}.\label{eq:derivative_equals_energy_eq_4}
\end{equation}
Integrating by parts in the right hand side and using the eigenvalue
equation, we get for each term in the sum above 
\begin{multline}
\int_{0}^{l_{e}(s)}\frac{\partial}{\partial x_{e}}f\left(s;~x_{e}\right)\frac{\partial^{2}}{\partial s\partial x_{e}}f\left(s;~x_{e}\right)dx_{e}=\\
~~~~~~~~\frac{\partial}{\partial x_{e}}f\left(s;~l_{e}\left(s\right)\right)\left(\frac{\partial}{\partial s}f\right)\left(s;~l_{e}(s)\right)-\frac{\partial}{\partial x_{e}}f\left(s;~0\right)\frac{\partial}{\partial s}f\left(s;~0\right)\\
+k_{1}(s)^{2}\int_{0}^{l_{e}(s)}f\left(s;~x_{e}\right)\frac{\partial}{\partial s}f\left(s;~x_{e}\right)dx_{e}\\
=\left.\frac{\partial f}{\partial x_{e}}\left(\frac{df}{ds}-\delta_{e,\tilde{e}}\frac{\partial f}{\partial x_{e}}\right)\right|_{\left(s;~l_{e}(s)\right)}-\left.\frac{\partial f}{\partial x_{e}}\frac{df}{ds}\right|_{\left(s;~0\right)}+k_{1}(s)^{2}\int_{0}^{l_{e}(s)}\left.f\frac{\partial f}{\partial s}\right|_{\left(s;~x_{e}\right)}dx_{e},\label{eq:derivative_equals_energy_eq_5}
\end{multline}
where the partial derivatives with respect to $s$ are rewritten in
terms of complete derivatives.

Summing the first two terms of the right hand side of \eqref{eq:derivative_equals_energy_eq_5}
over all edges and rewriting it as a sum over all graph vertices we
get
\begin{align}
\sum_{e=1}^{E}\left\{ \left.\frac{\partial f}{\partial x_{e}}f\left(\frac{df}{ds}-\delta_{e,\tilde{e}}\frac{\partial f}{\partial x_{e}}\right)\right|_{\left(s;~l_{e}(s)\right)}-\left.\frac{\partial f}{\partial x_{e}}\frac{df}{ds}\right|_{\left(s;~0\right)}\right\} \nonumber \\
=\sum_{v}\left.\left(\sum_{e\sim v}\frac{\partial f}{\partial x_{e}}\right)\frac{df}{ds}\right|_{\left(s;~v\right)} & -\left(\left.\frac{\partial f}{\partial x_{\tilde{e}}}\right|_{\left(s;~l_{\tilde{e}}(s)\right)}\right)^{2}\nonumber \\
= & -\left(\left.\frac{\partial f}{\partial x_{\tilde{e}}}\right|_{\left(s;~l_{\tilde{e}}(s)\right)}\right)^{2},\label{eq:derivative_equals_energy_eq_7}
\end{align}
where the sum $e\sim v$ above is taken over all edges adjacent to
a chosen vertex $v$, the derivatives $\frac{\partial}{\partial x_{e}}$
in this sum are all taken towards the vertex $v$ and $\sum_{e\sim v}\frac{\partial}{\partial x_{e}}f\left(s;~v\right)=0$,
as $f$ satisfies Neumann conditions at $v$.

Plugging \eqref{eq:derivative_equals_energy_eq_5}, \eqref{eq:derivative_equals_energy_eq_7}
and \eqref{eq:derivative_equals_energy_eq_2} in equation \eqref{eq:derivative_equals_energy_eq_4}
we get
\[
\frac{d}{ds}\left(k_{1}\left(s\right){}^{2}\right)=-\left(\left.\frac{\partial f}{\partial x_{\tilde{e}}}\right|_{\left(s;~l_{\tilde{e}}\left(s\right)\right)}\right)^{2}-\left(k_{1}\left(s\right)\right)^{2}\left(\left.f\right|_{\left(s;~l_{\tilde{e}}\left(s\right)\right)}\right){}^{2}=-\energy_{\tilde{e}},
\]
which finishes the proof once $s=0$ is taken.
\end{proof}
We note that the derivative of an eigenvalue with respect to an edge
length is derived in \cite{DelRoss_amp16} (theorem 4.4) for the general
case of the $p$-Laplacian on a graph. In the case of the $2$-Laplacian,
using Lemma \ref{lem:constant_energy} shows that the integral expression
obtained in \cite{DelRoss_amp16} simplifies to equal $-\energy_{\tilde{e}}$.

The lemma above provides a practical tool for increasing the spectral
gap once the corresponding eigenfunction is known. In order to do
so, one should increase the length of edges with lower energy on the
expense of shortening those with higher energy. In particular, focusing
on a particular vertex, one should increase the lengths of the edges
for which the eigenfunction derivative is the lowest and vice versa.
This method is useful as long as the spectral gap is not a critical
point in the edge length space, $\len_{\disgraph}$. An equilateral
star with an odd number of edges illustrates the importance of simplicity:
though we cannot increase the spectral gap, no eigenfunction on this
graph will have equal energy at all edges. 

The next lemma provides a necessary and sufficient condition for existence
of a critical point in the edge length space, $\len_{\disgraph}$.
\begin{lem}
\label{lem:critical_point_derivative_at_vertices} Let $\disgraph$
be a discrete graph and let $\lenvec^{*}\in\len_{\disgraph}$. Assume
that the spectral gap, $k_{1}\left[\metgraph(\disgraph;~\lenvec^{*})\right]$
is a simple eigenvalue and let $f$ be the corresponding eigenfunction.
The function $k_{1}\left[\metgraph(\disgraph;~\lenvec)\right]$ has
a critical value at $\lenvec=\lenvec^{*}$ if and only if both conditions
below are satisfied 
\end{lem}
\begin{enumerate}
\item \label{enu:lem_critical_point_odd_degree} The derivative of $f$
vanishes at all vertices of odd degree.
\item \label{enu:lem_critical_point_even_degree} The derivative of $f$
satisfy, $\left|\frac{\partial}{\partial x_{e_{1}}}f\left(v\right)\right|=\left|\frac{\partial}{\partial x_{e_{2}}}f\left(v\right)\right|$,
for all edges $e_{1},~e_{2}$ adjacent to a vertex of even degree,
$v$.
\end{enumerate}
\begin{proof}
We first observe that positivity of the spectral gap yields that $k_{1}\left[\metgraph(\disgraph;~\lenvec)\right]$
has a critical point at $\lenvec=\lenvec^{*}$ if and only if $\left(k_{1}\left[\metgraph(\disgraph;~\lenvec)\right]\right)^{2}$
has a critical point there. From Lemma \ref{lem:Derivative_on_edge_equals_energy}
we deduce that a critical point occurs if and only if the corresponding
eigenfunction has equal energies on all graph edges. The last deduction
comes as this is a critical point under the constraint $\sum_{e}l_{e}=1$.
Let $v$ be a graph vertex and $e_{1},~e_{2}$ two edges adjacent
to it. Since $f$ is continuous (i.e., single valued) at $v$ we conclude
\[
\energy_{e}=\energy_{\tilde{e}}~~\Leftrightarrow~~\left(\frac{\partial}{\partial x_{e}}f\left(v\right)\right)^{2}=\left(\frac{\partial}{\partial x_{\tilde{e}}}f\left(v\right)\right)^{2},
\]
which proves the second claim of the lemma. The first claim follows
since the Neumann condition gives that the sum of all derivatives
at $v$ vanishes.
\end{proof}
Obviously, graphs whose spectral gap is a critical point in the space
$\len_{\disgraph}$ serve as good candidates for maximizers. The next
lemma characterizes those graphs and their corresponding eigenfunctions.
\begin{lem}
\label{lem:path_decomposition} Let $\disgraph$ be a discrete graph,
$\lenvec^{*}\in\len_{\disgraph}$ and denote $\metgraph:=\metgraph(\disgraph;~\lenvec^{*})$.
Assume that $k:=k_{1}\left[\metgraph\right]$ is a critical value
and let $f$ be the corresponding eigenfunction. Then we have the
following edge-disjoint decomposition 
\begin{equation}
\metgraph=\bigcup_{i=1}^{P}\mathcal{P}_{i},\label{eq:path_decomposition}
\end{equation}
where
\end{lem}
\begin{enumerate}
\item \label{enu:thm_path_decomposition_1} All $\mathcal{P}_{i}$'s are
graphs which possess an Eulerian path or an Eulerian cycle. Namely,
for each $\mathcal{P}_{i}$ there is a path (either open path or a
cycle), which visits each edge exactly once.
\item \label{enu:thm_path_decomposition_2} Different $\mathcal{P}_{i}$'s
may share only vertices, but not edges.
\item \label{enu:thm_path_decomposition_3} $\left.f\right|_{\mathcal{P}_{i}}$
is a Neumann eigenfunction of $\mathcal{P}_{i}$, whose eigenvalue
equals $k$.
\item \label{enu:thm_path_decomposition_4} Denote by $\mu_{i}$ the number
of zeros of $\left.f\right|_{\mathcal{P}_{i}}$, where each zero at
a vertex of $\mathcal{P}_{i}$ is counted as half the degree of this
vertex in $\mathcal{P}_{i}$. Denoting by $L_{i}$ the metric length
of $\mathcal{P}_{i}$, the following holds 
\[
kL_{i}=\pi\mu_{i}.
\]
\item \label{enu:thm_path_decomposition_5} In addition, 
\begin{equation}
k=\pi\mu,\label{eq:lem_path_decomposition_1}
\end{equation}
where $\mu$ is the number of zeros of $f$ on $\Gamma$, where each
zero at a vertex of $\metgraph$ is counted as half the degree of
this vertex in $\metgraph$.
\end{enumerate}
\begin{proof}
We use the claims of Lemma \ref{lem:critical_point_derivative_at_vertices}
to describe a recursive process, which produces this path decomposition.

\begin{itemize}
\item \begin{flushleft}
Assume first that $\metgraph$ has at least one vertex of odd degree,
$v_{0}$. Take $v_{0}$ to be the starting point of a path $\mathcal{P}$
and add to $\mathcal{P}$ any edge, $e_{0},$ which is adjacent to
$v_{0}$ and the vertex connected at its other end, which we denote
by $v_{1}$. If $v_{1}$ is of even degree we seek for an edge $e_{1}$
connected to $v_{1}$ such that $\left.f'\right|_{e_{1}}\left(v_{1}\right)=-\left.f'\right|_{e_{0}}\left(v_{1}\right)$
(both derivatives are outgoing from $v_{1}$). Such edge exists by
lemma \ref{lem:critical_point_derivative_at_vertices},\eqref{enu:lem_critical_point_even_degree}
and as the sum of derivatives of $f$ at $v$ vanish. Add $e_{1}$
and its other endpoint, $v_{2}$ to $\mathcal{P}$ and repeat the
step above until reaching a vertex of odd degree. Once an odd degree
vertex is reached, we end the construction of $\mathcal{P}$ and continue
recursively to form the next path on $\metgraph\backslash\mathcal{P}$.
Note that a certain vertex may be reached more than once during $\mathcal{P}'s$
construction. Such a vertex would appear in $\mathcal{P}$ only once,
but with a degree greater than two. This process of path constructions
continues until we exhaust the whole of $\metgraph$ or alternatively,
until $\metgraph$ does not have any more odd degree vertices, at
which point we continue with performing the next stage.
\par\end{flushleft}
\item \begin{flushleft}
If $\metgraph$ has no vertex of odd degree, the construction of $\mathcal{P}$
is as follows. We choose an arbitrary vertex, $v_{0}$ as the starting
point of $\mathcal{P}$ and choose an arbitrary edge, $e_{0}$ which
is connected to $v_{0}$ and add it to $\mathcal{P}$ as well, together
with its other endpoint, $v_{1}$. Now, just as we did in the first
stage, we seek for an edge $e_{1}$ connected to $v_{1}$ such that
$\left.f'\right|_{e_{1}}\left(v_{1}\right)=-\left.f'\right|_{e_{0}}\left(v_{1}\right)$.
We keep constructing $\mathcal{P}$ as above, keeping in mind that
all vertices are of even degree. At some point we reach again the
vertex $v_{0}$, arriving from some edge denoted $e_{n}$. If $\left.f'\right|_{e_{0}}\left(v_{0}\right)=-\left.f'\right|_{e_{n}}\left(v_{0}\right)$
(both derivatives are outgoing from $v_{0}$) then we end the construction
of $\mathcal{P}$. Otherwise, continue the construction of $\mathcal{P}$
until the condition above is satisfied. This will indeed occur, as
the graph is finite and $f$ satisfies Neumann conditions on $\metgraph$.
Once we finish constructing of $\mathcal{P}$ we continue recursively
to form the next path on $\metgraph\backslash\mathcal{P}$. 
\par\end{flushleft}

\end{itemize}
By construction, each $\mathcal{P}_{i}$ either possesses an Eulerian
path (first stage above) or an Eulerian cycle (second stage) and $\left.f\right|_{\mathcal{P}_{i}}$
satisfies Neumann conditions on $\mathcal{P}_{i}$. Thus claims \eqref{enu:thm_path_decomposition_1}
and \eqref{enu:thm_path_decomposition_3} are valid. Also, as each
subgraph $\mathcal{P}_{i}$ is removed from $\metgraph$ once constructed,
it is clear that $\forall i\neq j,~~\mathcal{P}_{i}\cap\mathcal{P}_{j}$
may contain only vertices, which is stated in claim \eqref{enu:thm_path_decomposition_2}.
A subgraph $\mathcal{P}_{i}$ of the first stage of the construction,
where $\metgraph$ has some odd degree vertices, possesses an Eulerian
path and may be identified with an interval $\left[0,L_{i}\right]$,
where $L_{i}$ is the metric length of $\mathcal{P}_{i}$. Also by
way of construction, $\left.f\right|_{\left[0,L_{i}\right]}$ is a
Neumann eigenfunction (notice that this is more restrictive than stating
that $\left.f\right|_{\mathcal{P}_{i}}$ is a Neumann eigenfunction,
because of possible self-crossings). Hence $\left.f\right|_{\left[0,L_{i}\right]}=\cos\left(\frac{\pi}{L_{i}}\mu_{i}x\right)$
for some positive integer, $\mu_{i}$. Clearly, $\mu_{i}$ equals
the number of zeros of $\left.f\right|_{\left[0,L_{i}\right]}$. Furthermore,
$\mu_{i}$ also equals the number of zeros of $\left.f\right|_{\mathcal{P}_{i}}$,
where a zero at a vertex is counted as many times as half the degree
of that vertex in $\mathcal{P}_{i}$. A subgraph $\mathcal{P}_{i}$
of the second construction stage, where all $\metgraph$ vertices
are of even degrees possesses an Eulerian cycle and may be identified
with an interval $\left[0,L_{i}\right]$, where $L_{i}$ is the metric
length of $\mathcal{P}_{i}$. Also by way of construction, $\left.f\right|_{\left[0,L_{i}\right]}$
is a Neumann eigenfunction which satisfies periodic boundary conditions.
Hence $\left.f\right|_{\left[0,L_{i}\right]}=\cos\left(\frac{\pi}{L_{i}}\mu_{i}x\right)$
for some positive even integer, $\mu_{i}$. As before, $\mu_{i}$
equals the number of zeros of $\left.f\right|_{\mathcal{P}_{i}}$,
counted according to vertex degrees. In both cases, we have that $k=\frac{\pi}{L_{i}}\mu_{i}$,
which shows claim \eqref{enu:thm_path_decomposition_4} of the theorem.

Finally, claim \eqref{enu:thm_path_decomposition_5} is deduced from
claim \eqref{enu:thm_path_decomposition_4}, by summing over all $\mathcal{P}_{i}$'s.
\end{proof}
Having characterized local critical points, we wish to connect those
to supremizers.
\begin{lem}
\label{lem:supremal_and_simple_is_critical} Let $\metgraph\left(\disgraph;~\lenvec\right)$
be a supremizer of a discrete graph $\mathcal{G}$, such that its
spectral gap $k_{1}\left[\metgraph(\disgraph;~\lenvec)\right]$ is
simple. Then, there exists a discrete graph $\disgraph^{*}$ and positive
edge lengths $\lenvec^{*}\in\len_{\disgraph^{*}}$ such that $\metgraph(\disgraph;~\lenvec)=\metgraph(\disgraph^{*};~\lenvec^{*})$
and the spectral gap $k_{1}\left[\metgraph(\disgraph^{*};~\lenvec^{*}\right]$
is a critical value. 
\end{lem}
\begin{proof}
Start by forming a new discrete graph $\disgraph^{*}$ by contracting
the edges of $\disgraph$ which correspond to the vanishing values
of $\lenvec$, or setting $\disgraph^{*}=\disgraph$ if all entries
of $\lenvec$ are strictly positive. We get that there exists $\lenvec^{*}\in\len_{\disgraph^{*}}$
such that $\metgraph(\disgraph;~\lenvec)=\metgraph(\disgraph^{*};~\lenvec^{*})$.
In effect, $\lenvec^{*}$ entries are exactly the non-vanishing entries
of $\lenvec$. Since $\metgraph(\disgraph;~\lenvec)$ is a supremizer
of $\disgraph$ we get that $\metgraph(\disgraph^{*};~\lenvec^{*})$
is a supremizer of $\disgraph^{*}$. Furthermore, $\metgraph(\disgraph^{*};~\lenvec^{*})$
is even a maximizer of $\disgraph^{*}$ as all of $\lenvec^{*}$ entries
are positive. Since $k_{1}\left[\metgraph(\disgraph^{*};~\lenvec^{*})\right]$
is a simple eigenvalue, it is analytic with respect to edge lengths
and therefore must be a critical value. 
\end{proof}
Having Lemma \ref{lem:supremal_and_simple_is_critical} allows to
conclude that all the claims in lemmata \ref{lem:critical_point_derivative_at_vertices}
and \ref{lem:path_decomposition} hold for supremizers whose spectral
gaps are simple. We use this in proving Theorem \ref{thm:simple_supremum}.
\begin{proof}[Proof of Theorem \ref{thm:simple_supremum}]
 We start by noting that the path decomposition of Lemma \ref{lem:path_decomposition}
is valid under the assumptions of the theorem. Denote for brevity
$\metgraph:=\metgraph(\disgraph;~\lenvec)$ and $k:=k_{1}\left[\metgraph\right]$,
with corresponding eigenfunction $f$. Denote $\metgraph_{+}:=\left\{ \left.x\in\metgraph\right|f\left(x\right)>0\right\} $,~$\metgraph_{-}:=\left\{ \left.x\in\metgraph\right|f\left(x\right)<0\right\} $
and denote by $\beta_{+},\beta_{-}$ their corresponding first Betti
numbers. The connected components of $\metgraph_{+},\metgraph_{-}$
are called the nodal domains of $f$. As $k$ is the second eigenvalue
of $\metgraph$, we deduce from the Courant nodal theorem and the
simplicity of $k$ that $f$ has only two nodal domains (see \cite{Cou_ngwgmp23}
for the original proof of Courant, or \cite{GnuSmiWeb_wrm04,Ber_cmp08}
for its adaptation for graphs). Hence, the sets $\metgraph_{+}$ and
$\metgraph_{-}$ are connected (notice that $\metgraph_{\pm}$ are
not exactly subgraphs, as they do not include the vertices at which
$f$ vanishes).

Next, note that $f$ cannot completely vanish on an edge. Otherwise,
the energy of that edge equals to zero and as $k$ is a critical value,
by the proof of Lemma \ref{lem:critical_point_derivative_at_vertices}
all edge energies are equal to zero which leads to $f\equiv0$. Furthermore,
we show that $f$ cannot vanish more than once on the same edge, including
its endpoints. Assume by contradiction that there exists an edge,
$e=\left[u,v\right]$ on which $f$ vanishes at least twice. As $f$
has only two nodal domains, it can vanish at most twice on $e$. For
each zero of $f$ located on the interior of $e$, add a dummy vertex
of degree two at the position of this zero. Those two zeros now coincide
with two vertices of $\metgraph(\disgraph;~\lenvec)$, which we denote
by $v_{1},v_{2}$ and further denote the degrees of those vertices
by $d_{1},d_{2}$. We note that both $d_{1}$ and $d_{2}$ are even
and in particular not smaller than two. This holds as a zero at an
odd degree vertex implies by Lemma \ref{lem:critical_point_derivative_at_vertices}
that the energy at this vertex vanishes as well. As $k$ is a critical
value, all energies are equal throughout the graph, which implies
$f\equiv0$. From Lemma \ref{lem:path_decomposition}, \eqref{enu:thm_path_decomposition_5}
we get $k=\frac{1}{2}\left(d_{1}+d_{2}\right)\pi$. We modify $\metgraph$
by contracting the edge segment connecting between $v_{1}$ and $v_{2}$,
turning them into a single vertex which we denote by $v_{0}$. We
get that in the new graph, the vertex $v_{0}$ has a degree $d_{0}=d_{1}+d_{2}-2$.
This new graph is connected and we modify it by contracting all edges
except those $d_{0}$ edges connected to $v_{0}$. Doing so, we obtain
a mandarin graph with $d_{1}+d_{2}-2$ edges. By turning the mandarin
into an equilateral mandarin it achieves a spectral gap of $\left(d_{1}+d_{2}-2\right)\pi$
(see Example \ref{exa:mandarin}). As $\metgraph$ is a supremizer
we conclude $\left(d_{1}+d_{2}-2\right)\pi\leq\frac{1}{2}\left(d_{1}+d_{2}\right)\pi$,
so that $d_{1}+d_{2}\leq4$. Since we have seen above that $d_{1}\geq2,~d_{2}\geq2$
we deduce $d_{1}=d_{2}=2$. By the path decomposition in Lemma \ref{lem:path_decomposition},
each path must contain at least one zero of $f$. Hence only a single
path is possible in the decomposition and $\metgraph$ must be a single
cycle graph. We arrive at a contradiction, as the spectral gap of
this graph is not simple. Hence $f$ vanishes at most once on each
edge, which includes both the interior of the edge and its two endpoints.

If $f$ vanishes at points which are not vertices, we turn those points
into dummy vertices of degree two. Each zero of $f$ is now located
at some vertex of $\metgraph$. We introduce the following notation.
Denote by $V_{+}$ ($V_{-}$) the number of vertices at which $f$
is positive (negative), which is just the number of vertices of $\metgraph_{+}$
($\metgraph_{-}$). Denote by $V_{0}$ the number of vertices at which
$f$ vanishes (this includes the additional dummy vertices we added).
Similarly, denote by $E_{++}$ ($E_{--}$) the number of edges which
connect two vertices from $V_{+}$ ($V_{-}$). Note that $f$ does
not vanish at all on those edges. Further denote by $E_{0+}$ ($E_{0-}$)
the number of edges which connect a vertex of $V_{0}$ to a vertex
of $V_{+}$ ($V_{-}$). Note that due to the additional dummy vertices
there are no edges which connect a positive vertex to a negative one.
With those notations, the graph's first Betti number is 
\begin{align}
\beta & =E-V+1\nonumber \\
 & =\left(E_{++}+E_{--}+E_{0+}+E_{0-}\right)-\left(V_{+}+V_{-}+V_{0}\right)+1\nonumber \\
 & =\left(E_{++}-V_{+}+1\right)+\left(E_{--}-V_{-}+1\right)+\left(E_{0+}+E_{0-}-V_{0}\right)-1\nonumber \\
 & =\beta_{+}+\beta_{-}+\left(E_{0+}+E_{0-}-V_{0}\right)-1,\label{eq:thm_simple_supremum_1}
\end{align}
where $\beta_{+}:=E_{++}-V_{+}+1$ is the first Betti number of $\metgraph_{+}$
and similarly for $\beta_{-}:=E_{--}-V_{-}+1$ and $\metgraph_{-}$.
In addition, 
\begin{equation}
E_{0+}+E_{0-}=\sum_{v\in V_{0}}d_{v}=2V_{0}+2\delta,\label{eq:thm_simple_supremum_2}
\end{equation}
where $\delta\geq0$ is defined by the equality above. The sum above
is even by Lemma \ref{lem:critical_point_derivative_at_vertices}
and hence, $\delta$ is an integer. In addition, $\delta=0$ if and
only if $f$ does not vanish on the original vertices of $\metgraph$
(i.e., it vanishes only on the added dummy vertices which are of degree
two). The number of graph zeros, counted with their multiplicities
as in Lemma \ref{lem:path_decomposition} (namely, each zero is counted
as many times as half the degree of the corresponding vertex) is 
\begin{equation}
\mu=\frac{1}{2}\sum_{v\in V_{0}}d_{v}=E_{0+}+E_{0-}-V_{0}-\delta,\label{eq:thm_simple_supremum_3}
\end{equation}
where we used \eqref{eq:thm_simple_supremum_2}. Combining \eqref{eq:lem_path_decomposition_1},
\eqref{eq:thm_simple_supremum_1}, \eqref{eq:thm_simple_supremum_3}
we get 
\begin{equation}
k=\pi\left(\beta+1-\left(\beta_{+}+\beta_{-}\right)-\delta\right).\label{eq:thm_simple_supremum_4}
\end{equation}

Let $v$ be a vertex such that $f\left(v\right)=0$. We concluded
above such a vertex must be of even degree. Furthermore, from Lemma
\ref{lem:critical_point_derivative_at_vertices} we have that half
of $f$ derivatives at $v$ are positive and half negative. Hence,
$v$ is connected to the same number of positive values vertices as
to negative valued once. We conclude that $E_{0+}=E_{0-}$ and from
the left equalities in \eqref{eq:thm_simple_supremum_2} and \eqref{eq:thm_simple_supremum_3}
we get $\mu=E_{0-}$. Choose $\lenvec^{*}\in\lencl_{\disgraph}$ such
that all of its entries equal zero except those which correspond to
the $E_{0-}$ edges, which we set to be equal $\nicefrac{1}{E_{0-}}$.
We get that $\metgraph\left(\disgraph;~\lenvec^{*}\right)$ is an
equilateral mandarin graph whose spectral gap equals $\pi E_{0-}=\pi\mu$,
which finishes the proof of the theorem.
\end{proof}
The proof above yields the following.
\begin{cor}
\label{cor:simple_supremum} Let $\disgraph$ be a discrete graph
and let $\lenvec\in\len_{\disgraph}$. Assume that $\metgraph\left(\disgraph;~\lenvec\right)$
is a supremizer of $\disgraph$ and that the spectral gap $k_{1}\left(\metgraph(\disgraph;~\lenvec)\right)$
is a simple eigenvalue and let $f$ be the corresponding eigenfunction.
Denote $\metgraph_{+}:=\left\{ \left.x\in\metgraph\right|f\left(x\right)>0\right\} $,~$\metgraph_{-}:=\left\{ \left.x\in\metgraph\right|f\left(x\right)<0\right\} $
and further denote by $\beta_{+},\beta_{-}$ their corresponding first
Betti numbers. Then
\end{cor}
\begin{enumerate}
\item \label{enu:cor_simple_supremum_1} $\beta_{+}+\beta_{-}\leq1$.
\item \label{enu:cor_simple_supremum_2} If $\beta_{+}+\beta_{-}=1$ there
exists a choice of lengths $\lenvec^{*}\in\lencl_{\disgraph}$ such
that $\metgraph\left(\disgraph;~\lenvec^{*}\right)$ is an equilateral
flower and 
\[
k_{1}\left(\metgraph\left(\disgraph;~\lenvec\right)\right)=k_{1}\left(\metgraph\left(\disgraph;~\lenvec^{*}\right)\right)=\beta\pi.
\]
\item \label{enu:cor_simple_supremum_3} The number of (non-dummy) vertices
at which $f$ vanishes is at most one. Such a vertex may exist only
if $\beta_{+}+\beta_{-}=0$ and if it exists then this vertex is of
degree four.
\end{enumerate}
\begin{rem*}
We note that $\metgraph_{-},\metgraph_{+}$ defined above are open
sets and hence not metric graphs in the sense defined so far in the
paper. Nevertheless, we can still define their Betti numbers according
to the usual definition for topological spaces.
\end{rem*}
\begin{proof}
We start from equation \eqref{eq:thm_simple_supremum_4} in the preceding
proof. If $\beta_{+}+\beta_{-}>1$ we get that $k<\pi\beta$, so that
the spectral gap of $\metgraph\left(\disgraph;~\lenvec\right)$ is
strictly smaller than the one we can get by turning it into an equilateral
flower ($\pi\beta$) which contradicts it being a supremum. Therefore
$\beta_{+}+\beta_{-}\leq1$, which is claim \eqref{enu:cor_simple_supremum_1}.

If $\beta_{+}+\beta_{-}=1$, then by \eqref{eq:thm_simple_supremum_4},
the spectral gap of $\metgraph\left(\disgraph;~\lenvec\right)$ equals
$\pi\left(\beta-\delta\right)$. As it cannot be smaller than the
one of the equilateral flower we have $\delta=0$, which means that
$f$ does not vanish at vertices (with the exception of the dummy
ones) and also that there exists $\lenvec^{*}\in\lencl_{\disgraph}$
for which $\metgraph\left(\disgraph;~\lenvec^{*}\right)$ is an equilateral
flower, hence showing claim \eqref{enu:cor_simple_supremum_2}.

If $\beta_{+}+\beta_{-}=0$, then by \eqref{eq:thm_simple_supremum_4},
the spectral gap of $\metgraph\left(\disgraph;~\lenvec\right)$ equals
$\pi\left(\beta+1-\delta\right)$. As it cannot be smaller than the
one of the equilateral flower we have $\delta\leq1$, which means
that $f$ vanishes at most on a single (non-dummy) vertex. In addition,
if such a vertex exists its degree equals four.
\end{proof}
Another corollary of the proof of Theorem \ref{thm:simple_supremum}
is the following
\begin{cor}
Let $\disgraph$ be a discrete graph. Let $\lenvec\in\len_{\disgraph}$
and assume that $\metgraph:=\metgraph\left(\disgraph;~\lenvec\right)$
decomposes as 
\begin{equation}
\metgraph=\metgraph_{+}\cup\metgraph_{0}\cup\metgraph_{-},\label{DecompositionForNoSimpleCriticalPoint}
\end{equation}
such that 

\begin{enumerate}
\item The subgraphs $\metgraph_{+},\metgraph_{0}$ and $\metgraph_{-}$
are pairwise edge disjoint.
\item The subgraphs $\metgraph_{+}$ and $\metgraph_{-}$ do not have any
vertex in common. 
\item The vertices of $\metgraph_{0}$ have an odd degree in $\metgraph$. 
\end{enumerate}
Then, the spectral gap of $\metgraph$ cannot be both a simple eigenvalue
and a critical value as a function of $\lenvec\in\len_{\disgraph}$. 
\end{cor}
\begin{proof}
Let $k$ denote the spectral gap of $\metgraph$ and assume that it
is a simple eigenvalue and a critical value. Let $f$ be the eigenfunction
corresponding to $k$. Since $k$ is simple, Courant's nodal theorem
(\cite{Cou_ngwgmp23,GnuSmiWeb_wrm04,Ber_cmp08}) entails that $f$
has exactly two nodal domains. By Lemma \ref{lem:critical_point_derivative_at_vertices}
and as the vertices of $\metgraph_{0}$ are of odd degree, we deduce
that $f$ vanishes on every edge of $\metgraph_{0}$. From the decomposition
\eqref{DecompositionForNoSimpleCriticalPoint}, it follows that $\metgraph_{+}$
and $\metgraph_{-}$ are contained each in a different nodal domain
of $\metgraph$ and also that each is a connected subgraph. Furthermore,
$\metgraph_{0}$ does not have any interior vertex as otherwise, it
would belong to a third nodal domain. It follows that $\metgraph_{0}$
consists of edges connecting vertices of $\metgraph_{+}$ and $\metgraph_{-}$.

Observe that $\left.f\right|_{\metgraph_{+}}$ is a Neumann eigenfunction
on $\metgraph_{+}$. Indeed, it satisfies Neumann conditions at all
vertices of $\metgraph_{+}\backslash\metgraph_{0}$ and its derivative
vanishes at each edge connected to a vertex in $\metgraph_{+}\cap\metgraph_{0}$.
Therefore, $\left.f\right|_{\metgraph_{+}}$ should be orthogonal
to the constant function on $\metgraph_{+}$. As $\left.f\right|_{\metgraph_{+}}$
is positive everywhere, this is possible only if $\metgraph_{+}$consists
of a single vertex, which we denote by $v_{+}$ (it cannot contain
more than a single vertex as we have shown it is connected). The same
goes for $\metgraph_{-}$ (its vertex denoted by $v_{-}$) and as
we have shown that $\metgraph_{0}$ consists of edges connecting vertices
of $\metgraph_{+}$ and $\metgraph_{-}$, we conclude that $\metgraph$
is a mandarin graph. As all derivatives of $f$ at $v_{\pm}$ vanish
and $f$ cannot vanish more than once on edges connecting them we
deduce that all those edges are of equal length. Hence, $\metgraph$
is an equilateral mandarin, whose spectral gap is not a simple eigenvalue
and we get a contradiction.
\end{proof}
This corollary applies, among other examples, to graphs having a bridge
linking two vertices of odd degrees, or to bipartite and $d-$regular
graphs for some odd $d$. All of those cannot have a spectral gap
which is both simple and a critical value.

Demonstrating examples of the other side, we next show a family of
discrete graphs, $\disgraph$, and connected subsets $\len^{*}\subset\len_{\disgraph}$,
such that for all $\lenvec^{*}\in\len^{*}$, $\metgraph\left(\disgraph;~\lenvec^{*}\right)$
satisfies the conditions of Lemma \ref{lem:critical_point_derivative_at_vertices}.
This provides a collection of graphs whose spectral gap is both simple
and a critical value. Those graphs are essentially chains of mandarins
glued serially one to the other and with an optional star glued at
either side of this chain. We call those standarin chains (see Figure
\ref{fig:standarin_chains}).
\begin{prop}
\label{prop:Standarins} Let $n\geq2,~M\geq1$ be integers. Take some
$M$ discrete $n$-mandarin graphs and glue them serially to form
a chain of mandarins. At each end of this chain either glue or not
an $n$-star graph at its central vertex. Let $S\in\left\{ 0,1,2\right\} $
be the number of star graphs which were glued and assume $M+S\geq2$.
Denote the obtained discrete graph by $\disgraph$. Set $\lenvec^{*}\in\len_{\disgraph}$
to be a vector of edge lengths such that

\begin{enumerate}
\item \label{enu:prop-standarins_1}All edges belonging to the same mandarin
have equal length.
\item \label{enu:prop-standarins_2}All edges belonging to the same star
graph have equal length, which is in the range $(0,\frac{1}{2n})$.
\end{enumerate}
Then for all such $\lenvec^{*}\in\len_{\disgraph}$, $\metgraph(\disgraph;~\lenvec^{*})$
satisfies the conditions of Lemma \ref{lem:critical_point_derivative_at_vertices}.
Namely

\begin{enumerate}
\item The spectral gap, $k_{1}\left[\metgraph(\disgraph;~\lenvec^{*})\right]$,
is a simple eigenvalue. 
\item The function $\lenvec\mapsto k_{1}\left[\metgraph(\disgraph;~\lenvec)\right]$
has a critical value at $\lenvec=\lenvec^{*}$.
\end{enumerate}
In addition, the corresponding spectral gap $k=k_{1}\left[\metgraph(\disgraph;~\lenvec)\right]$
equals $n\pi$.
\end{prop}
\begin{figure}
\begin{minipage}[t]{0.4\columnwidth}%
(a)\includegraphics[scale=0.46]{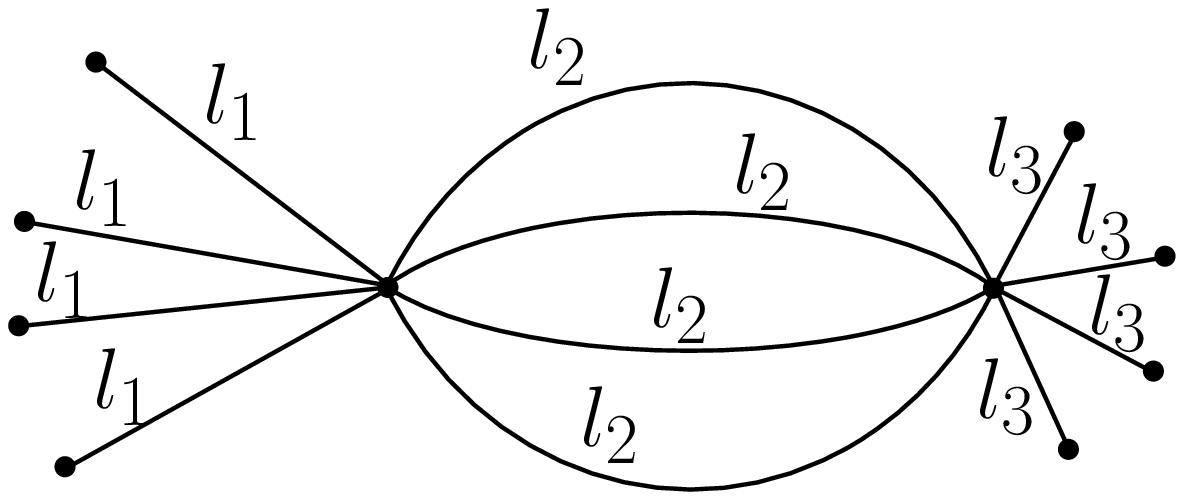}%
\end{minipage}\hfill{}%
\begin{minipage}[t]{0.6\columnwidth}%
(b)\includegraphics[scale=0.56]{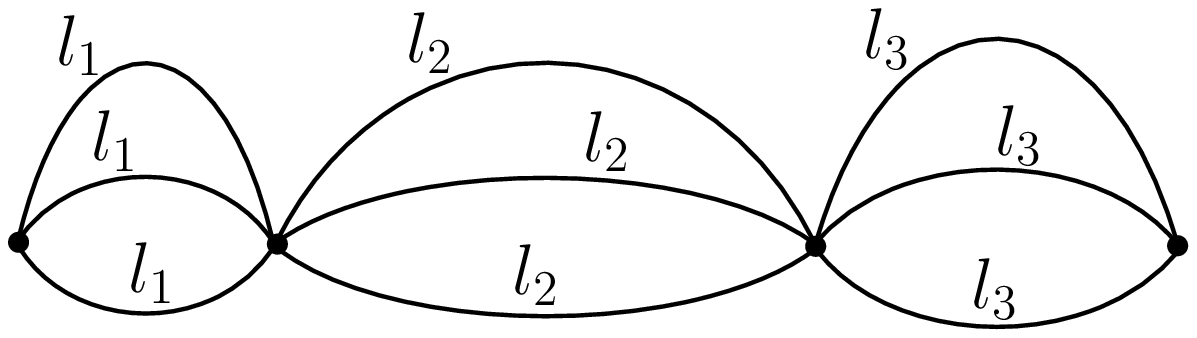}%
\end{minipage}\hspace{-12mm}

\caption{Two examples for the standarin chain graphs }
\label{fig:standarin_chains}
\end{figure}

\begin{proof}
Let $\lenvec^{*}\in\len_{\disgraph}$ which satisfies the assumptions
of the proposition. Denote $\metgraph:=\metgraph(\disgraph;~\lenvec^{*})$
and note that we may construct $\metgraph$ by taking $n$ intervals,
$\left\{ \gamma_{i}\right\} _{i=1}^{n}$, of length $\frac{1}{n}$
each, picking $M+1$ points on each interval which are similarly positioned
on each of the intervals, and identifying each set of parallel $n$
points to form a vertex of $\metgraph$. We use this decomposition
of $\metgraph$ to describe an eigenfunction which is shown on the
sequel to correspond to the spectral gap of $\metgraph$. Set $\left.f\right|_{\gamma_{i}}\left(x\right)=\cos\left(n\pi x_{i}\right)$
on each $\gamma_{i}$. It is easy to check that $f$ satisfies Neumann
conditions at all vertices and hence it is a valid eigenfunction and
its $k$-eigenvalue equals $n\pi$. We conclude that the spectral
gap obeys, $k_{1}\left[\metgraph\right]\leq n\pi$, and show in the
sequel that this is actually an equality and that the spectral gap
is a simple eigenvalue.

Let $g$ be an eigenfunction corresponding to the spectral gap $k_{1}\left[\metgraph\right]$.
We may assume that all the restrictions $\left.g\right|_{\gamma_{i}}$
at mentioned intervals are equal. Otherwise, we symmetrize $g$ by
taking

\[
\forall1\leq i\leq n,~~\left.\tilde{g}\right|_{\gamma_{i}}=\sum_{j=1}^{n}\left.g\right|_{\gamma_{j}}.
\]
This symmetrized function $\tilde{g}$ indeed satisfies Neumann conditions
at all vertices and we just need to justify that it is different from
the zero function. Assume by contradiction that it is the zero function.
In particular $\tilde{g}$ vanishes at all vertices and hence $g$
itself vanishes at all vertices which are not leaves. Necessarily,
there exists some edge on which $g$ does not identically vanish.
If such an edge, $e$, is an inner edge we get that $k_{1}\left[\metgraph\right]\geq\frac{\pi}{l_{e}}>n\pi$,
and a contradiction. If this edge is a dangling edge, we get by assumption
\eqref{enu:prop-standarins_2} that $k_{1}\left[\metgraph\right]\geq\frac{\pi}{2l_{e}}>n\pi$,
which is again a contradiction. Hence we continue assuming that $g$
is an eigenfunction with all $\{\left.g\right|_{\gamma_{i}}\}_{i=1}^{n}$
equal to each other. From here we conclude that for all $i$, $\left.g\right|_{\gamma_{i}}$
is an eigenfunction of the interval with Neumann condition at both
of its ends. This together with $g$ being an eigenfunction corresponding
to the spectral gap implies $g=f$ and $k_{1}\left[\metgraph\right]=n\pi$.

Next, we show the simplicity of $k_{1}\left[\metgraph\right]$. Let
$g$ be an eigenfunction of $k_{1}\left[\metgraph\right]$, not assuming
it is symmetric this time. Take all parallel edges of some mandarin
which is a subgraph of $\metgraph$. All those edges have a common
length $l<\frac{1}{n}$ and we have $k_{1}\left[\metgraph\right]\cdot l=n\pi l<\pi$
so that $\sin(k_{1}\left[\metgraph\right]\cdot l)\neq0$. Therefore,
the value of $g$ on each of those parallel edges is given by

\[
\left.g\right|_{e}(x)=\frac{1}{\sin\left(k_{1}\left[\metgraph\right]\cdot l\right)}\left\{ g\left(u\right)\sin\left(k_{1}\left[\metgraph\right]\cdot\left(l-x\right)\right)+g\left(v\right)\sin\left(k_{1}\left[\metgraph\right]\cdot x\right)\right\} ,
\]
where $u,v$ are the vertices of this mandarin and $e$ any edge connecting
them. A similar argument shows that $g$ is also uniquely determined
at the dangling edges. The simplicity of $k_{1}\left[\metgraph\right]$
follows.

Finally, computing the energy, ~$\energy_{e}=\left(f'\right)^{2}+k^{2}f^{2}$,
of $f$ as defined above, we get that it is equal on all edges. By
Lemma \ref{lem:Derivative_on_edge_equals_energy} we conclude that
the function $\lenvec\mapsto k_{1}\left[\metgraph(\disgraph;~\lenvec)\right]$
has a critical value at $\lenvec=\lenvec^{*}$.
\end{proof}
We note that the particular case $n=2,~M=1,~S=1$ is dealt with in
Lemma \ref{lem:stower_1_2}. It is stated there that for this particular
stower the graphs $\metgraph(\disgraph;~\lenvec)$ not only have the
spectral gap as a critical value, but they are also maximizers. Furthermore,
those graphs are supremizers and thus satisfy the conditions of Theorem
\ref{thm:simple_supremum}. Indeed, this stower has a spectral gap
of $2\pi$, which equals the spectral gap of a single cycle, which
is merely a one petal flower or a two edge mandarin.

In general, the graphs in the proposition above share the same spectral
gap as equilateral $n$-mandarin graphs. As such they obey the conclusion
of Theorem \ref{thm:simple_supremum} even though they do not satisfy
the requirements of the theorem as they are not necessarily supremizers.
For example, the graphs $\metgraph(\disgraph;~\lenvec^{*})$ of the
proposition above are not supremizers if we take $n\geq3$. In this
case, there is a choice of lengths, $\lenvec$, for which $\metgraph(\disgraph;~\lenvec)$
is a stower graph with $E_{p}=M\cdot(n-1)$ and $E_{l}=S\cdot n$,
whose spectral gap is $\frac{\pi}{2}\left(2M\cdot(n-1)+S\cdot n\right)$
and greater than $n\pi$.

\section{Gluing Graphs\label{sec:Gluing_Graphs}}

In this section we develop spectral gap inequalities for graphs whose
vertex connectivity equals one. Such graphs may be obtained by considering
two disjoint graphs and identifying two vertices, one of each graph.
We bound the spectral gap of the obtained graph by the sum of spectral
gaps of its two subgraphs and provide necessary and sufficient conditions
for equality to hold (Proposition \ref{prop:subadditive_gluing_principle}).
We use this in order to prove sufficient conditions needed for graphs
with vertex connectivity one to be supremizers (Theorem \ref{thm:supremum_of_gluing}).

We fix some notations to use throughout this section. Let $\Gamma$
be a graph and let $v$ be a vertex of $\Gamma$. We say that $f$
satisfies the $\delta$-type conditions at $v$ with parameter $\theta$
if 
\begin{align}
f\text{ is }\text{continuous~at }v\nonumber \\
\text{and}\qquad\quad\nonumber \\
\cos\left(\frac{\theta}{2}\right)\sum_{e\in\mathcal{E}_{v}}\frac{\ud f}{\ud x_{e}}\left(v\right) & =\sin\left(\frac{\theta}{2}\right)f\left(v\right),\label{eq:delta-condition_with_theta}
\end{align}
where $\theta\in\left(-\pi,\pi\right]$ (see Definition \ref{def:delta_type_conditions}).
Note that Neumann conditions are obtained as a special case with $\theta=0$
and Dirichlet conditions are obtained from $\theta=\pi$. We denote
by $k_{n}\left(\Gamma;~\theta\right)$ the $n^{\textrm{th}}$ $k$-eigenvalue
of $\Gamma$, endowed with the $\delta$-type condition with parameter
$\theta$ at $v$ and Neumann at all other vertices. The corresponding
$k$-spectrum is denoted by 
\begin{equation}
\sigma\left(\Gamma;~\theta\right):=\cup_{n}\left\{ k_{n}\left(\Gamma;~\theta\right)\right\} .\label{eq:spectrum_notation}
\end{equation}
It will be understood in the sequel which vertex $v$ is chosen so
that it is not indicated in the notation. In addition, we omit the
notation $\metgraph$ from $k_{n}\left(\Gamma;~\theta\right)$ and
$\sigma\left(\Gamma;~\theta\right)$ whenever it is clear which graph
we refer to. Similarly, $\theta$ is omitted from these notations
whenever $\theta=0$ to comply with the notations used so far. At
this point, we refer the reader to Appendix \ref{sec:appendix_interlacing_theorems},
where we quote some results from \cite{BerKuc_quantum_graphs} on
$\delta$-type conditions, that are used throughout this section.
The structure of the spectrum as it depends on the parameter $\theta$
(for some chosen vertex $v$) is described in the next lemma, which
quotes parts of theorem 3.1.13 from \cite{BerKuc_quantum_graphs},
slightly rephrased for our purpose. 
\begin{lem}
\label{lem:spectrum_structure}

Let $\metgraph$ be a metric graph and let $v$ be a vertex of $\metgraph$.
There exist a bounded from below discrete set, $\Delta\left(\metgraph\right)\subset\R$
and a real smooth function, $K\left(\metgraph;\cdot\right):\left(-\pi,\infty\right)\rightarrow\R$
(called ``dispersion relation'') such that 

\begin{enumerate}
\item \label{enu:lem_spectrum_structure_1}The function $\theta\mapsto K\left(\metgraph;\theta\right)$
is strictly increasing so that $\lim_{\theta\to\infty}K\left(\metgraph;\theta\right)=\infty$.
\item \label{enu:lem_spectrum_structure_2}For any $\theta\in\left(-\pi,\pi\right]$,
$\sigma\left(\Gamma;~\theta\right)=\left\{ K\left(\metgraph;~\theta+2\pi n\right)\right\} _{n=0}^{\infty}\cup\Delta\left(\metgraph\right)$. 
\end{enumerate}
\end{lem}
\begin{rem*}
We see from the lemma above that 
\[
\Delta\left(\metgraph\right)=\bigcap_{\theta\in\left(-\pi,\pi\right]}\sigma\left(\Gamma;~\theta\right).
\]
The values of this discrete set, common to all spectra, are often
called flat bands.
\end{rem*}
A particular value of $\theta$ which plays a special role is defined
below.
\begin{defn}
\label{def:spectral_gap_parameter} Let $\Gamma$ be a graph and let
$v$ be a vertex of $\Gamma$. A parameter $\theta^{SG}\in\R$ which
satisfies 
\[
K\left(\metgraph;\theta^{SG}\right)=k_{1}\left(\metgraph;0\right),
\]
is called the spectral gap parameter (SGP) of $\metgraph$ (with respect
to $v$). See Figure \ref{fig:SGP}.
\end{defn}
\begin{figure}
\begin{minipage}[t]{0.3\columnwidth}%
\includegraphics[scale=0.5]{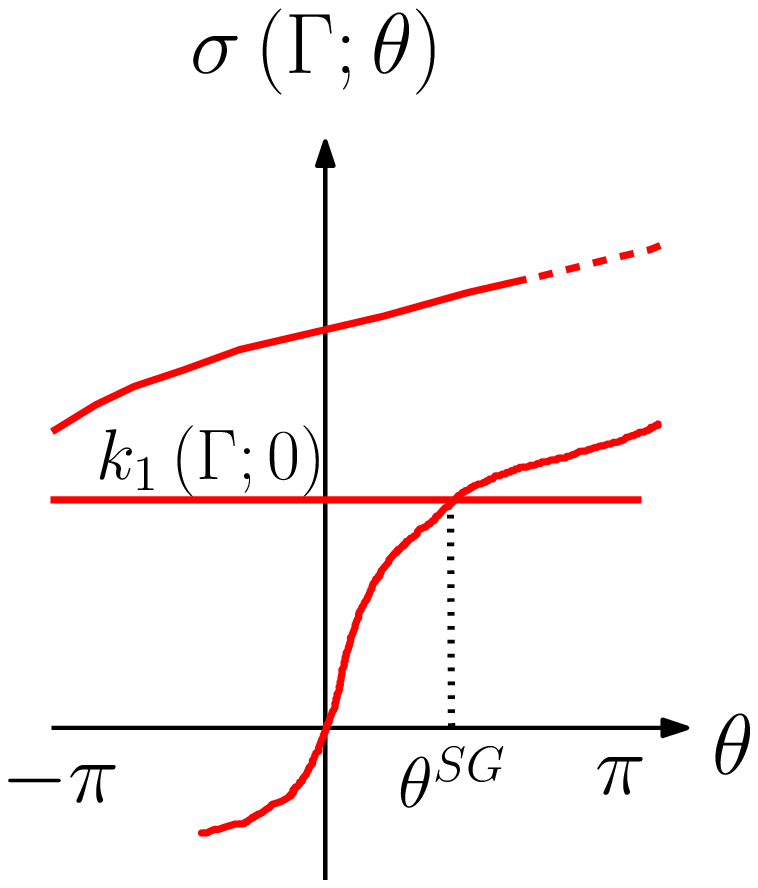}

(a) $\theta^{SG}\in\left(0,\pi\right)$%
\end{minipage}\hfill{}%
\begin{minipage}[t]{0.3\columnwidth}%
\includegraphics[scale=0.5]{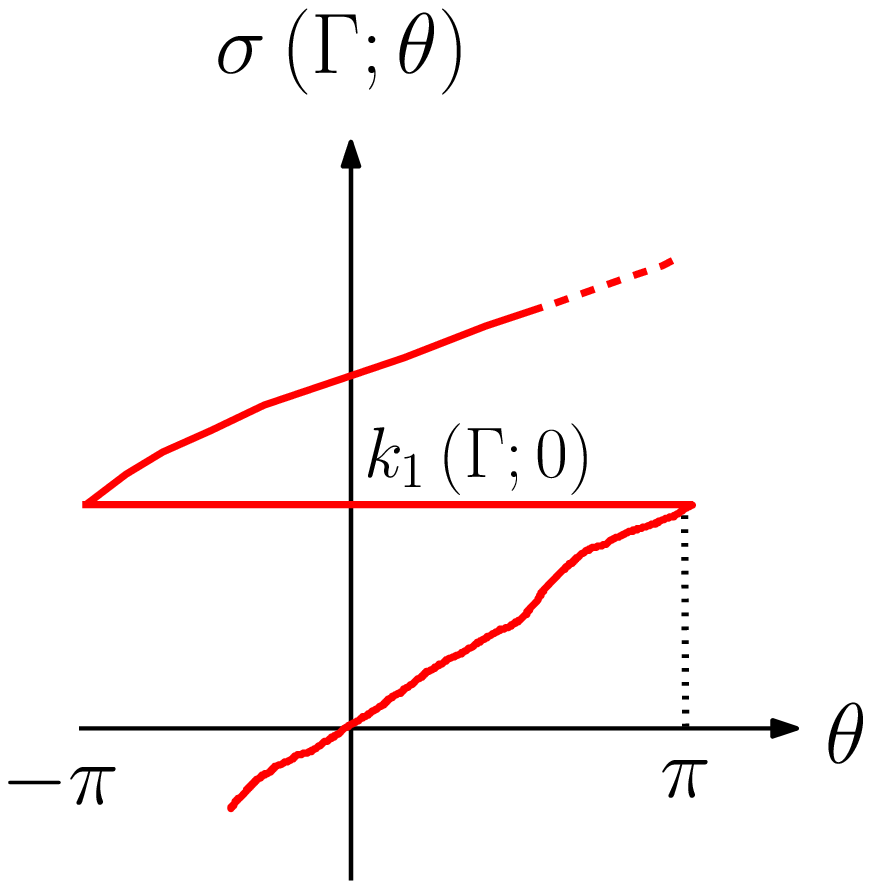}

(b) $\theta^{SG}=\pi$%
\end{minipage}\hfill{}%
\begin{minipage}[t]{0.3\columnwidth}%
\includegraphics[scale=0.5]{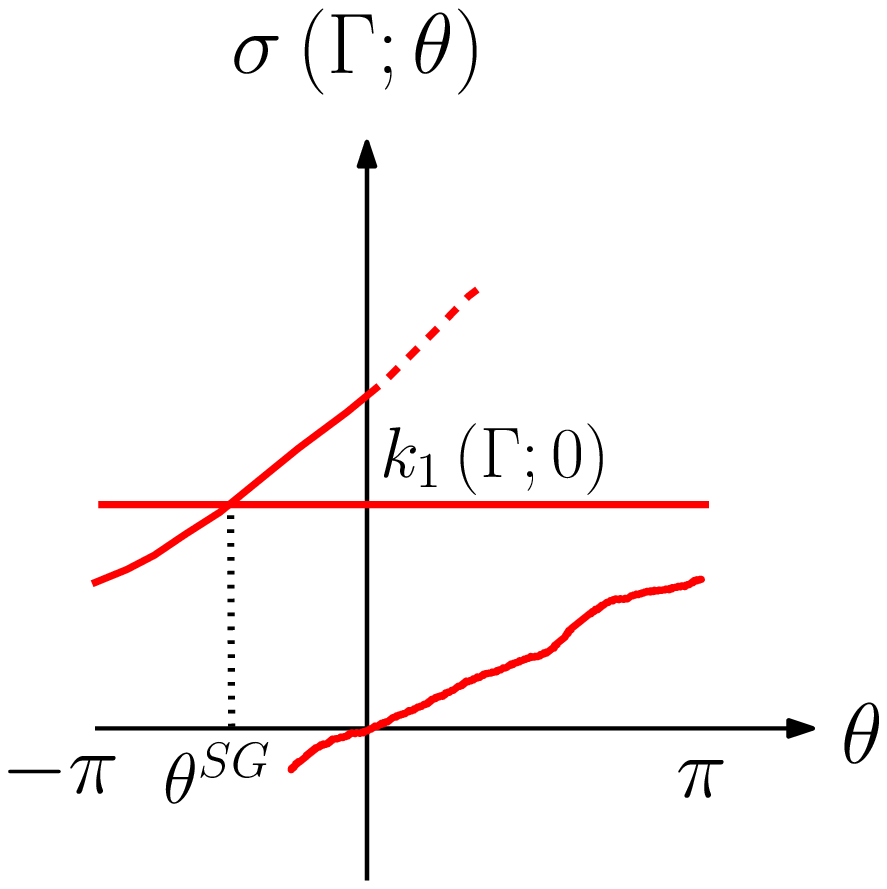}

(c) $\theta^{SG}\in\left(\pi,2\pi\right)$%
\end{minipage}

\caption{Three examples of dispersion relations curves \label{fig:SGP}}
\end{figure}

In the following we point out some of the SGP properties.
\begin{lem}
\label{lem:SG_paramter_properties}~ 

\begin{enumerate}
\item \label{enu:lem_SGP_properties_1}The spectral gap parameter exists
and it is unique.
\item \label{enu:lem_SGP_properties_2} $\theta^{SG}\in\left[0,2\pi\right]$. 
\item \label{enu:lem_SGP_properties_3} If $\theta^{SG}\neq2\pi$ then $k_{1}\left(\metgraph;~0\right)\in\Delta\left(\metgraph\right)$.
\item \label{enu:lem_SGP_properties_4} If $\theta^{SG}\in\left(0,\pi\right]$
then
\begin{equation}
\begin{cases}
k_{0}\left(\theta\right)<k_{1}\left(0\right) & ~~\textrm{for }~\theta\in\left(0,\theta^{SG}\right)\\
k_{0}\left(\theta\right)=k_{1}\left(0\right) & ~~\textrm{for }~\theta\in\left[\theta^{SG},\pi\right]\\
k_{1}\left(\theta-2\pi\right)=k_{1}\left(0\right) & ~~\textrm{for }~\theta\in\left(\pi,2\pi\right]
\end{cases}\label{eq:lem_SGP_property_4}
\end{equation}
\item \label{enu:lem_SGP_properties_5}If $\theta^{SG}\in\left(\pi,2\pi\right)$
then
\begin{equation}
\begin{cases}
k_{0}\left(\theta\right)<k_{1}\left(0\right) & ~~\textrm{for }~\theta\in\left[0,\pi\right]\\
k_{1}\left(\theta-2\pi\right)<k_{1}\left(0\right) & ~~\textrm{for }~\theta\in\left(\pi,\theta^{SG}\right)\\
k_{1}\left(\theta-2\pi\right)=k_{1}\left(0\right) & ~~\textrm{for }~\theta\in\left[\theta^{SG},2\pi\right]
\end{cases}\label{eq:lem_SGP_property_5}
\end{equation}
\end{enumerate}
\end{lem}
\begin{proof}
The existence of the spectral gap parameter follows from $K\left(\metgraph;~0\right)=0$
together with $K\left(\metgraph;\cdot\right)$ being monotonically
increasing. This latter argument also shows the uniqueness of the
SGP and that $\sgp\geq0$.

We have that $K(\metgraph;~2\pi)=k_{n}(\metgraph;~0)$ for some $n$
and hence, by continuity and monotonicity of $K$ we get $\sgp\leq2\pi$,
which shows property \eqref{enu:lem_SGP_properties_2} above.

If $\sgp<2\pi$ we have $k_{1}(\metgraph;~0)\in\sigma(\metgraph;~0)\cap\sigma(\metgraph;~\sgp)$
and by Lemma \ref{lem:flat_band_eigenvalue} conclude $k_{1}(\metgraph;~0)\in\Delta(\metgraph)$,
which proves property \eqref{enu:lem_SGP_properties_3}. Finally,
properties \eqref{enu:lem_SGP_properties_4} and \eqref{enu:lem_SGP_properties_5}
are straightforward consequences of the strict monotonicity of $K$
together with the eigenvalue interlacing with respect to the $\delta$-type
condition parameter (see Lemma \ref{lem:interlacing_wrt_delta_parameters}).
\end{proof}
The main construction in this section involves scaling two disjoint
graphs and gluing them at a vertex to form a new graph, as defined
below. 
\begin{defn}
Let $\Gamma_{1},\Gamma_{2}$ two Neumann graphs of total length 1
each. Let $v_{i}$ be a vertex of $\Gamma_{i}$ ($i=1,2$). Let $\Gamma$
be the graph obtained by the following process 

\begin{enumerate}
\item Multiply all edge lengths of $\Gamma_{1}$ by some factor $L\in\left[0,1\right]$. 
\item Multiply all edge lengths of $\Gamma_{2}$ by a factor of $1-L$. 
\item Identify $v_{1}$ and $v_{2}$ of the graphs above and endow the new
vertex with Neumann vertex conditions. 
\end{enumerate}
We call $\Gamma$ the \emph{gluing} of $\Gamma_{1},\Gamma_{2}$ (with
respect to $v_{1},v_{2}$ and $L$). 
\end{defn}
\begin{prop}
\label{prop:subadditive_gluing_principle} Let $\Gamma_{1},\Gamma_{2}$
two connected Neumann graphs of total length 1 each. Let $v_{i}$
be a vertex of $\Gamma_{i}$ ($i=1,2$). Let $\Gamma$ be the gluing
of $\Gamma_{1},\Gamma_{2}$ with respect to $v_{1},v_{2}$ and some
value $L\in\left[0,1\right]$. Let $\sgp_{1},\sgp_{2}$ be the spectral
gap parameters of $\Gamma_{1},\Gamma_{2}$ with respect to $v_{1},v_{2}$,
correspondingly. Then the following inequality holds 
\begin{equation}
k_{1}\left(\Gamma\right)\leq k_{1}\left(\Gamma_{1}\right)+k_{1}\left(\Gamma_{2}\right),\label{eq:subadditive_gluing_principle}
\end{equation}
with equality if and only if both conditions below are satisfied 

\begin{enumerate}
\item \label{enu:prop_SGP_condition_1}$L=\frac{k_{1}\left(\Gamma_{1}\right)}{k_{1}\left(\Gamma_{1}\right)+k_{1}\left(\Gamma_{2}\right)}$\vspace{2mm}
\item \label{enu:prop_SGP_condition_2} $\sgp_{1}+\sgp_{2}\leq2\pi$
\end{enumerate}
\noindent Additional necessary conditions for equality in \eqref{eq:subadditive_gluing_principle}
are \renewcommand{\labelenumi}{(\alph{enumi})}

\begin{enumerate}
\item \label{enu:prop_SGP_condition_a}The spectral gaps of the glued graphs
obey $k_{1}\left(\Gamma_{1}\right)\in\Delta\left(\metgraph_{1}\right)$
and $k_{1}\left(\Gamma_{2}\right)\in\Delta\left(\metgraph_{2}\right)$
.
\item \label{enu:prop_SGP_condition_b}The spectral gap of the outcome graph,
$k_{1}\left(\metgraph\right)$ is a multiple (i.e. non-simple) eigenvalue. 
\end{enumerate}
\end{prop}
\begin{proof}
We start by showing the inequality \eqref{eq:subadditive_gluing_principle}.

Let $L\in\left[0,1\right]$. If $L=0$ ($L=1$), then $\metgraph=\metgraph_{2}$
($\metgraph=\metgraph_{1}$) and \eqref{eq:subadditive_gluing_principle}
obviously holds as a strict inequality and indeed condition \eqref{enu:prop_SGP_condition_1}
is violated if $L=0$ or $L=1$. We therefore assume $L\in\left(0,1\right)$.
Denote by $\tilde{\Gamma}{}_{1}$ the graph obtained by multiplying
all edge lengths of $\Gamma_{1}$ by $L$ and by $\tilde{\Gamma}{}_{2}$
the graph obtained by multiplying all edge lengths of $\Gamma_{2}$
by $1-L$. Therefore identifying the vertices $v_{1},v_{2}$ of $\metgraph_{1},\metgraph_{2}$
gives the graph $\metgraph$. Applying Lemma \ref{lem:interlacing_of_gluing}
we get 
\[
k_{1}\left(\Gamma\right)\leq k_{2}\left(\tilde{\Gamma}{}_{1}\cup\tilde{\Gamma}{}_{2}\right).
\]
As the spectrum of $\tilde{\Gamma}{}_{1}\cup\tilde{\Gamma}{}_{2}$
is the union of spectra of both graphs, we have that 
\[
k_{0}\left(\tilde{\Gamma}{}_{1}\cup\tilde{\Gamma}{}_{2}\right)=k_{1}\left(\tilde{\Gamma}{}_{1}\cup\tilde{\Gamma}{}_{2}\right)=0\enspace\enspace\textrm{and}\enspace\enspace k_{2}\left(\tilde{\Gamma}{}_{1}\cup\tilde{\Gamma}{}_{2}\right)=\min\left(k_{1}\left(\tilde{\Gamma}{}_{1}\right),k_{1}\left(\tilde{\Gamma}{}_{2}\right)\right)
\]
and conclude 
\begin{equation}
k_{1}\left(\Gamma\right)\leq\min\left(k_{1}\left(\tilde{\Gamma}{}_{1}\right),k_{1}\left(\tilde{\Gamma}{}_{2}\right)\right)=\min\left(\frac{k_{1}\left(\Gamma{}_{1}\right)}{L},\frac{k_{1}\left(\Gamma{}_{2}\right)}{1-L}\right).\label{eq:prop_SGP_1}
\end{equation}
We consider the right hand side of \eqref{eq:prop_SGP_1} as a function
of $L$. The minimal value of this function is $k_{1}\left(\Gamma_{1}\right)+k_{1}\left(\Gamma_{2}\right)$
and it is obtained at $L=\frac{k_{1}\left(\Gamma_{1}\right)}{k_{1}\left(\Gamma_{1}\right)+k_{1}\left(\Gamma_{2}\right)}$,
which proves \eqref{eq:subadditive_gluing_principle}. In addition,
as the minimal value of this function is unique, it also proves that
condition \eqref{enu:prop_SGP_condition_1} is necessary for equality
in \eqref{eq:subadditive_gluing_principle} to hold. From now on we
assume throughout the proof that condition \eqref{enu:prop_SGP_condition_1}
of the proposition is satisfied, so that $k_{1}\left(\tilde{\Gamma}{}_{1}\right)=k_{1}\left(\tilde{\Gamma}{}_{2}\right)$. 

Next, we examine two ranges of $\sgp_{1},\sgp_{2}$ values and show
those values make the inequality in \eqref{eq:subadditive_gluing_principle}
strict.

\begin{enumerate}
\item \uline{$\sgp_{1}>\pi$ and $\sgp_{1}>\pi$.}\\
By \eqref{eq:lem_SGP_property_5} we have $k_{0}(\tilde{\metgraph}_{i};~\pi)<k_{1}(\tilde{\metgraph}_{i};~0)$
for both $i=1,2$. Assume first that $k_{0}(\tilde{\metgraph}_{1};~\pi)\neq k_{0}(\tilde{\metgraph}_{2};~\pi)$
and without loss of generality that $k_{0}(\tilde{\metgraph}_{1};~\pi)>k_{0}(\tilde{\metgraph}_{2};~\pi)$.\\
Examine the function 
\begin{equation}
h\left(\theta\right):=\begin{cases}
k_{0}\left(\tilde{\Gamma}{}_{1};~\theta\right)-k_{1}\left(\tilde{\Gamma}{}_{2};~-\theta\right) & ~~~\theta\in\left[0,\pi\right)\\
k_{0}\left(\tilde{\Gamma}{}_{1};~\pi\right)-k_{0}\left(\tilde{\Gamma}{}_{2};~\pi\right) & ~~~\theta=\pi
\end{cases}.\label{eq:prop_SGP_h_theta}
\end{equation}
By lemma \ref{lem:continuity_wrt_delta_parameter} we have that $h$
is a continuous non-decreasing function. In addition $h\left(0\right)=-k_{1}(\tilde{\metgraph}_{2};~0)<0$
and by the assumption $k_{0}(\tilde{\Gamma}{}_{1};~\pi)>k_{0}(\tilde{\Gamma}{}_{2};~\pi)$
we have $h\left(\pi\right)>0$. Hence $h$ vanishes at some value
$\tilde{\theta}\in\left(0,\pi\right)$, so that we find 
\begin{equation}
k_{0}\left(\tilde{\Gamma}{}_{1};~\tilde{\theta}\right)=k_{1}\left(\tilde{\Gamma}{}_{2};~-\tilde{\theta}\right).\label{eq:theta_and_minus_theta_dispertion}
\end{equation}
Denote by $\tilde{f}_{1}$ the eigenfunction corresponding to $k_{0}(\tilde{\Gamma}{}_{1};~\tilde{\theta})$
and by $\tilde{f}_{2}$ the eigenfunction corresponding to $k_{1}(\tilde{\Gamma}{}_{2};~-\tilde{\theta})$.
We use $\tilde{f}_{1},\tilde{f}_{2}$ to construct an eigenfunction
on the whole of $\metgraph$ as follows. First, notice that for both
$i=1,2$ , $\tilde{f}_{i}(v_{i})\neq0$. Assuming otherwise, we obtain
that $\tilde{f}_{i}$ obeys Dirichlet condition at $v_{i}$ and as
$\tilde{\theta}\ne\pi$ we get that $\tilde{f}_{i}$ obeys Neumann
conditions as well at $v_{i}$. Since $\tilde{\theta}<\sgp_{i}$,
the corresponding eigenvalue is strictly lower than the spectral gap.
As $\tilde{f}_{i}(v_{i})\neq0$ for $i=1,2$, we may normalize the
$\tilde{f}_{i}$'s so that $\tilde{f}_{1}(v_{1})=\tilde{f}_{2}(v_{2})$.
Now form an eigenfunction $f$ on $\metgraph$ by setting
\begin{equation}
f\left(x\right):=\begin{cases}
\tilde{f}_{1}\left(x\right) & ~x\in\tilde{\metgraph}_{1},\\
\tilde{f}_{2}\left(x\right) & ~x\in\tilde{\metgraph}_{2}.
\end{cases}\label{eq:prop_SGP_function_construction}
\end{equation}
where we consider $\tilde{\metgraph}_{1},\tilde{\metgraph}_{2}$ as
subgraphs of $\metgraph$. The normalization $\tilde{f}_{1}(v_{1})=\tilde{f}_{2}(v_{2})$
gives that $f$ is continuous at the glued vertex $v$. In addition,
its sum of derivatives there equals
\begin{equation}
\sum_{e\in\E_{v_{1}}}\left.\tilde{f}_{1}'\right|(v_{1})+\sum_{e\in\E_{v_{2}}}\left.\tilde{f}_{2}'\right|(v_{2})=\tan\left(\frac{\tilde{\theta}}{2}\right)\tilde{f}_{1}(v_{1})+\tan\left(\frac{-\tilde{\theta}}{2}\right)\tilde{f}_{2}(v_{2})=0.\label{eq:prop_SGP_derivatives_sum_zero}
\end{equation}
We conclude that $f$ is a Neumann eigenfunction on $\metgraph$ whose
eigenvalue equals $k_{0}(\tilde{\Gamma}{}_{1};~\tilde{\theta})=k_{1}(\tilde{\Gamma}{}_{2};~-\tilde{\theta})$.
However, this eigenvalue is strictly smaller than $k_{1}\left(\tilde{\Gamma}{}_{i}\right)$,
for both $i=1,2$, as shows the following chain of inequalities
\begin{equation}
k_{0}(\tilde{\Gamma}{}_{1};~\tilde{\theta})\leq k_{0}(\tilde{\Gamma}{}_{1};~\pi)<k_{1}(\tilde{\Gamma}{}_{1};~0)=k_{1}(\tilde{\Gamma}{}_{2};~0),\label{eq:prop_SGP_inequalities_chain}
\end{equation}
where the first inequality is due to eigenvalue monotonicity, the
second is by \eqref{eq:lem_SGP_property_5} and the last equality
results since our current working assumption is the validity of condition
\eqref{enu:prop_SGP_condition_1}, as discussed above. Therefore,
we have found an eigenvalue of $\metgraph$ strictly smaller than
both $k_{1}\left(\tilde{\Gamma}{}_{i}\right)$, so that there is a
strict inequality in \eqref{eq:prop_SGP_1} and therefore strict inequality
in \eqref{eq:subadditive_gluing_principle}.\\
We now assume $k_{0}(\tilde{\metgraph}_{1};~\pi)=k_{0}(\tilde{\metgraph}_{2};~\pi)$.
Denote by $\tilde{f}_{1},\tilde{f}_{2}$ as above the corresponding
eigenfunctions. By \eqref{eq:lem_SGP_property_5} $k_{0}(\tilde{\metgraph}_{i};~\pi)<k_{1}(\tilde{\metgraph}_{i};~0)$
for both $i=1,2$ and therefore $\tilde{f}_{i}$ does not obey Neumann
conditions at $v_{i}$ (as otherwise, its eigenvalue would be the
spectral gap). Using that the sum of derivatives of $\tilde{f}_{i}$
at $v_{i}$ differs from zero, we may normalize both $\tilde{f}_{1},\tilde{f}_{2}$
so that their sums of derivatives are opposite. Now, constructing
a function $f$ on $\metgraph$ as in \eqref{eq:prop_SGP_function_construction}
shows just as above (see \eqref{eq:prop_SGP_inequalities_chain} and
the argument which follows) that inequality \eqref{eq:prop_SGP_1}
is strict in this case as well. We conclude that the inequality in
\eqref{eq:subadditive_gluing_principle} is strict if $\sgp_{1}>\pi$
and $\sgp_{1}>\pi$.
\item \uline{$\sgp_{1}+\sgp_{2}>2\pi$ and $\left\{ \sgp_{1}\leq\pi<\sgp_{2}\text{ or }\sgp_{2}\leq\pi<\sgp_{1}\right\} $.}\\
Assume without loss of generality that $\sgp_{1}<\sgp_{2}$. We have
the following chain of inequalities 
\[
k_{0}(\tilde{\metgraph}_{2};~\pi)<k_{1}(\tilde{\metgraph}_{2};~0)=k_{1}(\tilde{\metgraph}_{1};~0)=k_{0}(\tilde{\metgraph}_{1};~\pi),
\]
where the first inequality comes from \mbox{\eqref{eq:lem_SGP_property_5}}
(keeping in mind that $\sgp_{2}>\pi$), the first equality is our
working assumption (assuming the validity of condition \mbox{\eqref{enu:prop_SGP_condition_1}})
and the second equality comes from \mbox{\eqref{eq:lem_SGP_property_4}}
(keeping in mind that $\sgp_{1}\leq\pi$). Therefore, defining the
function $h$ as in \mbox{\eqref{eq:prop_SGP_h_theta}} we find that
$h(0)<0$ and $h(\pi)>0$. As before we conclude that $h$ vanishes
for some value $\tilde{\theta}\in\left(0,\pi\right)$ and hence $k_{0}(\tilde{\metgraph}_{1};~\tilde{\theta})=k_{1}(\tilde{\metgraph}_{2};~-\tilde{\theta})$.
Similarly to the previous case, we may use this equality to construct
a Neumann eigenfunction on $\metgraph$ whose eigenvalue equals $k_{0}(\tilde{\metgraph}_{1};~\tilde{\theta})$
and to show that strict inequality happens in \mbox{\eqref{eq:subadditive_gluing_principle}}
for this case.
\end{enumerate}
Notice that condition \eqref{enu:prop_SGP_condition_2} of the proposition
forms the complement of the two cases examined above. Therefore, we
have proven so far that this condition is necessary for the equality
in \eqref{eq:subadditive_gluing_principle} to hold. We proceed to
show that conditions \eqref{enu:prop_SGP_condition_1},\eqref{enu:prop_SGP_condition_2}
are sufficient as well. Recall that assuming condition \eqref{enu:prop_SGP_condition_1}
implies $k_{1}(\tilde{\Gamma}{}_{1};~0)=k_{1}(\tilde{\Gamma}{}_{2};~0)$.
We further assume by contradiction that $k_{1}\left(\metgraph\right)<k_{1}(\tilde{\Gamma}{}_{1};~0)$,
and consider the following two cases for the $\sgp_{1},\sgp_{2}$
values:
\begin{enumerate}
\item \uline{$\sgp_{1}\leq\pi$ and $\sgp_{2}\leq\pi$.}\\
First, we note that by \eqref{eq:lem_SGP_property_4} we have $k_{1}(\tilde{\metgraph}_{i};~0)=k_{0}(\tilde{\metgraph}_{i};~\pi)$
for both $i=1,2$.\\
Let $f$ be the eigenfunction corresponding to $k_{1}\left(\metgraph\right)$.
Denote $\tilde{f}_{i}=f\at_{\tilde{\metgraph}_{i}}$ for $i=1,2$.
We find that there exists some $\tilde{\theta}$ such that $k_{n_{1}}(\tilde{\metgraph}_{1};~\tilde{\theta})=k_{1}\left(\metgraph\right)$,
for some $n_{1}$. We cannot have $\tilde{\theta}=\pi$, as otherwise
we get 
\[
k_{n_{1}}(\tilde{\metgraph}_{1};~\pi)=k_{1}\left(\metgraph\right)<k_{1}(\tilde{\Gamma}{}_{1};~0)=k_{0}(\tilde{\metgraph}_{1};~\pi)
\]
and contradiction. We find that as $\tilde{f}_{1}$ satisfies the
$\delta$-type condition at $v_{1}$ with the parameter $\tilde{\theta}$,
$\tilde{f}_{2}$ satisfies the $\delta$-type condition at $v_{2}$
with the parameter $-\tilde{\theta}$ (since the total sum of derivatives
is zero and see \eqref{eq:prop_SGP_derivatives_sum_zero}). Assume
without loss of generality that $\tilde{\theta}>0$. We get that 
\begin{equation}
k_{n_{2}}(\tilde{\metgraph}_{2};~-\tilde{\theta})=k_{1}\left(\metgraph\right)<k_{1}(\tilde{\Gamma}{}_{2};~0),\label{eq:prop_SGP_k_n2}
\end{equation}
which implies either $n_{2}=0$ or $n_{2}=1$. We rule out $n_{2}=0$
as it renders the left hand side of \eqref{eq:prop_SGP_k_n2} negative,
while $k_{1}(\metgraph)>0$. We also rule out $n_{2}=1$, as by \eqref{eq:lem_SGP_property_4}
the left and right hand sides of \eqref{eq:prop_SGP_k_n2} are equal.
Hence, in this case, we get a contradiction to the assumption $k_{1}\left(\metgraph\right)<k_{1}(\tilde{\Gamma}{}_{1};~0)$.
\item \uline{$\sgp_{1}+\sgp_{2}\leq2\pi$ and $\left\{ \sgp_{1}\leq\pi<\sgp_{2}\text{ or }\sgp_{2}\leq\pi<\sgp_{1}\right\} $.}\\
We repeat the construction of $\tilde{f}_{1},\tilde{f}_{2}$ as in
the previous case to get that there exists some $\tilde{\theta}\neq\pi$
such that $k_{n_{1}}(\tilde{\metgraph}_{1};~\tilde{\theta})=k_{1}\left(\metgraph\right)$,
for some $n_{1}$ and $k_{n_{2}}(\tilde{\metgraph}_{2};~-\tilde{\theta})=k_{1}\left(\metgraph\right)$,
for some $n_{2}$. Assume without loss of generality $\sgp_{1}<\sgp_{2}$.
Combining 
\[
k_{n_{1}}(\tilde{\metgraph}_{1};~\tilde{\theta})=k_{1}\left(\metgraph\right)<k_{1}(\tilde{\Gamma}{}_{1};~0)
\]
 with \mbox{\eqref{eq:lem_SGP_property_4}} shows that $n_{1}=0$
and $0<\tilde{\theta}<\sgp_{1}$. Similarly, we have for $\tilde{\metgraph}_{2}$,
\[
k_{n_{2}}(\tilde{\metgraph}_{2};~-\tilde{\theta})=k_{1}\left(\metgraph\right)<k_{1}(\tilde{\Gamma}{}_{2};~0),
\]
where the positivity of the left hand side implies $n_{2}=1$. Together
with \mbox{\eqref{eq:lem_SGP_property_5}} we get $-\tilde{\theta}<\sgp_{2}-2\pi$.
Combining that with $\tilde{\theta}<\sgp_{1}$ gives $\sgp_{1}+\sgp_{2}>2\pi$
and contradiction to the assumption in this case.
\end{enumerate}
Thus, we have shown that conditions \eqref{enu:prop_SGP_condition_1},\eqref{enu:prop_SGP_condition_2}
of the proposition are also sufficient for equality in \eqref{eq:subadditive_gluing_principle}
to hold.

Finally, we show the necessity of conditions (a),(b) of the proposition.
We have seen that necessary conditions for equality in \eqref{eq:subadditive_gluing_principle}
are \{$\sgp_{1}\leq\pi$ and $\sgp_{2}\leq\pi$\} or \{$\sgp_{1}+\sgp_{2}\leq2\pi$
and $\left\{ \sgp_{1}\leq\pi<\sgp_{2}\text{ or }\sgp_{2}\leq\pi<\sgp_{1}\right\} $\}.
Under those conditions we have both $\sgp_{1}\neq2\pi$ and $\sgp_{2}\neq2\pi$
and by Lemma \ref{lem:SG_paramter_properties},\eqref{enu:lem_SGP_properties_3}
we get $k_{1}(\tilde{\metgraph}_{i})\in\Delta(\tilde{\metgraph}_{i})$
for both $i=1,2$, which is condition (a). Now, in order show that
$k_{1}\left(\metgraph\right)$ is a non-simple eigenvalue we construct
two linearly independent eigenfunctions. As $k_{1}(\tilde{\metgraph}_{i})\in\Delta(\tilde{\metgraph}_{i})$,
by Lemma \ref{lem:flat_band_eigenvalue} there exists an eigenfunction
corresponding to $k_{1}(\tilde{\metgraph}_{i})$ which vanishes at
$v_{i}$ and its sum of derivatives vanishes there as well. Extend
this function to an eigenfunction of $\metgraph$, whose eigenvalue
is $k_{1}(\tilde{\metgraph}_{i})=k_{1}\left(\metgraph\right)$ by
setting it to be equal zero on the complementary subgraph, $\tilde{\metgraph}_{3-i}$.
Performing this for both $i=1$ and $i=2$ we get two linearly independent
eigenfunctions on $\metgraph$, which shows the necessity of condition
(b).
\end{proof}
We use Proposition \ref{prop:subadditive_gluing_principle} to study
the supremizers of graphs whose vertex connectivity equals one. Let
$\disgraph$ be such a graph which is obtained by taking two graphs
$\disgraph_{1},\disgraph_{2}$ and identifying two of their vertices
$v_{1},v_{2}$. An immediate guess is that a supremizer of $\disgraph$
may be obtained by taking the supremizers of $\disgraph_{1},\disgraph_{2}$
and identifying their vertices corresponding to $v_{1},v_{2}$. This
holds under some conditions, as stated in Theorem \ref{thm:supremum_of_gluing}
and proved below.
\begin{proof}[Proof of Theorem \ref{thm:supremum_of_gluing}]
 We start by formulating the Dirichlet criterion in terms of the
SGP, $\sgp$, used in the conditions of Proposition \ref{prop:subadditive_gluing_principle}.
Let $\metgraph$ be a graph which obeys the Dirichlet criterion. This
means that $k_{0}(\metgraph;~\pi)=k_{1}(\metgraph;~0)$ and by Lemma
\ref{lem:SG_paramter_properties} we deduce $\sgp\leq\pi$. Hence,
condition \eqref{enu:thm_SGP_condition_3} of Theorem \ref{thm:supremum_of_gluing}
implies condition \eqref{enu:prop_SGP_condition_2} of Proposition
\ref{prop:subadditive_gluing_principle}.

Assuming conditions \eqref{enu:thm_SGP_condition_1},\eqref{enu:thm_SGP_condition_3}
of the theorem we may now apply Proposition \ref{prop:subadditive_gluing_principle}
and get 
\begin{equation}
k_{1}\left(\metgraph\right)=k_{1}\left(\metgraph_{1}\right)+k_{1}\left(\metgraph_{2}\right).\label{eq:thm_supremum_of_gluing_1}
\end{equation}
Let $\hat{\metgraph}$ be a supremizer of $\disgraph$. In particular,
$k_{1}(\metgraph)\leq k_{1}(\hat{\metgraph})$. Denote by $\hat{\metgraph}_{1},\hat{\metgraph}_{2}$
the subgraphs of $\hat{\metgraph}$ corresponding to $\disgraph_{1},\disgraph_{2}$
and rescaled such that the total length of each of them equals $1$.
By Proposition \ref{prop:subadditive_gluing_principle} 
\begin{equation}
k_{1}\left(\hat{\metgraph}\right)\leq k_{1}\left(\hat{\metgraph}_{1}\right)+k_{1}\left(\hat{\metgraph}_{2}\right).\label{eq:thm_supremum_of_gluing_2}
\end{equation}
Hence we get 
\[
k_{1}\left(\hat{\metgraph}\right)\leq k_{1}\left(\hat{\metgraph}_{1}\right)+k_{1}\left(\hat{\metgraph}_{2}\right)\leq k_{1}\left(\metgraph_{1}\right)+k_{1}\left(\metgraph_{2}\right)=k_{1}\left(\metgraph\right),
\]
where the second inequality holds as $\metgraph_{1},\metgraph_{2}$
are supremizers. We therefore get that $k_{1}(\metgraph)=k_{1}(\hat{\metgraph})$,
so that $\metgraph$ is a supremizer of $\disgraph$ as $\hat{\metgraph}$
is a supremizer of $\disgraph$ (and possibly $\metgraph=\hat{\metgraph}$). 

We now further assume that either for both $i=1,2$ $\metgraph_{i}$
is the unique supremizer of $\disgraph_{i}$ or that both $\Gamma_{1},\metgraph_{2}$
obey the strong Dirichlet criterion and any other supremizer violates
the Dirichlet criterion. Assume that $\hat{\metgraph}$ is a supremizer
of $\disgraph$ so that $k_{1}(\metgraph)=k_{1}(\hat{\metgraph})$.
From \eqref{eq:thm_supremum_of_gluing_1}, \eqref{eq:thm_supremum_of_gluing_2}
we get 
\begin{equation}
k_{1}\left(\metgraph_{1}\right)+k_{1}\left(\metgraph_{2}\right)\leq k_{1}\left(\hat{\metgraph}_{1}\right)+k_{1}\left(\hat{\metgraph}_{2}\right).\label{eq:thm_supremum_of_gluing_3}
\end{equation}
As $\metgraph_{1},\metgraph_{2}$ are supremizers of $\disgraph_{1},\disgraph_{2}$,
we have an equality in \eqref{eq:thm_supremum_of_gluing_3} and get
that for both $i=1,2$, $k_{1}\left(\metgraph_{i}\right)=k_{1}\left(\hat{\metgraph}_{i}\right)$,
so that $\hat{\metgraph}_{1},~\hat{\metgraph}_{2}$ are supremizers
of $\disgraph_{1},\disgraph_{2}$ as well. If both $\metgraph_{1},\metgraph_{2}$
are unique supremizers of $\disgraph_{1},\disgraph_{2}$ then $\metgraph_{i}=\hat{\metgraph}_{i}$
for both $i=1,2$. Hence, $\hat{\metgraph}=\metgraph$. 

We carry on by assuming that both $\Gamma_{1},\metgraph_{2}$ obey
the strong Dirichlet criterion and any other supremizer violates the
Dirichlet criterion. From Lemma \ref{lem:SG_paramter_properties}
we deduce that a graph violates the Dirichlet criterion if and only
if its spectral gap parameter satisfies $\sgp\in\left(\pi,2\pi\right]$.
If for both $i=1,2$, $\hat{\metgraph}_{i}$ is different than $\metgraph_{i}$,
then we have $\sgp_{1},\sgp_{2}\in\left(\pi,2\pi\right]$ and by Proposition
\ref{prop:subadditive_gluing_principle} we have the strict inequality
\begin{equation}
k_{1}\left(\hat{\metgraph}\right)<k_{1}\left(\hat{\metgraph}_{1}\right)+k_{1}\left(\hat{\metgraph}_{2}\right),\label{eq:thm_supremum_of_gluing_4}
\end{equation}
which together with 
\[
k_{1}\left(\hat{\metgraph}_{1}\right)+k_{1}\left(\hat{\metgraph}_{2}\right)=k_{1}\left(\metgraph_{1}\right)+k_{1}\left(\metgraph_{2}\right)=k_{1}\left(\metgraph\right)
\]
contradicts $\hat{\metgraph}$ being a supremizer. From Lemma \ref{lem:multiplicity_at_crossings}
which follows this proof we deduce that a graph obeys the strong Dirichlet
criterion if and only if its SGP equals $\pi$. Therefore, if $\hat{\metgraph}_{i}=\metgraph_{i}$
for either $i=1$ or $i=2$, say $\hat{\metgraph}_{1}=\metgraph_{1}$,
then we have $\sgp_{1}=\pi$ and $\sgp_{2}\in\left(\pi,2\pi\right]$
and once again we get by Proposition \ref{prop:subadditive_gluing_principle}
the inequality \eqref{eq:thm_supremum_of_gluing_4} which contradicts
$\hat{\metgraph}$ being a supremizer.
\end{proof}
\begin{lem}
\label{lem:multiplicity_at_crossings}Let $k\in\Delta(\metgraph)$.
Let $n\in\mathbb{N}$ and $\theta\in(-\pi,\pi]$ such that $k=K(\theta+2n\pi)$.
Assume that $k$ has multiplicity $m+1$ in the spectrum $\sigma(\metgraph;~\theta)$.
Then, for any $\theta'\neq\theta$, $k$ has a multiplicity $m$ as
an eigenvalue in the spectrum $\sigma(\metgraph;~\theta')$.
\end{lem}
\begin{proof}
Since $\Delta(\metgraph)$ is a discrete set, for $k'<k$ sufficiently
close to $k$, $k'$ does not belong to $\Delta(\metgraph)$. Thus,
for $\theta'<\theta$ sufficiently close to $\theta$, $K(\theta'+2n\pi)$
is not in $\Delta(\metgraph)$. We define $a\in\mathbb{N}$ as the
unique integer satisfying $K(\theta'+2n\pi)=k_{a}(\metgraph,\theta')$
for all $\theta'<\theta$ sufficiently close to $\theta$. Since $K(\cdot+2n\pi)$
and $k(\metgraph,\cdot)$ are continuous functions of their arguments
(see Lemma \ref{lem:spectrum_structure} and Lemma \ref{lem:continuity_wrt_delta_parameter}),
letting $\theta'$ go to $\theta$ gives 
\[
k=k_{a}(\metgraph,\theta).
\]
If $\theta\neq\pi$, we may argue similarly with $\theta'>\theta$
sufficiently close to $\theta$ to find that 
\[
k=k_{b}(\metgraph,\theta)<k_{b}(\metgraph,\theta').
\]
Notice that since $a$ and $b$ are respectively minimal and maximal
integers such that $k=k_{a}(\metgraph,\theta)=k_{b}(\metgraph,\theta)$,
the multiplicity assumption on $k$ in $\sigma(\metgraph;~\theta)$
entails $b=a+m$. As $K$ is strictly increasing and by Lemma \ref{lem:interlacing_wrt_delta_parameters},
we get
\[
\forall\theta'\in(-\pi,\theta),\qquad~\quad k_{a}(\metgraph;~\theta')<k=k_{a+1}(\metgraph;~\theta')=\cdots=k_{b}(\metgraph;~\theta')<k_{b+1}(\metgraph;~\theta')
\]
and

\[
\forall\theta'\in(\theta,\pi],\qquad~\quad k_{a-1}(\metgraph;~\theta')<k=k_{a}(\metgraph;~\theta')=\cdots=k_{b-1}(\metgraph;~\theta')<k_{b}(\metgraph;~\theta').
\]
We conclude from these inequalities that $k$ has multiplicity $m$
in $\sigma(\metgraph;\theta')$ for all $\theta'\neq\theta$.

If $\theta=\pi$, we have 
\[
\forall\theta'\neq\pi,\qquad~\quad k=k_{b}(\metgraph,\pi)<k_{b+1}(\metgraph,\theta'),
\]
and once again
\[
\forall\theta'\neq\pi,\qquad~\quad k_{a}(\metgraph;~\theta')<k=k_{a+1}(\metgraph;~\theta')=\cdots=k_{b}(\metgraph;~\theta')<k_{b+1}(\metgraph;~\theta'),
\]
from which the result follows.
\end{proof}

\section{Symmetrization of dangling edges and loops \label{sec:Symmetrization}}
\begin{prop}
\label{prop:symmetrization}Let $\disgraph$ be a graph with $E\geq3$
edges. Let $v$ be a vertex of $\disgraph$ and $e_{1},e_{2}$ either
two dangling edges or two loops connected to $v$. Let $l_{1},l_{2}$
be the lengths of those edges and denote their average by $\ell:=\frac{1}{2}\left(l_{1}+l_{2}\right)$. 

Denoting $\tilde{\metgraph}:=\metgraph\left(\disgraph;~(l_{1},l_{2},l_{3},\ldots,l_{E})\right)$,
$\metgraph:=\metgraph\left(\disgraph;~(\ell,\ell,l_{3},\ldots,l_{E})\right)$,
we have 
\begin{equation}
k_{1}\left(\tilde{\metgraph}\right)\leq k_{1}\left(\metgraph\right).\label{eq:symmetrization}
\end{equation}
Moreover, if either $k_{1}\left(\metgraph\right)=\frac{\pi}{2\ell}$
in the dangling edges case (respectively, $k_{1}\left(\metgraph\right)=\frac{\pi}{\ell}$
in the loops case) or alternatively both the following conditions
are satisfied

\begin{enumerate}
\item \label{enu:prop_symmetrization_1}$\metgraph$ is a supremizer of
some graph.
\item \label{enu:prop_symmetrization_2}$k_{1}(\tilde{\metgraph})$ is a
simple eigenvalue.
\end{enumerate}
then equality above holds if and only if $l_{1}=l_{2}$.

\end{prop}
\begin{proof}
Let $f$ be an eigenfunction of $\metgraph$ corresponding to $k_{1}\left(\metgraph\right)$.
The proof for both cases - dangling edges and loops - is by constructing
a test function $\tilde{f}$ on $\tilde{\metgraph}$, whose Rayleigh
quotient obeys $\ray(\tilde{f})\leq\ray(f)=k_{1}\left(\metgraph\right)^{2}$,
from which \eqref{eq:symmetrization} follows.

We start with the dangling edges case. First, we get a bound on $k_{1}\left(\metgraph\right)$
using a test function, 
\begin{equation}
\left.g\right|_{e_{1}\cup e_{2}}=\cos\left(\frac{\pi x}{2\ell}\right),~~~~\left.g\right|_{\metgraph\backslash\left(e_{1}\cup e_{2}\right)}=0,
\end{equation}
where $e_{1}\cup e_{2}$ is considered as single interval. We have
$\ray(g)=\left(\frac{\pi}{2\ell}\right)^{2}$ and hence $k_{1}\left(\metgraph\right)\leq\frac{\pi}{2\ell}$.

Assume that $k_{1}\left(\metgraph\right)=\frac{\pi}{2\ell}$. Let
$\tilde{f}$ be the following test function on $\tilde{\metgraph}$.
\[
\left.\tilde{f}\right|_{\tilde{e}_{1}\cup\tilde{e}_{2}}=\cos\left(\frac{\pi x}{2\ell}\right),~~~~\left.\tilde{f}\right|_{\tilde{\metgraph}\backslash\left(\tilde{e}_{1}\cup\tilde{e}_{2}\right)}=\tilde{f}\left(v\right),
\]
 where $\tilde{f}\left(v\right)$ in the right equation is determined
from the value $\left.\tilde{f}\right|_{\tilde{e}_{1}\cup\tilde{e}_{2}}$
on the left attains at $v$. As $\tilde{f}$ is not necessarily orthogonal
to the constant function, we actually take $\tilde{f}-\left\langle \tilde{f}\right\rangle $
to be the test function, where $\left\langle \tilde{f}\right\rangle :=\int_{\tilde{\metgraph}}\tilde{f}dx$.
By Lemma \ref{lem:appendix_test_functions_extension_constant} 
\begin{equation}
\ray\left(\tilde{f}-\left\langle \tilde{f}\right\rangle \right)=\frac{\left(\frac{\pi}{2\ell}\right)^{2}\ell}{\ell+\left|\tilde{f}\left(v\right)\right|^{2}2\ell\left(1-2\ell\right)}<\left(\frac{\pi}{2\ell}\right)^{2}=\left(k_{1}\left(\metgraph\right)\right)^{2},\label{eq:prop_symmetrization_Rayleigh}
\end{equation}

where we use that $l_{1}\neq l_{2}\Rightarrow\tilde{f}\left(v\right)=\cos(\frac{\pi l_{1}}{2\ell})\neq0$
to get the inequality.

Next, assume $k_{1}\left(\metgraph\right)<\frac{\pi}{2\ell}$ and
also that $f(v)=0$. Then $f$ has to identically vanish on both $e_{1}$
and $e_{2}$. We may then choose the test function $\tilde{f}=f$
and get $\ray\left(\tilde{f}\right)=\ray\left(f\right)$, as required.

Finally, assume $k_{1}\left(\metgraph\right)<\frac{\pi}{2\ell}$ and
$f(v)\neq0$. This results with $\left.f\right|_{e_{1}}=\left.f\right|_{e_{2}}$.
Assume without loss of generality that $l_{1}<l_{2}$. We define the
test function $\tilde{f}$ on $\tilde{\metgraph}$ as follows.
\[
\left.\tilde{f}\right|_{\tilde{\metgraph}\backslash\left(\tilde{e}_{1}\cup\tilde{e}_{2}\right)}=\left.f\right|_{\metgraph\backslash\left(e_{1}\cup e_{2}\right)},~~~~\left.\tilde{f}\right|_{\tilde{e}_{1}}=\left.f\right|_{e_{1}\left(0,l_{1}\right)},
\]
where $e_{1}\left(0,l_{1}\right)$ denotes a subset of $e_{1}$ in
$\metgraph$ whose origin is $v$. On $\tilde{e}_{2}$ we set 
\[
\left.\tilde{f}\right|_{\tilde{e}_{2}}\left(x\right)=\begin{cases}
\left.f\right|_{e_{2}}\left(x\right) & ~~x\in\left(0,\ell\right)\\
\left.f\right|_{e_{1}}\left(l_{1}+l_{2}-x\right) & ~~x\in\left(\ell,l_{2}\right)
\end{cases}.
\]
This is a valid continuous test function and by construction, $\ray(\tilde{f})=\ray(f)$. 

We have therefore shown inequality \eqref{eq:symmetrization} and
also that assuming $k_{1}\left(\metgraph\right)=\frac{\pi}{2\ell}$
assures equivalence between $l_{1}=l_{2}$ and equality in \eqref{eq:symmetrization}.
It is therefore left to show that under assumptions \eqref{enu:prop_symmetrization_1},\eqref{enu:prop_symmetrization_2}
of the proposition, $l_{1}\neq l_{2}$ implies $k_{1}(\tilde{\metgraph})<k_{1}(\metgraph$).
Assume by contradiction that $l_{1}\neq l_{2}$ and also $k_{1}(\tilde{\metgraph})=k_{1}(\metgraph$).
As $\metgraph$ is a supremizer of some graph, $\tilde{\metgraph}$
is also a supremizer of the same graph. Since $k_{1}(\tilde{\metgraph})$
is simple we deduce from Lemma \ref{lem:supremal_and_simple_is_critical}
that its spectral gap is a critical value and by Lemma \ref{lem:critical_point_derivative_at_vertices}
we get $\left|\frac{\partial}{\partial x_{\tilde{e}_{1}}}\tilde{f}\left(v\right)\right|=\left|\frac{\partial}{\partial x_{\tilde{e}_{2}}}\tilde{f}\left(v\right)\right|$,
where $\tilde{f}$ is the eigenfunction corresponding to $k_{1}(\tilde{\metgraph})$.
If $\tilde{f}(v)=0$ we get that $\tilde{f}$ has at least three nodal
domains (at least one nodal domain on each of $\tilde{e}_{1},~\tilde{e}_{2}$
and $\tilde{\metgraph}\backslash\{\tilde{e}_{1}\cup\tilde{e}\}$),
which contradicts Courant's nodal theorem (\cite{Cou_ngwgmp23,GnuSmiWeb_wrm04,Ber_cmp08}).
Assume without loss of generality $\tilde{f}(v)>0$. As $l_{1}\neq l_{2}$
and as the derivative of $\tilde{f}$ vanishes at the endpoints of
$\tilde{e}_{1},~\tilde{e}_{2}$, we get that at least one of $\tilde{e}_{1},~\tilde{e}_{2}$
should contain two nodal domains of $\tilde{f}$. In addition, by
Courant's bound it is not possible for both derivatives, $\frac{\partial}{\partial x_{\tilde{e}_{1}}}\tilde{f}\left(v\right),~\frac{\partial}{\partial x_{\tilde{e}_{2}}}\tilde{f}\left(v\right)$
to be negative as this results with a total of at least three nodal
domains. If one derivative is positive and the second is negative,
i.e., $\frac{\partial}{\partial x_{\tilde{e}_{1}}}\tilde{f}\left(v\right)=-\frac{\partial}{\partial x_{\tilde{e}_{2}}}\tilde{f}\left(v\right)$,
we get that $f\vert_{\tilde{e}_{1}\cup\tilde{e}_{2}}$ is proportional
to $\cos(\frac{\pi}{2\ell}x)$, so that $k_{1}(\tilde{\metgraph})=\frac{\pi}{2\ell}$,
which is a contradiction, to what we have shown above (see \eqref{eq:prop_symmetrization_Rayleigh}).
If both derivatives are positive, $\frac{\partial}{\partial x_{\tilde{e}_{1}}}\tilde{f}\left(v\right)=\frac{\partial}{\partial x_{\tilde{e}_{2}}}\tilde{f}\left(v\right)$,
then we get contradiction as $\left\langle \tilde{f}\right\rangle \neq0$.
Indeed, assuming without loss of generality $l_{1}<l_{2}$, the restriction
of $\tilde{f}$ on an interval of length $l_{2}-l_{1}$ at the end
of edge $\tilde{e}_{2}$ is of zero mean, but the $\tilde{f}'s$ restriction
to the rest of the graph is positive, as $\tilde{f}$ has only two
nodal domains and therefore.

We turn to deal with the loops case. Just as above, we start by getting
an upper bound on the spectral gap. Choose the following test function
on $\metgraph$ 
\begin{equation}
\left.g\right|_{e_{1}\cup e_{2}}=\cos\left(\frac{\pi x}{\ell}\right)~~~~\left.g\right|_{\metgraph\backslash\left(e_{1}\cup e_{2}\right)}=0,
\end{equation}
where $e_{1}\cup e_{2}$ is considered as single cycle (self intersecting
itself at its middle). In this case, $\ray\left(g\right)=\left(\frac{\pi}{\ell}\right)^{2}$
so that $k_{1}\left(\metgraph\right)\leq\frac{\pi}{\ell}$.

\noindent The proof now splits into three cases exactly as it was
for the dangling edges:

\begin{enumerate}
\item If $k_{1}\left(\metgraph\right)=\frac{\pi}{\ell}$, we may construct
a test function $\tilde{f}$ on $\tilde{\metgraph}$, such that $\mathcal{R}(\tilde{f})\leq\mathcal{R}(f)$
and with equality only if $l_{1}=l_{2}$.
\item If $k_{1}\left(\metgraph\right)<\frac{\pi}{\ell}$ and $f\left(v\right)=0$,
we conclude that $f$ identically vanishes on the edges $e_{1},e_{2}$
and we may construct a test function $\tilde{f}$ on $\tilde{\metgraph}$,
such that $\mathcal{R}(\tilde{f})=\mathcal{R}(f)$.
\item If $k_{1}\left(\metgraph\right)<\frac{\pi}{\ell}$ and $f\left(v\right)\neq0$,
we conclude that both $\left.f\right|_{e_{1}}$ and $\left.f\right|_{e_{2}}$
are symmetric functions and write 
\[
\left.f\right|_{e_{i}}=A_{i}\cos\left(k_{1}\left(\metgraph\right)\cdot x\right),
\]
for $x\in\left(-\frac{\ell}{2},~\frac{\ell}{2}\right)$ and $A_{i}\in\mathbb{R}$.
Construct a test function $\tilde{f}$ on $\tilde{\metgraph}$ by
setting
\[
\left.\tilde{f}\right|_{\tilde{\metgraph}\backslash\left(\tilde{e}_{1}\cup\tilde{e}_{2}\right)}=\left.f\right|_{\metgraph\backslash\left(e_{1}\cup e_{2}\right)},
\]
and 
\[
\left.\tilde{f}\right|_{e_{i}}\left(x\right)=A_{i}\cos\left(k_{1}\left(\metgraph\right)\left|x-\frac{l_{i}-\ell}{2}\right|\right)~~\textrm{for}~~x\in\left(-\frac{l_{i}}{2},~\frac{l_{i}}{2}\right).
\]
This last relation pictorially means that if $\tilde{e}_{1}$ is the
shorter edge, $\tilde{f}\at_{\tilde{e}_{1}}$ is a symmetric function
which equals $\left.f\right|_{e_{1}}$ up to a piece of length $\ell-l_{1}$
around the middle of the edge $e_{1}$ which is glued to the middle
of the the edge $e_{2}$. Overall, $\tilde{f}$ has zero mean and
$\ray(\tilde{f})=\ray(f)$, as required. 
\end{enumerate}
Just as above, assumptions \eqref{enu:prop_symmetrization_1},\eqref{enu:prop_symmetrization_2}
of the proposition together with assuming $l_{1}\neq l_{2}$ and $k_{1}(\tilde{\metgraph})=k_{1}(\metgraph)$,
enables to use Lemmata \ref{lem:supremal_and_simple_is_critical}
and \ref{lem:critical_point_derivative_at_vertices} together with
Courant's bound to arrive at a contradiction.
\end{proof}
An immediate generalization of this proposition is the following. 
\begin{cor}
\label{cor:symmetrization_of_many} Let $\disgraph$ be a graph with
$E\geq3$ edges. Let $n\geq2$ be an integer. Let $v$ be a vertex
of $\disgraph$ and $e_{1},\ldots,e_{n}$ be either $n$ dangling
edges or $n$ loops connected to $v$. Denote by $l_{1},\ldots,l_{n}$
the lengths of those edges and by $l_{n+1},\ldots,l_{E}$ the lengths
of all other edges. Defining 
\[
\ell:=\frac{1}{n}\sum_{i=1}^{n}l_{i},
\]
and denoting $\tilde{\metgraph}:=\metgraph\left(\disgraph;~(l_{1},\ldots,l_{n},l_{n+1},\ldots,l_{E})\right)$,
$\metgraph:=\metgraph\left(\disgraph;~(\ell,\ldots,\ell,l_{n+1},\ldots,l_{E})\right)$,we
have 
\begin{equation}
k_{1}(\tilde{\metgraph})\leq k_{1}(\metgraph).
\end{equation}
Moreover, if either $k_{1}(\metgraph)=\frac{\pi}{2\ell}$ in the dangling
edges case (respectively, $k_{1}(\metgraph)=\frac{\pi}{\ell}$ in
the loops case) or alternatively both the following conditions are
satisfied

\begin{enumerate}
\item \label{enu:cor_symmetrization_of_many_1} $\metgraph$ is a supremizer
of some graph.
\item \label{enu:cor_symmetrization_of_many_2} $k_{1}(\tilde{\metgraph})$
is a simple eigenvalue.
\end{enumerate}
then equality above holds if and only if $l_{j}=\ell$, for all $1\leq j\leq n$.

\end{cor}
\begin{proof}
Denote by $\vec{L}$ the vector of lengths $(l_{1},\ldots,l_{n})$,
and by $k_{1}(l_{1},\ldots,l_{n})$ the corresponding spectral gap,
keeping all the other $E-n$ edge lengths fixed. Assume without loss
of generality that $l_{1}\geq\ldots\geq l_{n}$. If $l_{1}>l_{n}$,
we replace these two lengths by $\frac{1}{2}(l_{1}+l_{n})$ and get
by Proposition \ref{prop:symmetrization} that 
\begin{equation}
k_{1}(l_{1},\ldots,l_{n})\leq k_{1}\left(\frac{1}{2}(l_{1}+l_{n}),~l_{2},\ldots,l_{n-1},~\frac{1}{2}(l_{1}+l_{n})\right).
\end{equation}
Repeating this process infinitely many times, we get a sequence of
vectors 
\[
\left\{ \vec{L}^{\left(m\right)}\right\} _{m=1}^{\infty}:=\left\{ (l_{1}^{\left(m\right)},\ldots,l_{n}^{\left(m\right)})\right\} _{m=1}^{\infty}
\]
 such that 

\begin{itemize}
\item $l_{1}^{(m)}\geq\ldots\geq l_{n}^{(m)}$ (up to reordering the lengths), 
\item $\frac{1}{n}\sum_{i=1}^{n}l_{i}^{\left(m\right)}=\ell$
\item $l_{1}^{(m)}-l_{n}^{(m)}\rightarrow0$ as $m\rightarrow\infty$
\item the sequence $\left\{ k_{1}(\ell_{1}^{\left(m\right)},\ldots,\ell_{n}^{\left(m\right)})\right\} {}_{m=1}^{\infty}$
is non-decreasing
\end{itemize}
From the first three claims we deduce that, $l_{j}^{(m)}\to\ell\text{ as \ensuremath{m\to\infty}}$,
for any $1\leq j\leq n$. Therefore, the continuity of eigenvalues
with respect to edge lengths (see Appendix \ref{sec:appendix_eigenvalue_continuity})
gives 
\[
k_{1}(l_{1}^{\left(m\right)},\ldots,l_{n}^{\left(m\right)})\to k_{1}(\ell,\ldots,\ell)\text{ as \ensuremath{m\to\infty}.}
\]
 As the sequence $\left\{ k_{1}(\ell_{1}^{\left(m\right)},\ldots,\ell_{n}^{\left(m\right)})\right\} {}_{m=1}^{\infty}$
is non-decreasing it follows that 
\begin{equation}
k_{1}(l_{1},\ldots,l_{n})\leq k_{1}(\ell,\ldots,\ell),
\end{equation}
as desired.

We now turn to the strict inequality conditions. In the dangling edge
case, if the spectral gap satisfies $k_{1}(\ell,\ldots,\ell)=\frac{\pi}{2\ell}$,
then particular eigenfunctions are given by that of the equilateral
star with $n$ edges and total length $n\ell$. Among them, we choose
one supported only on two edges and repeat the argument given in Proposition
\ref{prop:symmetrization} to deduce the strict inequality if $l_{i}\neq l_{j}$
for some $i\neq j$. We argue similarly if $k_{1}(\ell,\ldots,\ell)=\frac{\pi}{\ell}$
in the dangling loops case. Alternatively, we may assume by contradiction
that there exist $i\neq j$ such that $l_{i}\neq l_{j}$ and $k_{1}(\tilde{\metgraph})=k_{1}\left(\metgraph\right)$.
This together with assumptions \eqref{enu:cor_symmetrization_of_many_1},\eqref{enu:cor_symmetrization_of_many_2}
enables to proceed exactly as in the proof of Proposition \ref{prop:symmetrization}
in order to get a contradiction.
\end{proof}

\section{Applications of graph gluing and symmetrization \label{sec:Applications_of_Gluing_and_Symmetrization}}

This section applies the techniques of graph gluing and edge symmetrization
developed in the previous two sections in order to prove the next
few corollaries.
\begin{proof}
[Proof of Corollary \ref{cor:gluing_stowers}]

This proof is a direct application of Theorem \ref{thm:supremum_of_gluing}
once we observe the following

\begin{enumerate}
\item The glued vertices, $v_{1},v_{2}$ become the central vertices of
the supremizing stowers.
\item Every equilateral stower obeys the Dirichlet criterion with respect
to its internal vertex, assuming the numbers of its petals and leaves
obey $E_{p}+E_{l}\geq2$. 
\item Denoting the supremizing stowers by $\metgraph_{1},\metgraph_{2}$,
their spectral gaps are 
\[
k_{1}\left(\Gamma_{i}\right)=\frac{\pi}{2}\left(2E_{p}^{\left(i\right)}+E_{l}^{\left(i\right)}\right).
\]
\item Gluing $\metgraph_{1},\metgraph_{2}$ with the length parameter 
\[
L=\frac{k_{1}\left(\metgraph_{1}\right)}{k_{1}\left(\metgraph_{1}\right)+k_{1}\left(\metgraph_{2}\right)}=\frac{2E_{p}^{\left(1\right)}+E_{l}^{\left(1\right)}}{2E_{p}^{\left(1\right)}+E_{l}^{\left(1\right)}+2E_{p}^{\left(2\right)}+E_{l}^{\left(2\right)}},
\]
results with an equilateral stower whose all petals are of length
$\frac{2}{2E_{p}^{\left(1\right)}+E_{l}^{\left(1\right)}+2E_{p}^{\left(2\right)}+E_{l}^{\left(2\right)}}$
and all dangling edges are of length $\frac{1}{2E_{p}^{\left(1\right)}+E_{l}^{\left(1\right)}+2E_{p}^{\left(2\right)}+E_{l}^{\left(2\right)}}$. 
\end{enumerate}
\end{proof}
\begin{rem*}
We note that an equilateral stower obeys the strong Dirichlet criterion.
Therefore, by Theorem \ref{thm:supremum_of_gluing}, if we assume
for $\disgraph_{1},\disgraph_{2}$ that all their supremizers other
than the stower violate the Dirichlet criterion, we also get uniqueness
in Corollary \ref{cor:gluing_stowers}.
\end{rem*}
~
\begin{proof}
[Proof of Corollary \ref{cor:stower_supremum}] We show that equilateral
stars and flowers (with $E\geq2)$ satisfy condition (b) of Theorem
\ref{thm:supremum_of_gluing}, when considered as supremizers of the
corresponding stowers. This allows to employ Theorem \ref{thm:supremum_of_gluing}
in order to glue a star with a flower and to show the statement of
the Corollary for all stowers with $E_{l}\geq2$ and $E_{p}\geq2$
(note that when gluing an equilateral flower and equilateral star
according to condition \eqref{enu:thm_SGP_condition_1} of Theorem
\ref{thm:supremum_of_gluing}, the stower obtained is equilateral).
The rest of the stowers will be dealt with, at the end of the proof.

Start by noting that Theorem \ref{thm:tree_supremizer} implies that
the statement of the corollary holds for all star graphs, which are
stowers with $E_{p}=0,~E_{l}\geq2$. The spectral gap of equilateral
star is $\frac{E\pi}{2}$ and it remains the same after imposing Dirichlet
condition at their central vertex, so that it obeys the Dirichlet
criterion. Furthermore the multiplicity of its spectral gap is $E-1$
and it increases to $E$ after imposing Dirichlet condition, so that
it obeys the strong Dirichlet criterion. As equilateral stars are
unique maximizers of stars, we conclude that they obey condition (b)
of Theorem \ref{thm:supremum_of_gluing}.

Among the flower graphs, we start with the two-petal and three-petal
flowers. An easy calculation reveals that the spectral gap of a flower
with two petals equals $2\pi$. Note that this spectral gap is independent
of the edge lengths, so that this give a continuous family of (trivial)
maximizers. In particular, the equilateral flower with two petals
is a non-unique maximizer. Yet, this equilateral two-petal flower
is the only maximizer in this family which obeys the Dirichlet criterion
and it further obeys the strong Dirichlet criterion, as we show next.
Consider a two-petal flower whose edge lengths are $l_{1}\neq l_{2}$
and assume $l_{1}>l_{2}$. Imposing Dirichlet condition at the vertex
lowers the spectral gap of the graph from $2\pi$ to $\nicefrac{\pi}{l_{1}}$,
so that it does not obey the Dirichlet criterion. The equilateral
flower, on the other hand, maintains the spectral gap of $2\pi$ even
after imposing a Dirichlet condition at its vertex. In addition, its
spectral gap with Neumann condition at the vertex is a simple eigenvalue,
but once imposing Dirichlet at the vertex, the spectral gap becomes
of multiplicity two. By this we have shown that the two-petal flower
satisfies condition (b) of Theorem \ref{thm:supremum_of_gluing}. 

Let $\Gamma$ be a flower with three petals and denote its vertex
by $v$. Let $\tilde{\Gamma}$ be the two petal subgraph which consists
of the largest two petals of $\Gamma$. Denote the total length of
$\tilde{\metgraph}$ by $\tilde{l}$ (so that $\tilde{l}\geq\frac{2}{3}$).
Let $\tilde{f}$ be the first non-constant eigenfunction on $\tilde{\metgraph}$.
Construct the following test function on $\metgraph$ 
\[
\left.f\right|_{\tilde{\metgraph}}=\tilde{f},~~~~\left.f\right|_{\metgraph\backslash\tilde{\metgraph}}=\tilde{f}\left(v\right).
\]
By Lemma \ref{lem:appendix_test_functions_extension_constant}
\[
\ray\left(f-\left\langle f\right\rangle \right)=\frac{\left(\frac{2\pi}{\tilde{l}}\right)^{2}\frac{\tilde{l}}{2}}{\frac{\tilde{l}}{2}+|\tilde{f}\left(v\right)|^{2}\tilde{l}\left(1-\tilde{l}\right)}\leq\left(\frac{2\pi}{\tilde{l}}\right)^{2}\leq\left(3\pi\right)^{2},
\]
where equality holds if and only if $\tilde{l}=\nicefrac{2}{3}$ and
$\tilde{f}\left(v\right)=0$. Conversely, it is easy to show that
the spectral gap of the equilateral three-petal flower equals $3\pi$.
Hence the equilateral three petal graph is a unique maximizer. In
addition, imposing a Dirichlet condition at the vertex maintains a
spectral gap of $3\pi$, so that the three-petal equilateral flower
obeys the Dirichlet criterion. It further obeys the strong Dirichlet
criterion as the multiplicity of its spectral gap is $2$ and it increases
to $3$ after imposing Dirichlet condition at central vertex. Therefore,
a three petal flower satisfies condition (b) of Theorem \ref{thm:supremum_of_gluing}.

From the above, we may glue two flowers of those types (each either
with two petals or three petals) and get a four, five or six petal
flower. Applying Theorem \ref{thm:supremum_of_gluing} shows that
the equilateral version of each of these graphs serves as the unique
maximizer. Furthermore, it is easy to show that any equilateral flower
obeys the strong Dirichlet criterion (as was shown for the two-petal
and three-petal flower above). This together with the uniqueness of
four, five and six petal flowers implies that they obey condition
(b) of Theorem \ref{thm:supremum_of_gluing}. Repeating this gluing
process as many times as needed shows that any equilateral flower
is both a unique maximizer (except for $E=2$ ) and obeys condition
(b) of Theorem \ref{thm:supremum_of_gluing} (which holds also for
$E=2$).

By this, we have both proved the corollary for all stars and flowers
with $E\geq2$ and also conclude the validity of the corollary for
all stowers with $E_{l}\geq2$ and $E_{p}\geq2$, as claimed in the
beginning of this proof. It is left to treat stowers with either $E_{p}=1$
or $E_{l}=1$. In order to do that, we state in lemmata \ref{lem:stower_1_2},
\ref{lem:stower_1_3}, \ref{lem:stower_2_1}, \ref{lem:stower_3_1}
(which follow this proof) that the current corollary is valid for
the following small stowers $\left(E_{p},E_{l}\right)\in\left\{ \left(1,2\right),\left(1,3\right),\left(2,1\right),\left(3,1\right)\right\} $
and that in addition, the equilateral versions of those stowers all
obey condition (b) of Theorem \ref{thm:supremum_of_gluing}. Hence,
each stower with either $E_{p}=1$ or $E_{l}=1$ may be obtained by
gluing one of those small stowers with an appropriate flower or star
and applying Theorem \ref{thm:supremum_of_gluing} for such a gluing
finishes the proof.
\end{proof}
\begin{rem*}
We note that the proof above might have been simplified if we were
after a weaker result. Namely, using the more elementary methods of
Rayleigh quotient calculations one can prove the statement in the
Corollary for all stowers except those with $E_{p}=1$ or $E_{l}=1$
and without the uniqueness part of the result.
\end{rem*}
\begin{proof}
[Proof of Corollary \ref{cor:upper_bound_for_supremum}] Let $\disgraph$
be a graph with $E$ edges out of which $E_{l}$ are leaves and $E-E_{l}$
are internal edges. Let $\lenvec\in\len_{\disgraph}$ and denote $\metgraph:=\metgraph(\disgraph;~\lenvec)$.
Identifying all internal (i.e. non-leaf) vertices of $\metgraph$
we get a stower graph with $E_{l}$ leaves and $E-E_{l}$ petals which
we denote by $\tilde{\metgraph}$ and by Lemma \ref{lem:interlacing_of_gluing}
we get 
\[
k_{1}(\metgraph)\leq k_{1}(\tilde{\metgraph}).
\]
From Corollary \ref{cor:stower_supremum} we have 
\begin{equation}
k_{1}(\tilde{\metgraph})\leq\pi\left((E-E_{l})+\frac{E_{l}}{2}\right)=\pi\left(E-\frac{E_{l}}{2}\right)\label{eq:cor_upper_bound_of_supremum}
\end{equation}
 if $E\geq2$ and $(E,E_{l})\neq(2,1)$ which are exactly the conditions
in this corollary and this proves its first part.

Assuming equality in \eqref{eq:global_upper_bound} we have equality
in \eqref{eq:cor_upper_bound_of_supremum}. If we further assume $(E,E_{l})\notin\left\{ \left(2,0\right),\,\left(3,2\right)\right\} $,
we satisfy the uniqueness conditions in Corollary \ref{cor:stower_supremum}.
Namely, we conclude that equality in \eqref{eq:cor_upper_bound_of_supremum}
is possible only if $\tilde{\metgraph}$ is equilateral in the stower
sense: leaves are of half length than petals. We conclude that $\metgraph$
is also equilateral in the following sense: all of its leaves are
of length $\frac{1}{2E-E_{l}}$ each and all the rest (inner) edges
are of length $\frac{2}{2E-E_{l}}$ each. We carry on by conditioning
on the number of internal (i.e. non-leaf) vertices of $\metgraph$
and keeping in mind that $k_{1}(\metgraph)=\pi\left(E-\frac{E_{l}}{2}\right)$.

If $\metgraph$ has a single internal vertex then it is a stower graph
and we are done. Assume that $\metgraph$ has at least two internal
vertices. Choose two such internal vertices. In the following we described
a recursive process which marks some set of edges of the graphs, to
be denoted by $\E_{0}$. Choose a path on $\metgraph$ connecting
$v_{+}$ with $v_{-}$ without going through graph leaves. This is
possible as $\metgraph$ is connected. Choose an arbitrary edge, $e$,
on this path and add it to $\E_{0}$. Next, if $\metgraph\backslash e$
is connected repeat the step above on $\metgraph\backslash e$. Namely,
choose a path on $\metgraph\backslash e$ connecting $v_{+}$ and
$v_{-}$ not going through graph leaves (with the exception of $v_{+},v_{-}$
which might have now turned themselves into leaves). Repeat this process
until $\metgraph\backslash\E_{0}$ is a disconnected graph. We may
then write $\metgraph=\metgraph_{+}\cup\metgraph_{-}\cup\E_{0}$,
where $\metgraph_{+}$ is a connected subgraph of $\metgraph$ containing
$v_{+}$, and similarly for $\metgraph_{-}$ and $v_{-}$. Set the
following test function on $\metgraph$:
\[
f\left(x\right)=\begin{cases}
1 & ~~x\in\metgraph_{+}\\
-1 & ~~x\in\metgraph_{-}\\
\cos\left(k_{1}(\metgraph)\cdot x\right) & ~~x\in e~s.t.~e\in\E_{0}.
\end{cases}
\]
By construction, this test function is continuous. It is easy to verify
by \eqref{eq:rayleigh_with_mean} (alternatively, by an easy extension
of Lemma \ref{lem:appendix_test_functions_extension_constant}) that
$\ray(f)<k_{1}(\metgraph)$ if $\metgraph_{+}\cup\metgraph_{-}\neq\emptyset$.
As $\ray(f)<k_{1}(\metgraph)$ contradicts the equality in \eqref{eq:global_upper_bound}
we conclude that $\metgraph_{+}\cup\metgraph_{-}\neq\emptyset$, which
implies that $\metgraph=\E_{0}$ and hence $\metgraph$ is a mandarin
graph. It is actually an equilateral mandarin, as we have shown above. 
\end{proof}
The lemmata needed in the proof of Corollary \ref{cor:stower_supremum}
are now stated. Their proofs involve some technical computations and
appear in Appendix \ref{sec:appendix_small_stowers}.
\begin{lem}
\label{lem:stower_1_2} Let $\mathcal{G}$ be a stower with $E_{p}=1$
petal and $E_{l}=2$ leaves. Then $\mathcal{G}$ has a continuous
family of maximizers whose spectral gap is $2\pi$. Those are all
the stowers with both leaf lengths equal and not greater than $\nicefrac{1}{4}$.
Furthermore, the equilateral stower obeys condition (b) of Theorem
\ref{thm:supremum_of_gluing}.
\end{lem}
~ 
\begin{lem}
\label{lem:stower_1_3}Let $\mathcal{G}$ be a stower graph with $E_{p}=1$
petal and $E_{l}=3$ leaves. Then the equilateral stower graph is
the unique maximizer of $\disgraph$, and the corresponding spectral
gap equals $\frac{5\pi}{2}$. Furthermore, the equilateral stower
obeys condition (b) of Theorem \ref{thm:supremum_of_gluing}.
\end{lem}
~ 
\begin{lem}
\label{lem:stower_2_1} Let $\disgraph$ be a stower graph with $E_{l}=1$
and $E_{p}=2$. Then $\disgraph$ has a unique maximizer, which is
the equilateral stower graph with spectral gap equal to $\frac{5\pi}{2}$.
Furthermore, the equilateral stower obeys condition (b) of Theorem
\ref{thm:supremum_of_gluing}. 
\end{lem}
~ 
\begin{lem}
\label{lem:stower_3_1} Let $\disgraph$ be a stower graph with $E_{l}=1$
and $E_{p}=3$. Then $\disgraph$ has a unique maximizer, which is
the equilateral stower graph with spectral gap equal to $\frac{7\pi}{2}$.
Furthermore, the equilateral stower obeys condition (b) of Theorem
\ref{thm:supremum_of_gluing}.

~~ 
\end{lem}
The stower with $E_{p}=E_{l}=1$ was not mentioned in the theorem
above, as it is not maximized by the equilateral stower. Its unique
supremizer is the single loop graph ($E_{p}=1,~E_{l}=0$), as we state
in the following in order to complete the picture.
\begin{lem}
\label{lem:stower_1_1} Let $\mathcal{G}$ be a stower graph with
one leaf and one petal. Then $\mathcal{G}$ has a unique maximizer,
which is the unit circle, with spectral gap equal to $2\pi$. 
\end{lem}

\section{Summary}

This work investigates the problem of optimizing a graph's spectral
gap in terms of its edge lengths. We start by providing a natural
formulation of this problem (Definitions \ref{def:discrete_and_metric_graph},\ref{def:optimizers}
and adjacent discussion). Our formalism allows both to state the optimization
questions in utmost generality (for all graph topologies and all edge
length values) and moreover to determine when such a question is fully
answered. For example, this is the case with the infimization problem
for which both the optimal bounds and all the possible infimizing
topologies are found, with no more room for improvement (see the discussion
which follows Theorem \ref{thm:infimizers}). Contrary to the infimization
problem, we point out that the supremization problem is not solved
in full generality. We show its complete solution for tree graphs
and for a family of graphs whose vertex connectivity equals one. In
addition, a global upper bound in provided (Corollary \ref{cor:upper_bound_for_supremum}),
improving the upper bound known so far, by taking into account the
number of graph leaves. Furthermore, we provide a set of techniques
to tackle the supremization problem. Among those are the gluing graphs
approach, the symmetrization of dangling edges and loops and the characterization
of local maximizers. Those tools are applicable in the current work
and might assist in further exploration of the problem. The techniques
and the results of the current work lead to forming a few conjectures
regarding the maximization problem.

First, the supremizer graph families known so far are stower graphs
(including stars and flowers as particular cases) and mandarin graphs.
The spectral gap of these graphs is highly degenerate due to their
large symmetry groups. The symmetry groups corresponding to the stower
and the mandarin are correspondingly $S_{E_{p}}\times S_{E_{l}}$
and $S_{E}$, where $E$ is the number of mandarin edges and $E_{p},E_{l}$
numbers of stower petals and leaves. The corresponding spectral gap
multiplicity of both a stower and a mandarin is $E-1$, which is indeed
high. In the other extreme of spectral gaps which are simple eigenvalues,
we show that those are unlikely to be supremizers. In Theorem \ref{thm:simple_supremum}
we prove that a supremizer whose spectral gap is simple can never
have a spectral gap higher than a mandarin and in some cases than
a flower (Corollary \ref{cor:simple_supremum}). In Proposition \ref{prop:subadditive_gluing_principle}
we prove that if a supremizer is obtained by the gluing method then
its spectral gap is necessarily a multiple eigenvalue. As high multiplicities
of eigenvalues is related to large order symmetry groups (or even
to large dimension of their representations), the discussion above
leads to the following two conjectures:
\begin{enumerate}
\item A supremizer of a graph is obtained by choosing edge lengths which
maximize the order of the symmetry group of the resulting graph\footnote{We thank Gregory Berkolaiko for raising this conjecture in a private
communication.}.
\item A supremizer of a graph is obtained by choosing edge lengths which
maximize the multiplicity of the spectral gap.
\end{enumerate}
We note that the conjectures above are not necessarily correlated.
We demonstrate this by mandarin chains, which are $M$ copies of $n$-mandarin
graphs glued serially, as presented in Proposition \ref{prop:Standarins}.
The symmetry group of those graphs is $(S_{n})^{M}$ whose order is
$(n!)^{M}$. Yet, a mandarin chain with $n\geq2,~M\geq2$ always has
a simple spectral gap, as proved in Proposition \ref{prop:Standarins}.
Hence, the large order of the symmetry group does not guarantee large
multiplicity of the spectral gap. Seeking for supremizers for those
graphs, we observe that turning such a graph into an equilateral flower
with $m(n-1)$ petals, increases its spectral gap from $n\pi$ to
$M(n-1)\pi$. The symmetry group of this flower is $S_{M(n-1)}$,
which is of order $(M(n-1))!$. For most values of $n,M$, the flower's
symmetry group is of larger order than that of the mandarin chain,
which is correlated to its spectral gap being of higher multiplicity.
However, for $n=3,~M=2$, the symmetry group of the flower is of order
$24$, while that of the mandarin chain is of order $36$. This flower
possesses a higher spectral gap ($3\pi$) than the mandarin chain
($4\pi$) despite its lower order symmetry group. On one hand, this
example serves in the favor of the second conjecture over the first
one. On the other hand, we still do not know what is the supremizer
in this example and feel that at this stage, both conjectures are
equally appealing.

Finally, we state a more explicit conjecture: the supremizer of a
certain graph is either a stower graph (in its generalized sense)
or a mandarin. These are indeed the only supremizers this work revealed.
Given a certain graph, the maximal spectral gap among all stowers
which may be obtained from that graph equals $\pi(\beta+\frac{E_{l}}{2})$,
where $\beta$ is the graph's first Betti number and $E_{l}$ is the
number of its dangling edges. The maximal spectral gap among all possible
mandarins has a less explicit expression, and we describe it next.
Let $\disgraph$ be a graph and let $\disgraph_{1},\disgraph_{2}$
be two connected subgraphs, sharing neither an edge nor a vertex and
such that each vertex of $\disgraph$ belongs to $\disgraph_{1}\cup\disgraph_{2}$.
Let $E(\disgraph_{1},\disgraph_{2})$ be the number of edges connecting
a vertex of $\disgraph_{1}$ to a vertex of $\disgraph_{2}$. Contracting
all edges of $\disgraph_{1}$and $\disgraph_{2}$ we get a mandarin
of $E(\disgraph_{1},\disgraph_{2})$ edges. The maximal spectral gap
among all mandarins is therefore given by 
\[
\pi\cdot\max_{\disgraph_{1},\disgraph_{2}}E(\disgraph_{1},\disgraph_{2}).
\]
We note that the expression above is curiously related to the Cheeger
constant, but do not further elaborate on that. For the allowed $(\disgraph_{1},\disgraph_{2})$
partitions among which we maximize we may also write $E(\disgraph_{1},\disgraph_{2})=\beta+1-\left(\beta_{1}+\beta_{2}\right)$,
where $\beta_{i}$ is the first Betti number of $\disgraph_{i}$.
This expression allows for a comparison with the optimal stower spectral
gap, $\pi(\beta+\frac{E_{l}}{2})$. For example it is seen that for
a graph with at most one dangling edge, the mandarin achieves a strictly
higher spectral gap than the stower (or flower in this case) only
if there is a partition where both $\disgraph_{1},\disgraph_{2}$
are tree graphs. On the other hand, if the graph has at least three
dangling edges, any mandarin has a lower spectral gap than the optimal
stower. Does the conjecture above hold or are there supremizers other
than stowers and mandarins? This question remains open.

\section*{Acknowledgments}

We acknowledge Richard Maynes for taking part in the preliminary examination
of the problem. We thank Gregory Berkolaiko and Uzy Smilansky for
their stimulating feedback. We thank Adam Sawicki and his student,
Oskar S\l{}owik, for some fruitful discussions regarding graph connectivity.
We thank Sebastian Egger and Lior Alon, as well as the anonymous referees
for a careful reading and useful comments.

R.B. was supported by ISF (Grant No. 494/14), Marie Curie Actions
(Grant No. PCIG13-GA-2013-618468) and the Taub Foundation (Taub Fellow).

G.L. thanks the Mathematics faculty of the Technion for their kind
hospitality, without which the current collaboration would not have
been possible.

\appendix

\section{Eigenvalue continuity with respect to edge lengths\label{sec:appendix_eigenvalue_continuity}}

In this section we sketch a proof for the continuity of all the graph's
eigenvalues (not only the spectral gap) with respect to the graph's
edge lengths. The continuity (and even differentiability) of eigenvalues
with respect to edge lengths is proven in \cite{BerKuc_quantum_graphs,DelRoss_amp16}.
Yet, those proofs deal only with positive edge lengths\footnote{It is possible that the proof in section 4 of \cite{DelRoss_amp16},
which is based on test functions, may be adapted for the zero edge
length case. Nevertheless, we provide here a different argument based
on the scattering approach.}, whereas in the current work we are interested in particular in $\lenvec\in\partial\len_{\disgraph}$,
when we distinguish between supremizers and maximizers (see definition
\ref{def:optimizers}). We claim that eigenvalue continuity indeed
carries over to the zero edge length case. We do not prove this in
full rigor, but rather point out the general lines for forming a proof
for this statement. We start by introducing the scattering approach
for quantum graphs (see also \cite{GnuSmi_ap06,BerKuc_quantum_graphs}).

\subsection{The scattering approach to the graph spectrum}

\label{sec:scat_mat} Let $\metgraph$ be a Neumann graph. The eigenvalue
equation, 
\begin{equation}
-\frac{d^{2}f}{dx^{2}}=k^{2}f(x)\ ,\label{eq:eig_eq_V0}
\end{equation}
has a solution on each directed edge $e$, written as (assuming $k\neq0$)
\begin{equation}
f_{e}(x_{e})=a_{e}^{\textrm{in}}\ue^{-\ui kx_{e}}+a_{e}^{\textrm{out}}\ue^{\ui kx_{e}}.\label{eq:sol_on_edge_a}
\end{equation}
We may consider the edge $\hat{e}$, which is the same as $e$, but
with a reverse direction (resulting in different parametrization of
the coordinate, $x_{\hat{e}}=l_{e}-x_{e}$) and write the same function
as above in the following form 
\begin{equation}
f_{\hat{e}}(x_{\hat{e}})=a_{\hat{e}}^{\textrm{in}}\ue^{-\ui kx_{\hat{e}}}+a_{\hat{e}}^{\textrm{out}}\ue^{\ui kx_{\hat{e}}}\ .\label{eq:sol_on_edge_other}
\end{equation}
Comparing both expressions above we arrive at 
\begin{equation}
a_{e}^{\textrm{in}}=\ue^{\ui kl_{e}}a_{\hat{e}}^{\textrm{out}}\qquad\textrm{and}\qquad a_{\hat{e}}^{\textrm{in}}=\ue^{\ui kl_{e}}a_{e}^{\textrm{out}}\ .\label{eq:a_connection}
\end{equation}
Fixing a vertex $v$ and using the Neumann vertex conditions to relate
solutions $f_{e}$ for all edges whose origin is $v$ one arrives
at 
\begin{equation}
\vec{a}_{v}^{\,\textrm{out}}=\sigma^{\left(v\right)}\vec{a}_{v}^{\,\textrm{in}},\label{eq:vert_scat}
\end{equation}
where $\vec{a}_{v}^{\,\textrm{out}}$ and $\vec{a}_{v}^{\,\textrm{in}}$
are vectors of the outgoing and incoming coefficients ($a_{e}^{\textrm{in}},~a_{e}^{\textrm{out}}$)
at $v$ and $\sigma^{(v)}$ is a $d_{v}\times d_{v}$ unitary matrix,
$d_{v}$ being the degree of the vertex $v$. The matrix $\sigma^{(v)}$
is called the vertex-scattering matrix and its entries were first
calculated in \cite{KotSmi_anp99}: 
\begin{equation}
\sigma_{e,e'}^{\left(v\right)}=\frac{2}{d_{v}}-\delta_{e,e'}\ .
\end{equation}
We collect all coefficients $a_{e}^{\textrm{in}}$ from the whole
graph into a vector $\vec{a}$ of size $2E$ such that the first $E$
entries correspond to edges which are the inverses of the last $E$
entries. We can then define the matrix $J$ acting on $\vec{a}$ by
requiring that it exchanges $a_{e}^{\textrm{in}}$ and $a_{\hat{e}}^{\textrm{in}}$
for all $e$ such that, 
\begin{equation}
J=\left(\begin{array}{cc}
0 & \mathbf{I}\\
\mathbf{I} & 0
\end{array}\right).
\end{equation}
Then, collecting equations (\ref{eq:vert_scat}) for all vertices
into one system and using (\ref{eq:a_connection}) we have, 
\begin{equation}
J\ue^{-\ui kL}\vec{a}=\Sigma\vec{a}\ ,
\end{equation}
where $L=\textrm{diag}\{l_{1},\dots,l_{E},l_{1},\dots,l_{E}\}$ is
a diagonal matrix of edge lengths and $\Sigma$ is block-diagonalizable
with individual $\sigma^{(v)}$ as blocks. This can be rewritten as
(note that $J^{-1}=J$), 
\begin{equation}
\vec{a}=\ue^{\ui kL}J\Sigma\vec{a}\ ,\label{eq:bond_scat}
\end{equation}
and hence all the non zero eigenvalues of the graph are the solutions
of 
\begin{equation}
\det\left(\mathbf{I}-U\left(k\right)\right)=0\ ,\label{eq:secular_cond}
\end{equation}
where $U(k):=\ue^{\ui kL}J\Sigma$. 

\subsection{Continuity of eigenvalues via scattering approach}

The scattering approach allows for a reduction in the dimensions of
the matrix $U(k)$ by reducing a subgraph into a single composite
vertex with some (non-trivial) vertex conditions (see section 3.3
in \cite{GnuSmi_ap06}). We pick a certain edge, $e$, to be the mentioned
subgraph and turn it into a single (composite) vertex by shrinking
it to zero length.

The length of this edge, $l_{e}$, will show up only in the scattering
matrix of this composite vertex and will allow to examine how the
eigenvalues depend on this length. We carry on with an explicit computation.
Let $e$ be an edge connecting two vertices, $v_{1},v_{2}$, of degrees
$d_{1},d_{2}$. Hence, the new composite vertex, $v$, would be of
degree $d_{1}+d_{2}-2$. We calculate a reflection coefficient of
this vertex (i.e., an on-diagonal entry of its vertex-scattering matrix).
The calculation may be done by summing infinitely many trajectories
on the original graph all starting by entering $v_{1}$ from some
edge $e_{1}$ (different than $e$) and eventually leaving $v_{1}$
along the same edge, $e_{1}$ (see section 3.3 in \cite{GnuSmi_ap06},
for further details). 
\begin{align}
\sigma_{e_{1},e_{1}}^{\left(v\right)} & =\frac{2-d_{1}}{d_{1}}+\frac{2}{d_{1}}\cdot\ue^{\ui k2l_{e}}\cdot\frac{2-d_{2}}{d_{2}}\cdot\sum_{n=0}^{\infty}\left(\ue^{\ui k2l_{e}}\frac{2-d_{2}}{d_{2}}\frac{2-d_{1}}{d_{1}}\right)^{n}\cdot\frac{2}{d_{1}}\nonumber \\
 & =-1+\frac{2}{d_{1}}\left(1+\frac{4-2d_{2}}{\ue^{-\ui k2l_{e}}d_{1}d_{2}-\left(2-d_{1}\right)\left(2-d_{2}\right)}\right)\underset{_{l_{e}\rightarrow0}}{\longrightarrow}-1+\frac{2}{d_{1}+d_{2}-2}.\label{eq:appendix_continuity_reflection}
\end{align}
where the continuity of the expression above in $l_{e}$ is apparent
and allows to take the limit $l_{e}\rightarrow0$. We calculate just
another entry of the composite vertex scattering matrix - the entry
which corresponds to entering at vertex $v_{1}$ and leaving at $v_{2}$.
The calculation is similar to the one above and gives 
\begin{align}
\sigma_{e_{1},e_{2}}^{\left(v\right)} & =\frac{2}{d_{1}}\cdot\ue^{\ui kl_{e}}\cdot\sum_{n=0}^{\infty}\left(\ue^{\ui k2l_{e}}\frac{2-d_{2}}{d_{2}}\frac{2-d_{1}}{d_{1}}\right)^{n}\cdot\frac{2}{d_{2}}\nonumber \\
 & =\frac{4}{\ue^{-\ui kl_{e}}d_{1}d_{2}-\ue^{\ui kl_{e}}\left(2-d_{1}\right)\left(2-d_{2}\right)}\underset{_{l_{e}\rightarrow0}}{\longrightarrow}\frac{2}{d_{1}+d_{2}-2}.\label{eq:appendix_continuity_transmission}
\end{align}
There is just another computation which is similar in nature and will
not be repeated here. All the rest of the composite vertex scattering
matrix entries may be obtained by symmetry. We hence get that the
resulting scattering matrix when taking the limit $l_{e}\rightarrow0$
is the same as the one obtained by considering Neumann conditions
at the composite vertex. As the scattering matrix continuously determines
the graph's eigenvalues (see \eqref{eq:secular_cond}) we get the
desired continuity result.

\section{$\delta$-type conditions and interlacing theorems\label{sec:appendix_interlacing_theorems}}

We present here the so-called $\delta$-type conditions, of which
both Neumann and Dirichlet conditions form special cases.
\begin{defn}
\label{def:delta_type_conditions} We say that $f$ satisfies the
$\delta$-type condition with the coefficient $\alpha\in\R$ at vertex
$v$ if

\begin{enumerate}
\item $f$ is continuous at $v$: 
\begin{equation}
f_{e_{1}}(v)=f_{e_{2}}(v),\label{eq:delta_continuity}
\end{equation}
 for all edges $e_{1},e_{2}\in\E_{v}$, where $\E_{v}$ is the set
of edges incident to $v$. 
\item the derivatives of $f$ at $v$ satisfy 
\begin{equation}
\sum_{e\in\E_{v}}\frac{df}{dx_{e}}(v)=\alpha f(v).\label{eq:delta_deriv_with_alpha}
\end{equation}
\end{enumerate}
\end{defn}
We consider the following transformations 
\begin{equation}
\alpha\mapsto\theta=\arg\left(\frac{\ui+\alpha}{\ui+\alpha}\right),\label{eq:alpha_to_theta}
\end{equation}
and 
\begin{equation}
\theta\mapsto\alpha=\ui\frac{1-\exp\left(\ui\theta\right)}{1+\exp\left(\ui\theta\right)}=\tan\left(\frac{\theta}{2}\right).\label{eq:theta_to_alpha}
\end{equation}

The transformations \eqref{eq:alpha_to_theta}, \eqref{eq:theta_to_alpha}
are the inverses one of the other and allow to write the condition
\eqref{eq:delta_deriv_with_alpha} in the form \eqref{eq:delta-condition_with_theta},
which is the one used throughout the paper. We denote by $k_{n}(\metgraph;~\theta)$
the $n^{\textrm{th}}$ $k$-eigenvalue of such a graph and possibly
omit either $\metgraph$ or $\theta$ from this notation whenever
it is clear what they are from the context. Similarly, the spectrum
is denoted $\sigma(\metgraph;~\theta)$ (see \eqref{eq:spectrum_notation}).

We quote below some useful results from \cite{BerKuc_quantum_graphs}
as lemmata.

The following lemma is a slight rephrasing of theorem 3.1.8 from \cite{BerKuc_quantum_graphs}.
\begin{lem}
\label{lem:interlacing_wrt_delta_parameters} Let $\metgraph$ be
a compact (not necessarily connected) graph. Let $v$ be a vertex
of $\metgraph$ endowed with the $\delta$-type condition and arbitrary
self-adjoint vertex conditions at all other vertices of $\metgraph$.
If $-\pi<\theta<\theta'\leq\pi$, then 
\[
k_{n}\left(\theta\right)\leq k_{n}\left(\theta'\right)\leq k_{n+1}\left(\theta\right).
\]

If the eigenvalue $k_{n}\left(\theta'\right)$ is simple and its eigenfunction
$f$ is such that either $f\left(v\right)$ or $\sum f'\left(v\right)$
is non-zero, then the inequalities above are strict, 
\[
k_{n}\left(\theta\right)<k_{n}\left(\theta'\right)<k_{n+1}\left(\theta\right).
\]
\end{lem}
The following lemma is a slight rephrasing of theorem 3.1.10 from
\cite{BerKuc_quantum_graphs}.
\begin{lem}
\label{lem:interlacing_of_gluing} Let $\metgraph$ be a compact (not
necessarily connected) graph. Let $v_{1}$ and $v_{2}$ be vertices
of $\metgraph$ endowed with the $\delta$-type conditions with corresponding
coefficients $\alpha_{1},\alpha_{2}$ and arbitrary self-adjoint vertex
conditions at all other vertices of $\metgraph$. Let $\metgraph'$
be the graph obtained from $\metgraph$by gluing the vertices $v_{1}$
and $v_{2}$ together into a single vertex $v$, so that $\E_{v}=\E_{v_{1}}\cup\E_{v_{2}}$
and endowed with $\delta$-type condition at $v$, with the coefficient
$\alpha_{1}+\alpha_{2}$. 

Then the eigenvalues of the two graphs satisfy the inequalities
\[
k_{n}\left(\metgraph\right)\leq k_{n}\left(\metgraph'\right)\leq k_{n+1}\left(\metgraph\right).
\]

We apply the lemma above in the case $\alpha_{1}=-\alpha_{2}$, for
which $\metgraph'$ satisfies Neumann conditions at $v$.
\end{lem}
The following lemma is a rephrasing of part of lemma 3.1.14 from \cite{BerKuc_quantum_graphs}
and the discussion which precedes it.
\begin{lem}
\label{lem:continuity_wrt_delta_parameter} $k_{n}\left(\theta\right)$
is a continuous non-decreasing function of $\theta\in\left(-\pi,\pi\right]$
and obeys the following continuity relation
\[
k_{n}\left(\pi\right)=\lim_{\theta\rightarrow-\pi^{+}}k_{n+1}\left(\theta\right).
\]
\end{lem}
~

The following lemma contains a statement which is proved in the course
of the proof of lemma 3.1.15 in \cite{BerKuc_quantum_graphs}. We
state here the lemma we need and its proof for completeness.
\begin{lem}
\label{lem:flat_band_eigenvalue}Let $\Gamma$ be a graph and let
$v$ be a vertex of $\Gamma$. Let $\theta_{1}\neq\theta_{2}$ and
let $k\in\sigma\left(\Gamma;~\theta_{1}\right)\cap\sigma\left(\Gamma;~\theta_{2}\right)$.
Then there exists an eigenfunction corresponding to $k$ which vanishes
at $v$ and its sum of derivatives vanish at $v$. Therefore, this
eigenfunction satisfies the $\delta$-type condition at $v$ for every
$\theta\in\left(-\pi,\pi\right]$. Hence $k\in\Delta\left(\metgraph\right)$.
\end{lem}
\begin{proof}
Let $f_{1},f_{2}$ the eigenfunctions corresponding to $k\in\sigma\left(\Gamma;~\theta_{1}\right)\cap\sigma\left(\Gamma;~\theta_{2}\right)$,
with coefficients $\theta_{1},\theta_{2}$, respectively. Assume first
that either $k\in\sigma\left(\Gamma;~\theta_{1}\right)$ or $k\in\sigma\left(\Gamma;~\theta_{2}\right)$
is a multiple eigenvalue. Assume without loss of generality that it
is $k\in\sigma\left(\Gamma;~\theta_{1}\right)$. Further assume that
$\theta_{1}\neq\pi$. As the eigenvalue is multiple, we can choose
a corresponding eigenfunction which vanishes at $v$ and denote it
by $f_{1}$. We deduce from the $\delta$-type condition that the
sum of derivatives of $f_{1}$ at $v$ vanishes as well and conclude
that $f_{1}$ satisfies $\delta$-type condition at $v$ for any value
of $\theta$. If we assume $\theta_{1}=\pi$, then we may use the
multiplicity of the eigenvalue to choose an eigenfunction $f_{1}$
whose sum of derivatives at $v$ vanishes and once again conclude
that $f_{1}$ satisfies $\delta$-type condition at $v$ for any value
of $\theta$. We have shown that the lemma holds if one of the eigenvalues
is multiple. Otherwise, assume that $k\in\sigma\left(\Gamma;~\theta_{1}\right)$
and $k\in\sigma\left(\Gamma;~\theta_{2}\right)$ are simple eigenvalues.
Assume without loss of generality that $\theta_{1}\neq\pi$. Let $f_{1}$
be the eigenfunction corresponding to $k$ and satisfying the $\delta$-type
condition with $\theta_{1}$. If $f_{1}\left(v\right)\neq0$, then
the strict eigenvalue interlacing (Lemma \ref{lem:interlacing_wrt_delta_parameters})
contradicts $k\in\sigma\left(\Gamma;~\theta_{1}\right)\cap\sigma\left(\Gamma;~\theta_{2}\right)$.
Therefore $f_{1}\left(v\right)=0$ and the sum of derivatives of $f_{1}$
vanishes at $v$, due to the $\delta$-type condition.
\end{proof}

\section{A Basic Rayleigh quotient computation\label{sec:appendix_test_functions_using_subgraphs}}

In the current section, we develop a basic but useful bound on the
Rayleigh quotient, which is used throughout the paper. We define the
mean of a function on a graph as 
\begin{equation}
\left\langle f\right\rangle :=\int_{\metgraph}f\,\ud x,\label{eq:mean_of_function}
\end{equation}
and observe that 
\begin{equation}
\ray\left(f-\left\langle f\right\rangle \right)=\frac{\int_{\metgraph}|f'|^{2}\,\ud x}{\int_{\metgraph}f^{2}\,\ud x-\left\langle f\right\rangle ^{2}},\label{eq:rayleigh_with_mean}
\end{equation}
which is useful as the test functions for which the Rayleigh quotient
is computed ought to be of zero mean.
\begin{lem}
\label{lem:appendix_test_functions_extension_constant} Let $\metgraph$
be a graph of length $1$. Assume that $\metgraph=\metgraph_{1}\cup\metgraph_{2}$
where $\metgraph_{1,2}$ are subgraphs of $\metgraph$ such that $\metgraph_{1}\cap\metgraph_{2}$
is a single vertex, denoted by $v$. Choose an eigenfunction $f$
on $\metgraph_{1}$ corresponding to $k_{1}(\metgraph_{1})$ and extend
it to $\metgraph_{2}$ by the constant $f(v)$. The resulting test
function on $\metgraph$, denoted $\tilde{f}$, satisfies 
\begin{equation}
\ray(\tilde{f}-\left\langle \tilde{f}\right\rangle )=\frac{k_{1}(\Gamma_{1})^{2}\left(\int_{\metgraph_{1}}f^{2}\ud x\right)}{\left(\int_{\metgraph_{1}}f^{2}\ud x\right)+|f(v)|^{2}l_{2}(1-l_{2})},
\end{equation}
where $l_{2}$ denotes the total length of $\metgraph_{2}$. 
\end{lem}
\begin{proof}
We compute the mean and the $L^{2}$ norm of $\tilde{f}$: 
\[
\langle\tilde{f}\rangle=\int_{\metgraph_{2}}f(v)dx=f(v)l_{2}
\]
and 
\[
\int_{\Gamma}|\tilde{f}|^{2}dx=\left(\int_{\metgraph_{1}}f^{2}\ud x\right)+\int_{\Gamma_{2}}|f(v)|^{2}dx=\left(\int_{\metgraph_{1}}f^{2}\ud x\right)+|f(v)|^{2}l_{2}.
\]
As $\tilde{f}$ is constant on $\Gamma_{2}$ and $f$ is an eigenfunction
on $\Gamma_{1}$, we have 
\[
\int_{\Gamma}|\tilde{f}'(x)|^{2}\ud x=\int_{\Gamma_{1}}|f'(x)|^{2}\ud x=k_{1}(\Gamma_{1})^{2}\left(\int_{\metgraph_{1}}f^{2}\ud x\right).
\]
Plugging the above in \eqref{eq:rayleigh_with_mean} gives the desired
result.
\end{proof}
An immediate corollary of Lemma \ref{lem:appendix_test_functions_extension_constant}
is the following. 
\begin{cor}
\label{cor:extension_constant} With the notations above we have $k_{1}(\Gamma)\leq k_{1}(\Gamma_{1})$.
This inequality is strict if there exists an eigenfunction of $k_{1}(\Gamma_{1})$
not vanishing at $v$. 
\end{cor}
In the decomposition discussed above, $\metgraph=\metgraph_{1}\cup\metgraph_{2}$,
we call $\Gamma_{1}$ the \emph{main subgraph} of $\Gamma$ and $\Gamma_{2}$
the \emph{attached subgraph}. Note that when the main subgraph is
a single loop, we may rotate its eigenfunction so that it achieves
its maximal value at $v$. We exploit this in the sequel when applying
Lemma \ref{lem:appendix_test_functions_extension_constant}, since
this choice leads to a low value of the Rayleigh quotient.

\section{Proofs for small stowers (Lemmata \ref{lem:stower_1_2}-\ref{lem:stower_1_1})
\label{sec:appendix_small_stowers}}

In this more technical Appendix, we extensively use Lemma \ref{lem:appendix_test_functions_extension_constant}.
Namely, we consider the decomposition $\metgraph=\metgraph_{1}\cup\metgraph_{2}$
and refer to $\Gamma_{1,2}$ as either the main or the attached subgraph
of $\Gamma$ (see Appendix \ref{sec:appendix_test_functions_using_subgraphs}).

~
\begin{proof}[Proof of Lemma \ref{lem:stower_1_2}]
Let us denote by $l_{1},l_{2}$ and $l_{p}$ the lengths of the two
leaves and the petal, respectively and by $v$ the vertex of degree
three. Denote by $k_{1}(l_{1},l_{2},l_{p})$ the spectral gap corresponding
to these edge lengths. First, if $l_{1}+l_{2}>\frac{1}{2}$, we use
the interval made of the two leaves as the main subgraph and the petal
as the attached subgraph. We thus get, in this case, the inequality
$k_{1}(l_{1},l_{2},l_{p})<2\pi$. Now, if $l_{1}+l_{2}\leq\frac{1}{2}$
and $l_{1}=l_{2}$, explicit calculations show that the spectral gap
is equal to $2\pi$. Applying the symmetrization principle on the
leaves (Proposition \ref{prop:symmetrization}) shows that whenever
$l_{1}+l_{2}\leq\frac{1}{2}$ and $l_{1}\neq l_{2}$, we have $k_{1}(l_{1},l_{2},l_{p})\leq2\pi$.
We further wish to prove that this inequality is strict and do so
by checking the assumptions in Proposition \ref{prop:symmetrization}.
Assumption \eqref{enu:prop_symmetrization_1} is valid as we have
shown above that the stower with $l_{1}=l_{2}\leq\frac{1}{4}$ is
a supremizer. We now check assumption \eqref{enu:prop_symmetrization_2}
- that whenever $0\leq l_{1}<l_{2}$ and $l_{1}+l_{2}\leq\frac{1}{2}$
the corresponding spectral gap is simple. In turn, thanks to Proposition
\ref{prop:symmetrization}, we will get the strict inequality $k_{1}(l_{1},l_{2},l_{p})<2\pi$
for $l_{1}\neq l_{2}$ and $l_{1}+l_{2}\leq\frac{1}{2}$. Assume by
contradiction that there exist $0\leq l_{1}<l_{2}$ with $l_{1}+l_{2}\leq\frac{1}{2}$
such that the spectral gap $k_{1}(l_{1},l_{2},l_{p})$ is not simple.
Thanks to the multiplicity, we may choose an eigenfunction vanishing
at $v$. Since $l_{1}<\frac{1}{4}$, such an eigenfunction has to
vanish on the whole edge $e_{1}$ for otherwise, the spectral gap
would satisfy $k_{1}(l_{1},l_{2},l_{p})\geq\frac{\pi}{2l_{1}}>2\pi$.
Furthermore, the eigenfunction does not identically vanish neither
on $e_{2}$ (again, this would contradict the bound on $k_{1}$) nor
on $e_{p}$ (because of the Neumann condition at $v$). Thus, there
exist two integers $\alpha,\beta$ with $\alpha$ odd such that $k_{1}(l_{1},l_{2},l_{p})=\frac{\alpha\pi}{2l_{2}}=\frac{\beta\pi}{l_{p}}$.
From the bound on $k_{1}(l_{1},l_{2},l_{p})$ and the conditions on
the lengths, we get $\alpha=\beta=1$. But as $k_{1}(l_{1},l_{2},l_{p})=\frac{\pi}{2l_{2}}$
and $l_{1}\neq l_{2}$, \emph{all} eigenfunctions should vanish at
$v$. Using again multiplicity, we may choose another eigenfunction
which vanishes at $v$ and at another point on $e_{2}$, call it $w$.
But this contradicts the equality $k_{1}(l_{1},l_{2},l_{p})=\frac{\pi}{2l_{2}}$,
hence the simplicity. We have therefore found a continuous family
of maximizers - all stowers with $l_{1}=l_{2}\leq\frac{1}{4}$. It
is easy to check that among all those, only the equilateral stower
satisfies the Dirichlet criterion. In addition, the multiplicity of
the spectral gap increases from two to three when imposing the Dirichlet
condition at the central vertex, which is exactly the strong Dirichlet
criterion. Hence, the equilateral stower satisfies condition (b) of
Theorem \ref{thm:supremum_of_gluing}.
\end{proof}
~
\begin{proof}[Proof of Lemma \ref{lem:stower_1_3}]
Denote by $\metgraph$ the metric graph corresponding to $\disgraph$,
whose length of the petal is $l_{p}$ and lengths of the leaves are
$l_{1},l_{2},l_{3}$ (so that $l_{p}+l_{1}+l_{2}+l_{3}=1$). Assume
for instance that $l_{1}\geq l_{2}\geq l_{3}$ and denote $\ell:=\frac{l_{1}+l_{2}+l_{3}}{3}$.
Using the three leaves a main subgraph and the petal as an attached
subgraph, we get the inequality 
\[
k_{1}\left(\metgraph\right)\leq\frac{\pi}{2\ell}.
\]
On the other hand, using the petal and the longest two leaves as a
main subgraph and the shortest leaf as an attached subgraph, we use
Lemma \ref{lem:stower_1_2} to get 
\[
k_{1}\left(\metgraph\right)\leq\frac{2\pi}{1-l_{3}}.
\]
Combining these two inequalities, 
\[
k_{1}\left(\metgraph\right)\leq\min\left(\frac{\pi}{2\ell},\frac{2\pi}{1-l_{3}}\right)\leq\min\left(\frac{\pi}{2\ell},\frac{2\pi}{1-\ell}\right).
\]
This immediately yields, for any choice of $l_{l}$, 
\[
k_{1}\left(\metgraph\right)\leq\frac{5\pi}{2},
\]
with equality possible only if $\ell=\frac{1}{5}$ and $l_{3}=\ell$.
These two conditions together imply $l_{1}=l_{2}=l_{3}=\frac{1}{5}$
and $l_{p}=\frac{2}{5}$. Conversely, for this specific choice of
lengths, it is straightforward to point out the eigenfunction whose
$k$-eigenvalue equals $\frac{5\pi}{2}$. Furthermore, it is also
easy to check that in this case, the spectral gap indeed equals $\frac{5\pi}{2}$,
with multiplicity three. Furthermore, imposing the Dirichlet condition
at the central vertex increases the multiplicity of the spectral gap
from three to four. Hence, the equilateral stower satisfies the strong
Dirichlet criterion and is a unique supremizer, which proves that
the equilateral stower satisfies condition (b) of Theorem \ref{thm:supremum_of_gluing}.
\end{proof}
~
\begin{proof}[Proof of Lemma \ref{lem:stower_2_1}]
Let us denote by $l_{1},l_{2}$ and $l_{l}$ the lengths of the two
petals and the leaf, respectively. Denote $\ell:=\frac{l_{1}+l_{2}}{2}$.
From Proposition \ref{prop:symmetrization}, we have the inequality
$k_{1}(l_{1},l_{2},l_{l})\leq k_{1}(\ell,\ell,l_{l})$. We now focus
on the case where $l_{1}=l_{2}=\ell$. Let $v$ be the central vertex
of the stower. Using the two petals as a main subgraph and the leaf
as an attached subgraph, we get 
\[
k_{1}(\ell,\ell,l_{l})\leq\frac{2\pi}{1-l_{l}}.
\]
Thus, for $0\leq l_{l}\leq\frac{1}{5}$, we have $k_{1}(\ell,\ell,l_{l})\leq\frac{5\pi}{2}$,
with equality possible only if $l_{l}=\frac{1}{5}$. Now, using the
leaf as a main subgraph and the two loops as an attached subgraph,
we get 
\[
k_{1}(\ell,\ell,l_{l})\leq\frac{\pi}{2l_{l}\sqrt{3-l_{l}}}.
\]
In particular, we have $k_{1}(\ell,\ell,l_{l})<\frac{5\pi}{2}$ for
$0.26\leq l_{l}\leq1$. To cover the remaining values of $l_{l}$,
we construct the following test function. Take the function $x\mapsto\cos(\frac{\pi x}{l_{l}})$
on the leaf, so that it vanishes at $v$. On each petal, take the
function $x\mapsto\frac{l_{l}}{1-l_{l}}\sin(\frac{2\pi x}{1-l_{l}})$.
Denoting the resulting function by $h$, we have 
\[
\mathcal{R}(h)=\pi^{2}\frac{(1-l_{l})^{3}+16l_{l}^{3}}{4l_{l}^{2}(1-l_{l})^{2}}.
\]
In particular, we have $k_{1}(\ell,\ell,l_{l})\leq\frac{5\pi}{2}$
for $\frac{1}{5}\leq l_{l}\leq\frac{2}{5}$, with equality possible
only if $l_{l}=\frac{1}{5}$. Gathering the information given by these
three test functions, we conclude that for all $l_{l}$ values we
have $k_{1}(\ell,\ell,l_{l})\leq\frac{5\pi}{2}$, with equality possible
only if $l_{l}=\frac{1}{5}$. 

Moreover, it is easy to show that $k_{1}(\frac{2}{5},\frac{2}{5},\frac{1}{5})=\frac{5\pi}{2}$
with multiplicity two. This multiplicity increases to three when imposing
the Dirichlet condition at the central vertex, so that the equilateral
stower satisfies the strong Dirichlet criterion. It only remains to
show that if $l_{l}=\frac{1}{5}$ and $l_{1}\neq l_{2}$, we have
$k_{1}(l_{1},l_{2},l_{l})<\frac{5\pi}{2}$. This is obtained by applying
Corollary \ref{cor:extension_constant} to the two loops as the main
subgraph and the leaf as the attached subgraph. Thus, the equilateral
stower is a unique maximizer and satisfies in particular condition
(b) of Theorem \ref{thm:supremum_of_gluing}.
\end{proof}
~
\begin{proof}[Proof of Lemma \ref{lem:stower_3_1}]
Denote by $l_{1},l_{2},l_{3}$ and $l_{l}$ the lengths of the three
petals and the leaf. Assume without loss of generality that $l_{1}\geq l_{2}\geq l_{3}$
and define $\ell:=\frac{l_{1}+l_{2}+l_{3}}{3}$. Using the three petals
as a main subgraph and the leaf as an attached subgraph, we have $k_{1}(l_{1},l_{2},l_{3},l_{p})\leq\frac{\pi}{2\ell}$.
Moreover, equality is possible only if $l_{1}=l_{2}=l_{3}=\ell$.
Using the longest two petals and the leaf as a main subgraph and the
shortest petal as an attached subgraph we further have 
\[
k_{1}(l_{1},l_{2},l_{3},l_{l})\leq\frac{5\pi}{2(1-l_{3})}\leq\frac{5\pi}{2(1-\ell)}.
\]
Combining the two bounds we got on $k_{1}$, it follows that $k_{1}(l_{1},l_{2},l_{3},l_{p})\leq\frac{7\pi}{2}$,
with an equality possible only if $\ell=\frac{2}{7}$ and $l_{3}=\ell$.
These two equalities together entail that $l_{1}=l_{2}=l_{3}=\frac{2}{7}$
and $l_{l}=\frac{1}{7}$. With this choice of lengths, it is easy
to show that the spectral gap equals $\frac{7\pi}{2}$ and of multiplicity
three. This multiplicity increases to four when imposing the Dirichlet
condition at the central vertex, which means that the equilateral
stower satisfies the strong Dirichlet criterion. As the equilateral
stower is a unique supremizer, it also satisfies condition (b) of
Theorem \ref{thm:supremum_of_gluing}.
\end{proof}
~
\begin{proof}[Proof of Lemma \ref{lem:stower_1_1}]
 Let $\ell\in[0,1]$ be the length of the leaf and $1-\ell$ the
length of the petal. Using the leaf as a main subgraph and the petal
as an attached subgraph, we get 
\[
k_{1}(\ell,1-\ell)\leq\frac{2\pi}{2\ell\sqrt{3-2\ell}}.
\]
In particular, we have $k_{1}(\ell,1-\ell)\leq2\pi$ as long as $2\ell\sqrt{3-2\ell}\geq1.$
This is satisfied for $\ell\geq\frac{1}{3}$, and in this case the
inequality is strict. Next, we refer to the scattering approach described
in Appendix A and more precisely to equation \eqref{eq:secular_cond},
whose zeros are the graph's eigenvalues. This equation is equivalent,
in our case, to $F(k,\ell)=0$, where 
\begin{equation}
F(k,\ell):=2\cos(k\ell)\sin\left(k\frac{1-\ell}{2}\right)+\sin(k\ell)\cos\left(k\frac{1-\ell}{2}\right).
\end{equation}
Substituting $k=2\pi$, and using basic trigonometric identities,
we get 
\begin{align}
F(2\pi,\ell) & =2\cos(2\pi\ell)\sin\left(\pi(1-\ell)\right)+\sin(2\pi\ell)\cos\left(\pi(1-\ell)\right)\\
 & =2\sin\left(\pi\ell\right)\left(\cos(2\pi\ell)-\cos^{2}\left(\pi\ell\right)\right)=2\sin\left(\pi\ell\right)\left(\cos^{2}\left(\pi\ell\right)-1\right).
\end{align}
We notice that $F(k,\ell)>0$ for small positive values of $k$ and
that $F(2\pi,\ell)<0$ for $\ell\in\left(0,\frac{1}{3}\right]$. As
$F$ is continuous in $k$, we deduce that there exists some $k<2\pi$
such that $F(k,\ell)=0$. This means that for $\ell\in\left(0,\frac{1}{3}\right]$,
the spectral gap is strictly below $2\pi$. As we have seen above
that this is also the case for $\ell>\frac{1}{3}$ and since the spectral
gap is $2\pi$ for $\ell=0$ (single cycle graph), the result follows.
\end{proof}
\bibliographystyle{plain}
\bibliography{SpectralGap}

\end{document}